\documentclass[11pt]{article}

\usepackage[margin=1in]{geometry}
\usepackage{amsmath, amssymb, amsthm}
\usepackage{graphicx}
\usepackage{hyperref}
\usepackage{cleveref}
\usepackage{enumitem} 
\usepackage{xcolor} 
\usepackage{mathtools}
\usepackage{placeins}

\usepackage[sort&compress]{natbib}
\bibpunct[, ]{[}{]}{,}{n}{}{,}%

\let\cite\citet

\definecolor{myblue}{HTML}{1E6FFF} 

\newif\ifshowchanges
\showchangestrue        

\ifshowchanges
  
  \newenvironment{blueblock}{\begingroup\color{myblue}}{\endgroup}
\else

\fi

\theoremstyle{plain}
\newtheorem{thm}{Theorem}[section]

\newtheorem{lem}[thm]{Lemma}
\newtheorem{cor}[thm]{Corollary}

\theoremstyle{remark}
\newtheorem*{rem*}{Remark}

\newtheoremstyle{boldtitle}%
  {}{}
  {\itshape}
  {}
  {\bfseries}
  {}
  {\newline}
  {}

\theoremstyle{boldtitle}
\newtheorem*{innerfixedthm}{\thmtitle}

\usepackage{enumitem}

\newlist{assumplist}{enumerate}{1}
\setlist[assumplist]{%
  label=\textup{(A\arabic*)}, 
  ref=A\arabic*,
  before=\itshape,           
  after=\normalfont,         
  leftmargin=*,
  align=left,
  labelsep=.5em,
  topsep=.6\baselineskip,
  itemsep=.25\baselineskip,
  parsep=0pt,
  partopsep=0pt
}

\newlist{assumplistB}{enumerate}{1}
\setlist[assumplistB]{%
  label=\textup{(B\arabic*)}, 
  ref=B\arabic*,
  before=\itshape,           
  after=\normalfont,         
  leftmargin=*,
  align=left,
  labelsep=.5em,
  topsep=.6\baselineskip,
  itemsep=.25\baselineskip,
  parsep=0pt,
  partopsep=0pt
}

\newtheorem*{innerfixedlem}{\lemtitle}

\newtheoremstyle{fullbold}%
  {3pt}   
  {3pt}   
  {\normalfont} 
  {}      
  {\bfseries} 
  {.}     
  {0.5em} 
  {\thmname{#1}~\thmnumber{#2}} 

\newcommand{\ttheta}{\ensuremath{\tilde{\theta}^{*}}}  
\newcommand{\rtheta}{\ensuremath{\theta^{*}}}          
\newcommand{\hb}[1][u]{\ensuremath{h_{\beta}(#1)+\tilde{\theta}^{*}}} 

\newcommand{\KL}{\text{KL}}

\newtheoremstyle{nonbolditalichead}
  {\topsep}    
  {\topsep}    
  {\normalfont}  
  {}           
  {\itshape}   
  {.}          
  { }          
  {\thmname{#1}\thmnumber{ #2}} 

\theoremstyle{nonbolditalichead}
\newtheorem*{remark}{Remark}

\title{\Large \textbf{Variance Reduction Methods for Dirichlet Expectations}}
\author{
Ayeong Lee \\[6pt]
Graduate School of Business \\ 
Columbia University, New York, USA \\ 
ayeong.lee@columbia.edu
}
\date{\today}

\begin{document}
\maketitle

\begin{abstract}

Dirichlet distributions are probability measures on the unit simplex. They are often used as prior distributions in modeling categorical data, such as in topic analysis of text data. Motivated by this application, we consider Monte Carlo estimation of expectations $\mathbb{E}[\exp(nH(\theta))]$, where $\theta$ has a Dirichlet distribution, $H$ is a real-valued function, and $n$ is a parameter. We develop variance reduction techniques particularly designed to work well for large $n$. Our analysis is guided by the Laplace method for approximating integrals, which we extend to fit our problem setting. We develop an importance sampling method that achieves a near-optimal asymptotic relative error. We use related ideas to select a provably effective control variate. We illustrate these results through their application in topic analysis.

\end{abstract}


\bigskip
\hrule
\bigskip

\section{Introduction}

This paper proposes and analyzes variance reduction techniques for Monte Carlo estimation of expectations with respect to Dirichlet distributions. Dirichlet distributions are probability measures on the unit simplex and can therefore be interpreted as distributions over random probability vectors. They are often used in Bayesian settings to model prior distributions over such vectors.
The Dirichlet family can generate a wide variety of shapes, including distributions concentrated near vertices or faces of the simplex and the uniform distribution over the simplex, and thus offers a great deal of modeling flexibility.

We focus on expectations of the form $\mathbb{E}[\exp(nH(\theta))]$, in which 
$\theta$ has a Dirichlet distribution, $H$ is a real-valued function on the simplex, and $n\in\mathbb{N}$. A plain Monte Carlo estimator would average values of $\exp(nH(\theta))$ over independent draws of $\theta$. For large $n$, the standard deviation of this estimator will typically be large relative to its mean. We develop variance reduction techniques 
--- an importance sampling method and a control variate method ---
specifically designed to be effective for large $n$.

Our investigation is motivated in part by an application to Latent Dirichlet Allocation (LDA), introduced in \cite{blei2003latent}. LDA is a widely used Bayesian model for discovering topics in a collection of documents. The topics are latent in the sense that they are not directly observed or imposed; they are inferred from the co-occurrence of words in documents. A topic is represented by a probability distribution over a vocabulary with a Dirichlet prior. Each document is represented by a probability distribution over the set of topics, also with a Dirichlet prior.
Evaluating a topic model entails evaluating the \emph{model evidence} on held-out documents. This takes the form $\mathbb{E}[\exp(nH(\theta))]$, with $H$ a log-likelihood, $\theta$ a Dirichlet vector, and $n$ representing the number of words in the document, which is often large.

Our analysis of the large-$n$ setting is guided by the Laplace method for approximating integrals. In its classical form, the Laplace method yields
an asymptotic relation of the form
\begin{equation}
I(n):=\int_D e^{nH(\theta)} g(\theta)d\theta
\sim C n^{-k} e^{nH(\theta^*)},
\label{claplace}
\end{equation}
where $C>0$ is a constant, $\theta^*$ is the global maximizer of $H$ on $D \subset \mathbb{R}^d$, 
$H$ is concave near $\theta^*$,
and the relation $\sim$ indicates that the ratio of the two sides approaches one as $n\to\infty$. In the simplest case, $k=d/2$.
We defer precise conditions to later sections, but we note that our setting will require two extensions of (\ref{claplace}):
one to handle the case where $\theta^*$ lies on a face of the simplex (changing the value of $k$), and a further extension to allow terms with two different powers of $n$ in the exponent. These extensions are of independent interest beyond our specific application.

The approximation (\ref{claplace}) indicates that the value of the integral is mainly determined by values of $H$ near the maximizer $\theta^*$. 
Applied to $\mathbb{E}[\exp(nH(\theta))]$, this suggests that we may be able to reduce variance by concentrating more of our sampling near $\theta^*$.
This insight drives our importance sampling (IS) strategy. Rather than sample $\theta$ from its original distribution, we sample it from a different Dirichlet distribution that shifts more mass close to $\theta^*$; we correct for the change of probability measure by weighting each sample by the ratio of the original and new Dirichlet densities --- the usual likelihood ratio for importance sampling.

We let the IS distribution depend on the parameter $n$. If the original Dirichlet distribution has (vector) parameter $\alpha$, our IS distribution has parameter $\alpha + n^{\gamma}\theta^*$, for any $\gamma\in(0,1)$. This choice makes the IS distribution increasingly concentrated near $\theta^*$ as $n$ increases.
We analyze the performance of our estimator through an extension of the classical Laplace method. We say that an estimator of $I(n)$ has \emph{bounded relative error} if the ratio of its root mean squared error to $I(n)$ remains bounded as $n$ grows; this is essentially the best possible performance of any Monte Carlo estimator.
We show that, under appropriate conditions on $H$, our estimator has relative error that is $\Theta(n^{(1-\gamma)c})$, with the constant $c>0$ explicitly given, while the plain MC estimator has relative error of $\Theta(n^c)$. Thus, the performance of IS can be brought arbitrarily close to bounded relative error by choosing $\gamma<1$ close to 1. The need for the restriction to $\gamma<1$ will become clear from our extension of the Laplace method to handle the second moment of the IS estimator.
The technical conditions needed for these results are all satisfied in the LDA application. 

The likelihood ratio for our IS estimator involves the Kullback-Leibler (KL) divergence between pairs of vectors on the unit simplex. This property drives much of the analysis of our IS estimator. 
It also suggests a convenient KL-based control variate (CV). CV and IS estimators each have advantages. Whereas our IS estimator achieves near-optimal performance for large $n$, our CV estimator guarantees variance reduction for all $n$.

The degree of variance reduction achieved through any CV is determined by the squared correlation between the CV and the quantity of interest. Through the Laplace method, we derive an explicit expression for the limiting squared correlation. This expression is close to 1 (indicating large variance reduction) when the Hessians of $H(\theta^*)$ and the KL divergence are close. In the LDA setting, we show that this closeness condition aligns well with the sparsity structure of LDA: sparsity results from the fact that different topics put most of their weight on different sets of words. We prove a lower bound on the degree of variance reduction achieved based on a measure of the model's sparsity.

Laplace approximations have been used often in Bayesian statistics; see, for example,
\cite{tierney1986accurate}, \cite{kass1995bayes},  and the many references in 
\cite{kasprzak2025good}.
We are not aware of their use with LDA, so our Theorem \ref{thm:boundary-laplace_standard} may be useful as an approximation quite apart from providing a tool for analyzing Monte Carlo estimators. \cite{hennig2012kernel} use a Laplace approximation in developing an alternative to LDA based on Gaussian processes, which has a very different structure. \cite{wallach2009evaluation} run computational comparisons of various simulation methods for LDA evaluation metrics. 
They mainly focus on methods that apply Gibbs sampling to make topic assignments on held-out documents and are thus not directly comparable to our setting. They do not provide any theoretical analysis of the methods they test.

There is an extensive literature applying large deviations techniques to design IS procedures for rare-event simulation; see, among many others,
\cite{siegmund1976importance},
\cite{sadowsky2002large},
\cite{dupuis2004importance},
\cite{asmussen2007stochastic},
\cite{blanchet2012state}, and references there.
Like many rare-event estimators, our IS method applies an exponential change of measure, exploiting the exponential-family structure of Dirichlet distributions. But our setting differs from the rare-event simulation literature in several respects. The first, of course, is that we estimate expectations rather than probabilities, but large deviations ideas have also been used for expectations in, e.g., 
\cite{ghs},
\cite{guasoni2008optimal},
and
\cite{setayeshgar2026importance}.
We discuss two more important features: (1) choice of domain and
(2) sub-exponential convergence rates.

\emph{(1) Choice of domain.} The fact that we work on the unit simplex has important implications for our analysis, particularly when the maximizing $\theta^*$ falls on a face of the simplex. This often happens in LDA, where $\theta^*_i=0$ means that topic $i$ is not present in a held-out document, under the maximum likelihood 
topic distribution for that document. We need to prove an extension of the Laplace method to handle this case. In this extension  (Theorem \ref{thm:boundary-laplace_standard}), the power $k$ in (\ref{claplace}) depends on the number of coordinates of $\theta^*$ equal to zero and on the corresponding Dirichlet parameters for these coordinates. 
We have not seen similar asymptotic behavior in other settings. Moreover, a Dirichlet density can take the values zero and infinity on a face, so our IS estimator truncates certain faces. This introduces a bias, which we show is negligible for large $n$ and does not affect the rate of decrease of the estimator's mean squared error.

\emph{(2) Sub-exponential rates.} Large deviations analysis is ordinarily concerned with measuring exponential rates of decay. But this perspective is too coarse for our setting. Indeed, even using plain Monte Carlo, we see in (\ref{claplace}) that the second moment $I(2n)$ achieves the same (optimal) exponential rate as $I(n)^2$. Achieving asymptotic improvements relies on improving the polynomial factor, which is why understanding how $\theta^*$ determines $k$ is fundamental to our IS results. This point is also relevant to a distinction in terminology:
the Laplace \emph{principle}, as the term is used in large deviations theory 
(especially \cite{dupuis2011weak}), refers to identifying the exponential rate in expressions like (\ref{claplace}), whereas the Laplace \emph{method} or \emph{approximation} refers to the full asymptotics in (\ref{claplace}).

Some preliminary results on IS for LDA were reported in \cite{glee}. This paper goes beyond that one in three important respects.
\begin{itemize}
\item \cite{glee} considered only the ``interior'' case in which all coordinates of $\theta^*$ are strictly positive. A focus of this paper is handling the boundary case, which requires proving a new version of the Laplace method (Theorem \ref{thm:boundary-laplace_standard}).
\item They considered only the case $\gamma = 1/2$ for IS, which is easier to analyze but fails to achieve near-optimality. This paper covers all $\gamma\in(0,1)$, which requires a further extension of the Laplace method (Theorem~\ref{thm:boundary-laplace-with-kl}).
\item They did not consider control variates.
\end{itemize}

The rest of the paper is organized as follows: Section \ref{sec:preliminaries} formulates our problem precisely. Section \ref{sec: Laplace Method} reviews the classical Laplace approximation and states our extension of the approximation for expectations on the unit simplex, with particular attention to the case of a boundary optimizer. Section \ref{sec:importance_sampling} develops our IS estimator and analyzes its asymptotic error reduction.
Section \ref{sec:control_variate}  introduces and analyzes our control variate estimator. Section \ref{sec:numerical} reports numerical experiments, and Section \ref{sec:conclusion} concludes. Proofs are deferred to the appendix, and some additional details are provided in the Supplementary Material.

\section{Preliminaries}
\label{sec:preliminaries}
\subsection{Definitions}
\label{subsec: definitions}

The $(K-1)$-simplex is defined as
\begin{equation}
\label{eq:simplex}
\Delta_{K-1}
=
\left\{
\theta \in \mathbb{R}^K
\,\middle|\,
\mathbf{1}^\top \theta = 1,\;
\theta_i \ge 0,\ i=1,\dots,K
\right\}
\end{equation}
where $\mathbf{1}$ is the column vector of ones in $\mathbb{R}^{K}$. The simplex $\Delta_{K-1}$ is a compact convex polytope of dimension $K-1$.
Each $\theta \in \Delta_{K-1}$ can be identified with a discrete
probability vector on a $K$-point space.
For a parameter vector $\alpha=(\alpha_1,\dots,\alpha_K)\in\mathbb{R}_{++}^K$,
the Dirichlet distribution with parameter $\alpha$ is a probability measure
on $\Delta_{K-1}$.
The density is commonly written as
\begin{equation}
\label{eq:Dirichlet}
\mathrm{Dir}_\alpha(\theta)
=
\frac{1}{B(\alpha)}
\prod_{i=1}^K \theta_i^{\alpha_i-1},
\qquad \theta \in \Delta_{K-1},
\end{equation}
where
\(
B(\alpha)
=
{\prod_{i=1}^K \Gamma(\alpha_i)}/
{\Gamma(\sum_{i=1}^K \alpha_i)}
\)
is the multivariate Beta function and $\Gamma(\alpha_i)$ the Gamma function. 
The parameters $\alpha$ dictate the concentration of the Dirichlet distribution. If $\alpha_i < 1$, the density is singular along the face $\{\theta_i = 0\}$ and mass is concentrated near that face. If $\alpha_i > 1$, the density vanishes at $\{\theta_i = 0\}$,
and mass is concentrated away from that face and more towards the interior. If $\alpha_i=1$ for all $i$, then it becomes a uniform distribution on the simplex.

For analytical convenience, we also work with a coordinate representation of
the Dirichlet density.
Because \(
\theta_K = 1-\sum_{i=1}^{K-1}\theta_i
\), a point in $\Delta_{K-1}$ is determined by its first $K-1$ components, so
each point in the simplex $\Delta_{K-1}$ is identified with a point in  the projected simplex
\begin{equation}
\label{eq:projected_simplex}
\tilde{\Delta}_{K-1}
=
\left\{
y \in \mathbb{R}^{K-1}
:
y_i \ge 0,\;
\sum_{i=1}^{K-1} y_i \le 1
\right\}.
\end{equation}
In these coordinates, the Dirichlet distribution has the Lebesgue density
\begin{equation}
\label{eq:Dirichlet_alternative}
\mathrm{Dir}_{\alpha}^{(K-1)}(y)
=
\frac{1}{B(\alpha)}
\prod_{i=1}^{K-1}y_i^{\alpha_i-1}
\Bigl(1-\sum_{i=1}^{K-1}y_i\Bigr)^{\alpha_K-1},
\qquad
y \in \tilde{\Delta}_{K-1}.
\end{equation}

\subsection{Problem Formulation}
\label{subsec: the problem}

We want to estimate the quantity
\begin{equation}
\label{eq:expectation}
I(n):=\mathbb{E}_{\text{Dir}_\alpha} \left[ e^{nH({\theta})}\right],
\end{equation} where $n\in \mathbb{R}^+$ is a scaling parameter and $H:\mathbb{R}^K \to \mathbb{R}\cup\{-\infty\}$ is a function satisfying the following conditions. 

\par\vspace*{-\parskip}
\begin{assumplist} 
 
    \item (Unique Maximum) $H$ achieves its global maximum on $\Delta_{K-1}$ only at $\theta^*$.

\item (Differentiability) $H$ is continuous on $\Delta_{K-1}$ and twice continuously differentiable in an open neighborhood of $\theta^*$ in $\mathbb{R}^K$.

\end{assumplist}

We could drop the condition of continuity on $\Delta_{K-1}$ in (A2) if we strengthened (A1)
to require that for every open neighborhood $U \subset \mathbb{R}^K$ of $\theta^*$,
\begin{equation}
    \label{eq:continuity_relaxation}
    \sup_{\theta \in \Delta_{K-1}\setminus U} H(\theta) < H(\theta^*).
\end{equation}
Condition \eqref{eq:continuity_relaxation} is sufficient for all subsequent proofs in the paper. Continuity on the compact set $\Delta_{K-1}$ in (A2) and uniqueness of the maximizer (A1) implies \eqref{eq:continuity_relaxation}.

The unique maximizer in (A1) may lie at the boundary of the simplex. Throughout the paper, the \emph{boundary} and \emph{interior} are understood relative to the affine hyperplane
$\left\{ \theta \in \mathbb{R}^K : \mathbf{1}^\top \theta = 1 \right\}.$ Thus, a point $\theta \in \Delta_{K-1}$ is an interior point if and only if all of its components are strictly positive. In contrast, topological notions such as \emph{continuity}, \emph{differentiability}, and \emph{compactness}, e.g., in (A2), are defined with respect to the ambient Euclidean space.

Let $m$ denote the number of zero components of $\theta^*$. 
Without loss of generality, relabel the coordinates so that these correspond to the first $m$ entries:

\begin{equation}
\label{eq:ordering_assumption}
    \theta_1^* = \cdots = \theta_m^* = 0, \qquad \theta_{m+1}^*, \dots, \theta_K^* > 0.
\end{equation}

The problem of maximizing $H$ over the simplex satisfies linear independence constraint qualification (LICQ). Therefore at the optimal point $\theta^*$, there exist Karush–Kuhn–Tucker (KKT) multipliers $\lambda_i\ge 0$ for the inequality constraints and $\mu\in\mathbb{R}$ for the equality constraint such that
 \begin{equation}
 \label{eq:KKT_conditions}
 \nabla H(\theta^*) = -\sum_{i =1}^{m} \lambda_i e_i + \mu\,\mathbf{1}_K.
 \end{equation}

Further assume the following.

\begin{assumplist}[resume]
\item (Strict complementarity) The KKT multipliers corresponding to the active inequality constraints are strictly positive: $\lambda_i > 0$ for all $i=1,\dots,m$.

\item (Negative Definiteness of the Hessian) The Hessian of $H$ at $\theta^*$ is negative definite on the critical cone, i.e.
\begin{equation}
\label{eq:reduced_hessian_condition}
    d^\top \nabla^2 H(\theta^*) d < 0,
\quad \forall d \in \mathcal{C}(\theta^*) \setminus \{0\}
\end{equation}
where
\begin{equation}
\label{eq:critical_cone}
\mathcal{C}(\theta^*)
=\Bigl\{
d\in \mathbb{R}^K:\ \mathbf{1}_K^\top d = 0,\ d_i = 0, \quad i\le m
\Bigr\}.
\end{equation}

\end{assumplist}

Underlying the Laplace method is a quadratic approximation to $H$.
The critical cone is the set of feasible directions in the simplex along which the first-order directional derivatives of $H$ vanish.
Along such directions, curvature determines local maximality, whereas outside the critical cone the first-order term already precludes any increase  (cf. \cite{nocedal2006numerical}, Section 12.5). 
Consequently, in the Laplace method only directions in the critical cone require a second-order expansion of $H$. 
We note that a sufficient condition for (A4) is negative definiteness of the Hessian of $H$ on $\mathbb{R}^K$, which we denote as $\nabla^2 H(\theta^*)\prec0$.

Under assumptions (A1)--(A4), we will see that the Laplace method implies that the biggest contribution to the expectation in \eqref{eq:expectation} comes from points near $\theta^*$.  
We will exploit this insight to improve upon the standard MC estimator.

We next discuss the LDA as an example of problem \eqref{eq:expectation} where (A1)--(A4) are satisfied.

\subsubsection{Example: Latent Dirichlet Allocation (LDA)} 
\label{sec:LDA}

Latent Dirichlet Allocation (LDA), introduced in \cite{blei2003latent}, is a method for discovering latent topics in a collection of documents; topics provide a low-dimensional representation of a large set of words drawn from a large vocabulary. Let $K$ denote the number of topics, which is a hyperparameter of the model. 
A topic is represented by a probability distribution over a fixed vocabulary of size $V$; individual words are more likely under some topics than others.
Thus, the topics are given by vectors $\phi_k\in \Delta_{V-1}$, $k=1,\dots,K$. Each document is represented by a probability distribution $\theta \in \Delta_{K-1}$, in which $\theta_k$ represents the weight of topic $k$ in the document. LDA takes a Bayesian perspective to infer the topic vectors $\phi_k$ and the document vectors $\theta$, imposing Dirichlet priors on both.

Let $\phi = \{\phi_{k}\}_{k=1}^{K}\in \mathbb{R}^{K\times V}$ be the collection of topic vectors. An important task is to evaluate the quality of the topics $\phi$ extracted by LDA. This is often done by considering the expected likelihood of a held-out document with words $w$; see, for example, \cite{wallach2009evaluation}, especially equations (6) and (12).

Let $n$ be the total number of words in the document, and let $p_{v}=n_{v}/n$ denote the frequency of each vocabulary element $v$.  The expected likelihood of the document is averaged over topic proportions $\theta$ drawn from the Dirichlet prior with parameter $\alpha$,
\begin{equation}
\label{eq:likelihood}
p(w|\phi)=\mathbb{E}_{\text{Dir}_{\alpha}}[e^{nH(\theta)}], 
\end{equation}
where $H$ is the log-likelihood
\begin{equation} \label{eq:entropy}
H(\theta)=\sum_{v=1}^{V}p_{v}\log\left(\sum_{k=1}^{K}\theta_k \phi_{k}(v)\right).
\end{equation}
The expression $\sum_k\theta_k\phi_k(v)$ is the probability of word $v$ as represented by the topic model: a document picks topic $k$ with probability $\theta_k$, and then topic $k$ picks the word $v$ with probability $\phi_k(v)$. The expression in (\ref{eq:likelihood}) compares this model-based probability with the empirical frequency $p_v$. This expression, which is always negative, will be closer to zero when the two distributions are more closely aligned.

We note the expressions for the gradient
\begin{align}
\label{eq:gradient}
\nabla_{\theta}H(\theta)=\sum_{v}p_{v}\frac{\phi(v)}{\theta^{\top}\phi(v)}\end{align}
and Hessian of $H$
\begin{equation}
\nabla^{2}H(\theta) = -\sum_{v=1}^{V} p_{v} \cdot \frac{\phi(v) \phi(v)^{\top}}{(\theta^{\top} \phi(v))^{2}}.
\label{eq:hess}
\end{equation}

They are all well defined since $\theta^\top \phi(v)>0$ for all $v$. This follows from the fact that the topic-word distributions $\phi_k$ and the document-topic distributions $\theta$ are sampled from Dirichlet distributions and therefore have strictly positive components with probability 1.

Since $\Delta_{K-1}$ is a compact subset of $\mathbb{R}^K$ and $H$ is continuous on $\mathbb{R}^K$, we have the existence of a maximizer $\theta^*$. For uniqueness, note that $\nabla^{2}H$ is negative definite as long as the set of topic vectors $\phi$ has full rank.
This holds in practice because the size $V$ of the vocabulary is ordinarily much larger than the number of topics $K$.

It follows from the gradient expression in \eqref{eq:gradient} that $\theta^\top \nabla H(\theta)=1$. This implies that
the KKT multiplier $\mu$ in \eqref{eq:KKT_conditions} equals 1, which then implies that 
\begin{equation}
\label{eq:LDA_gradient}
    \nabla H(\theta^{*})_{i}=\begin{cases}
1, & \theta_{i}^{*}>0;\\
\le1, & \theta_{i}^{*}=0.
\end{cases}
\end{equation}

The optimizer $\theta^*$ can be calculated very efficiently using the simple recursive scheme analyzed in a different setting by \cite{cover1984algorithm}. A brief discussion of the algorithm is in the Supplemental Material \ref{sec:cover_algorithm}.

A model very similar to LDA, due to \cite{pritchard2000inference}, is widely used in population genetics. In that setting, each $v$ represents an allele,
the $\phi_k$ vectors represent allele frequencies in different populations, and $\theta_k$ measures the fraction of an individual's genome that originates from population $k$. Our results apply in that setting as well.
More generally, LDA is used for dimension reduction with categorical data in many applications outside of text analysis --- for example, in modeling survey responses (\cite{munro2022latent}) or CEO time usage (\cite{bandiera2020ceo}).

\section{Laplace Method on the Simplex}
\label{sec: Laplace Method}

To improve the estimation of $I(n)$, it will be useful to understand the behavior of $I(n)$ for large $n$, where the plain MC estimator has large relative error. In the LDA setting, large $n$ corresponds to evaluating the model on a large document, which is the most relevant case. The Laplace method is well-suited to our setting. We first discuss the method in a classical setting and then develop the necessary extension for integrals with respect to Dirichlet distribution on the simplex.

\subsection{Classical Laplace Method}
In its basic form, the Laplace method (as in, e.g., \cite{breitung1994asymptotic}, Theorem 41, p.56) states that integrals with exponential dependence on a large parameter $n$ should follow 

\begin{equation}
\label{eq:Laplace_approximation_simple}
\int_F f(\textbf{x}) \exp(nh(\textbf{x}))d\textbf{x}\sim
(2\pi)^{K/2} \frac{f(\textbf{x}^*)}{\sqrt{|\det(\nabla^2 h(\textbf{x}^*)|}}\cdot \exp(nh(\textbf{x}^*))\cdot n^{-K/2},
\quad \mbox{as $n\to\infty$,}
\end{equation}
where $F$ is a closed domain in $\mathbb{R}^K$ and $\textbf{x}^*$ is an interior maximizer of a twice continuously differentiable function $h$ on $F$. Here $f(\cdot)$  is a continuous function assumed to be neither vanishing nor singular at $\mathbf{x}^*$, and the Hessian of $h(\cdot)$, $\nabla^2 h (\textbf{x}^*)$, is negative definite. 
In the context of our problem $I(n)$ in \eqref{eq:expectation}, $f$ can be thought of as the Dirichlet density $\text{Dir}_{\alpha}^{(K-1)}$ and $h$ as the function $H$, both viewed as functions over the projected simplex defined in \eqref{eq:projected_simplex}. 
The biggest contribution to the integral in (\ref{eq:Laplace_approximation_simple}), and therefore to $I(n)$ in \eqref{eq:expectation}, should come from the neighborhood of the maximizer $\theta^*$.

An underlying assumption in \eqref{eq:Laplace_approximation_simple} is that $\textbf{x}^*$ is an interior point of the domain. However in LDA, maximum likelihood topic proportions $\theta^*$ often have zero components. 
There are well-known extensions of the Laplace method to settings where the maximum is attained at the boundary or a surface part of the domain (see Chapter 5 of \cite{breitung1994asymptotic}). However, a key requirement of these results is that $f$ be continuous and neither singular nor vanishing at the maximizer $\textbf{x}^*$.
In contrast, in our problem, the underlying Dirichlet density always vanishes or is singular at the boundary. 
Hence, a new asymptotic result is required.

Note that the right side of \eqref{eq:Laplace_approximation_simple} does not explicitly depend on the integration domain $F$. We would get the same limit with a smaller $F$ so long as $\textbf{x}^*$ remains in the interior of the domain. Similar ideas hold when $\textbf{x}^*$ is at the boundary of $F$ with appropriate adjustments to the asymptotics in
\eqref{eq:Laplace_approximation_simple}. 
The key idea is to utilize the strict gap property \eqref{eq:continuity_relaxation}, which ensures that the contribution to the integral from outside any neighborhood of the maximizer is asymptotically negligible.
This point is formalized through the localization lemma in Supplemental Material \ref{sec:localization_lemma}. We will use this flexibility in the choice of domain at several points in our analysis.

\subsection{Laplace Method with Dirichlet Distribution}
\label{subsec: Laplace_method_on_simplex}

We extend the Laplace method to a setting where the underlying function $f$ is the Dirichlet density over the simplex. 
A subtlety here is that the maximizer $\theta^*$ may lie at the boundary on which the Dirichlet density $\text{Dir}_{\alpha}(\theta^*)$ is either singular or zero. While there exist variants of the Laplace method in singular and boundary cases separately, our setting requires us to establish a new variant that combines both. 
The main difficulty is that the boundary geometry and the singular behavior of the Dirichlet density must be handled simultaneously.

When the maximizer is an interior point, as in the standard Laplace method in \eqref{eq:Laplace_approximation_simple}, the asymptotic polynomial factor depends only on the dimension of the domain (i.e., $K-1$). However when it is a boundary point, the integral picks up the behavior of the underlying Dirichlet density along the components for which $\theta^*$ is zero. The polynomial factor depends on the number of zero components of $\theta^*$ and the corresponding Dirichlet parameters $\alpha_i$, which determines the level of singularity or decay of $\text{Dir}_\alpha$ at the boundary.  
The product form of the Dirichlet density allows each active coordinate to contribute independently to the boundary behavior.

\begin{thm}[Laplace Method on the Simplex]
\label{thm:boundary-laplace_standard}
Suppose $H:\mathbb{R}^{K}\to\mathbb{R}$ satisfies (A1)-(A4) at the maximizer $\theta^* \in \Delta_{K-1}$. Let $m$ denote the number of zero components of $\theta^*$.  
Then as $n\to\infty$
\begin{equation}
\label{eq:Laplace_boundary_integral} 
\mathbb{E}_{\text{Dir}_\alpha} \left[ e^{nH({\theta})}\right]\sim C_{H}\cdot \exp({nH(\theta^*)})\cdot n^{-\frac{(K-1-m)}{2}}\cdot n^{-\sum_{i=1}^{m}\alpha_{i}}
\end{equation}
where $C_{H}\in(0,\infty)$ is a constant. 
\end{thm}
\begin{proof}
See proof in Section~\ref{sec:appendix-proof-boundary-laplace_standard}. An explicit expression for the constant $C_H$ is given in \eqref{eq:C_H_constant}.
\end{proof}

Theorem \ref{thm:boundary-laplace_standard} will be crucial in analyzing the asymptotic efficiency of alternative simulation estimators. 
For the second moment of each replication of the plain MC estimator, we can replace $n$ with $2n$ in (\ref{eq:Laplace_boundary_integral}) to get
\begin{equation}
\label{eq:Laplace_boundary_integral_second_moment}
\mathbb{E}_{\text{Dir}_\alpha} \left[ e^{2nH({\theta})}\right]\sim C_{H}\cdot \exp({2nH(\theta^*)})\cdot (2n)^{-\frac{(K-1-m)}{2}}(2n)^{-\sum_{i = 1}^{m}\alpha_{i}}.
\end{equation}
From this, we can already draw two conclusions. First, the square of the first moment in (\ref{eq:Laplace_boundary_integral}) is asymptotically negligible relative to the second moment in (\ref{eq:Laplace_boundary_integral_second_moment}), so the variance of the plain MC estimator is dominated by the second moment. Second, we see that the second moment has twice the exponential rate of the first moment;
for $H(\theta^*)<0$, this is the fastest possible rate of decay.
It follows that asymptotic improvements in efficiency need to reduce the polynomial terms in (\ref{eq:Laplace_boundary_integral_second_moment}). This distinguishes our setting from most applications of large deviations theory to simulation, which focus on changing the exponential rate of decay. Indeed, the Laplace {\it principle} considers only the exponential rate (\cite{dupuis2011weak} chapter 1), whereas the Laplace {\it method} captures polynomial terms as well.  

Theorem~\ref{thm:boundary-laplace_standard} is of broader interest beyond our specific application. It identifies a setting in which parameters of a prior distribution (the $\alpha_i$ from the Dirichlet distribution) influence the asymptotics of the model evidence as $n\to\infty$. This contrasts with commonly used approximations to the model evidence in Bayesian inference, such as the Bayesian Information Criterion (BIC) (cf. \cite{schwarz1978estimating}, \cite{kass1995bayes}).
BIC arises as an application of the Laplace method and relies on the assumption of an interior mode of the likelihood. This results in a leading-order behavior that is independent of the prior. While modified approximations have been developed for boundary modes in limited BIC settings (cf. \cite{hsiao1997approximate}), prior parameters still vanish from the leading-order terms. The distinction in our setting results from the fact that Theorem~\ref{thm:boundary-laplace_standard} does not assume an interior maximizer and the underlying measure is Dirichlet whose parameters affect the level of singularity or decay.

\section{Importance Sampling Estimator}
\label{sec:importance_sampling}

\subsection{Dirichlet Importance Sampling}

The standard Monte Carlo estimator of $\mathbb{E}_{\text{Dir}_\alpha}[e^{nH(\theta)}]$ is 
\begin{equation}
\label{eq:standard_MC}
\hat{p}_{\text{MC}}=\frac{1}{N}\sum_{i=1}^{N}e^{nH({\theta}^{(i)})}, \quad {\theta}^{(i)}\stackrel{\text{iid}}{\sim} {\text{Dir}_\alpha}.
\end{equation} 
For methods of sampling from the Dirichlet distribution, see, e.g., Section~4.3 of \cite{kroese2013handbook}.

The main insight from Theorem \ref{thm:boundary-laplace_standard} is that, for large $n$, $I(n)$ in \eqref{eq:expectation} is mainly determined by points close to $\theta^*$ and that $\hat{p}_{\text{MC}}$ is not efficient. This suggests that we may be able to improve simulation efficiency through an importance sampling procedure that gives more weight to the region close to $\theta^*$.
Given that the Dirichlet distributions form an exponential family supported on the simplex, it is natural to consider
sampling from another Dirichlet distribution. Importance sampling within exponential families has proved effective in other settings, and we will see that it achieves near-optimal performance in our setting when properly designed. 

For any distribution $\text{Dir}_{\eta}$ with parameter $\eta$ we have
\begin{align}
\label{eq:importance_sampling_change_of_measure}
\mathbb{E}_{\text{Dir}_{\alpha}}[e^{nH({\theta})}] 
&= \mathbb{E}_{\text{Dir}_{\eta}}\left[
e^{nH({\theta})} \cdot \frac{\text{Dir}_{\alpha}}{\text{Dir}_{\eta}}({\theta})
\right] \nonumber \\
&= \mathbb{E}_{\text{Dir}_{\eta}}\left[
e^{nH({\theta})} \cdot \frac{B(\eta)}{B(\alpha)} 
\prod_{j=1}^K \theta_j^{\alpha_j - \eta_j}
\right].
\end{align}
This representation of the expectation leads to an importance sampling estimator of the following:
\begin{equation}
\label{eq:generic_IS_estimator}
\hat{p}_{\text{IS}}=\frac{1} {N}\sum_{i=1}^{N} \frac{\text{Dir}_{\alpha}(\theta^{(i)})}  {\text{Dir}_{\eta}(\theta^{(i)})}e^{nH(\theta^{(i)})}, \quad \theta^{(i)}\stackrel{\text{iid}}{\sim} \text{Dir}_{\eta}.
\end{equation}
For any choice of $\eta$, the IS estimator $\hat{p}_{\text{IS}}$ provides an unbiased estimator. 

To gain insight into an effective choice of $\eta$, we apply the Laplace method to the second moment of the IS estimator.  For simplicity we begin by considering the case where $\theta^*$ is in the interior.
For any fixed $\eta$, the Laplace method gives that, for large $n$,
\begin{equation}
\label{eq:secondmoment_IS}
\mathbb{E}_{\text{Dir}_{\eta}} \left[ \left( \frac{\text{Dir}_{\alpha}(\theta)}  {\text{Dir}_{\eta}(\theta)} \right)^2 e^{2nH(\theta)} \right] \propto   \frac{\text{Dir}_{\alpha}(\theta^*)^2}  {\text{Dir}_{\eta}(\theta^*)} \cdot n^{-(K-1)/2}\cdot  e^{2nH(\theta^*)} .
\end{equation}

 This expression suggests that to reduce variance we should choose $\eta$ to increase $\text{Dir}_{\eta}(\theta^*)$. In fact, we will let $\eta = \eta_n$ depend on $n$ so that $\text{Dir}_{\eta_n}$ becomes increasingly concentrated around $\theta^*$.
 A natural choice to consider is the Dirichlet parameter $$\eta_{n}=\alpha + n^{\gamma}\theta^{*},\,\quad \text{for any } \gamma>0.$$ 
This choice of $\eta_n$ shifts the mode of $\text{Dir}_{\eta_n}$ toward $\theta^*$ as $n$ grows. This is easiest to see in the case when $\alpha_j+n^\gamma\theta_j^*>1$, for all $j$, since then
\begin{equation*}
\text{Mode}(\text{Dir}_{\alpha+n^\gamma \theta^*})_i
= \frac{\alpha_i + n^\gamma \theta^*_i -1}{\sum_{j=1}^{K} (\alpha_j + n^\gamma\theta^*_j )-K} \to \theta^*_i.
\end{equation*}
The rate at which the mode concentrates around $\theta^*$ is governed by the parameter $\gamma$. 

With this parameter choice, by expanding the likelihood ratio as in (\ref{eq:importance_sampling_change_of_measure}), 
the second moment of the IS estimator takes the form
\begin{equation}
\label{eq:secondmoment_IS_gamman}
\mathbb{E}_{\text{Dir}_{\alpha+n^\gamma\theta^{*}}}\left[\left(e^{nH(\theta)}\frac{\text{Dir}_{\alpha}(\theta)}{\text{Dir}_{\alpha+n^\gamma\theta^{*}}(\theta)}\right)^{2}\right] 
    = \frac{B(\alpha+n^\gamma\theta^{*})e^{-n^\gamma\theta^{*}\cdot \log\theta^{*}}}{B(\alpha)}\mathbb{E}_{\text{Dir}_{\alpha}}\left[e^{2nH(\theta)+n^\gamma\text{KL}(\theta^{*}|\theta)}\right],
\end{equation}
where we have introduced the Kullback-Leibler (KL) divergence defined as
\begin{align}
\label{eq:KL_divergence}
\text{KL}(\theta^{*}|\theta)
&=\sum_{k=1}^K\theta_{k}^{*}\log \frac{\theta_{k}^{*}}{\theta_{k}},\quad \theta\in\Delta_{K-1},
\end{align}
using the conventions $0\log 0=0$ and $0\log\infty=0$.

In \eqref{eq:secondmoment_IS_gamman}, we will see that the factor $B(\alpha+n^\gamma\theta^{*})e^{-n^\gamma\theta^{*}\cdot \log\theta^{*}}$ is $\Theta(n^{-\gamma(K-1)/2})$, implying a faster reduction with larger $\gamma$. 
On the other hand, the expectation in \eqref{eq:secondmoment_IS_gamman} involves the $\KL$ divergence term, which blows up at certain boundaries of the simplex. 
If $\gamma$ is too large, the KL term could result in variance that is exponentially larger than that of the standard MC estimator. 
While having $\text{Dir}_{\eta_{n}}$ concentrate near $\theta^{*}$ as $n\to \infty$ is useful for capturing the asymptotic behavior of the integrand, if it concentrates too quickly, the behavior of the likelihood ratio can dominate and slow the exponential rate of decay of the second moment. 
Designing an effective IS estimator requires balancing these considerations. 
The solution will be to take $\gamma$ close to but strictly less than 1.

\subsection{Truncation Near the Boundary}
\label{subsec:IS_estimator}

We have seen that an importance sampling scheme based on $\text{Dir}_{\eta_\gamma}$ introduces a factor of $\exp(n^\gamma \KL(\theta^*|\theta))$ in \eqref{eq:secondmoment_IS_gamman}, which depends on $\theta$ through the following term:
$$
\prod_{i:\theta^*_i>0} \theta_i^{\alpha_j-1-n^\gamma \theta^*_i}.
$$
The integrability of this term near the boundary face $\{ \theta_i=0\}$ is determined by whether $\alpha_i-n^\gamma \theta^*_i>0$ for all $i$ such that $\theta^*_i>0$, which is violated for $n$ large enough. 

We can avoid this difficulty by restricting the domain of our estimator and truncating points near certain boundaries. 
To ensure that the bias is negligible, we will ensure that the maximizer of the exponent remains in the domain after the truncation.

Let $\gamma\in (0,1)$. We propose the following $\gamma$-importance sampling (IS) estimator
\begin{equation}
\label{eq:IS_estimator}
\hat{p}_{\text{IS}}^{\gamma}=\frac{1}{N}\sum_{i=1}^{N}e^{nH({\theta}^{(i)})}\frac{\text{Dir}_{\alpha}({\theta}^{(i)})}{\text{Dir}_{\alpha+n^\gamma {\theta}^{*}}({\theta}^{(i)})}\mathbf{1}_{\Delta_{K-1}^{\epsilon}},\quad {\theta}^{(i)}\stackrel{\text{iid}}{\sim}\text{Dir}_{\alpha+n^\gamma {\theta}^{*}}
\end{equation}
where we introduce an additional modification to the IS estimator, an indicator function ${1}_{\Delta_{K-1}^{\epsilon}}$ of the truncated simplex, 
\begin{equation} \label{eq:truncated_simplex} \Delta_{K-1}^\epsilon = \left\{ \theta \in \Delta_{K-1} \,\middle|\, \theta_i \geq \epsilon \enspace \text{for all } i \text{ such that } \theta^*_i>0\right\}. \end{equation}
The requirement on the truncation factor $\epsilon$ is that $0<\epsilon < \min_{i:\theta^*_i>0} \theta_i^*$, to ensure that the truncated simplex still contains $\theta^*$.
In Theorem \ref{thm:boundary-mse-is}, we will see that any $\epsilon$ satisfying this condition will produce the same asymptotic behavior for \eqref{eq:IS_estimator} as $n\to\infty$.

In the following section, we show that for any valid choice of $\epsilon$, the bias introduced by the truncation diminishes exponentially fast and is of much smaller order than the variance of the standard MC estimator. 
As a result, for any $\gamma \in (0,1)$ and any valid $\epsilon$, our IS estimator improves upon the standard MC estimator in the mean squared error (MSE) sense as $n\to\infty$.

 \subsection{Laplace Method for the Importance Sampling Estimator}

Through the factorization in (\ref{eq:secondmoment_IS_gamman}),
analyzing the second moment of the estimator (\ref{eq:IS_estimator}) requires studying the behavior of the expectation
\begin{align}
    \mathbb{E}_{\text{Dir}_\alpha} \left[ e^{2nH({\theta})+n^{\gamma }\text{KL}(\theta^{*}|\theta)}\mathbf{1}_{\Delta_{K-1}^\epsilon}(\theta)\right].
    \label{eq:hkl}
\end{align}
A core challenge is that the term $n^{\gamma} \KL(\theta^*|\theta)$ is not concave in $\theta$, so the usual conditions for applying standard Laplace asymptotics are not directly satisfied. Nevertheless, when $\gamma<1$ we can prove an alternative version of the Laplace method which gives the asymptotics for the second moment. The reason is that on the truncated simplex the $\KL$ divergence term is uniformly bounded, which implies $n^\gamma \KL(\theta^*|\theta)=O(n^\gamma)$ on $\Delta_{K-1}^\epsilon$. Therefore, when $\gamma<1$, the $\KL$ term is of lower order than the leading $O(n)$ term $2nH(\theta)$, and the $\KL$ term acts as a sublinear perturbation of the Laplace exponent.

We now present our second main result, which proves a Laplace method for expectations of this form, addressing the two powers of $n$ in the exponent in \eqref{eq:hkl}. This combines the boundary analysis of Theorem~\ref{thm:boundary-laplace_standard} while controlling the lower-order but non-concave KL-divergence term.

\begin{thm}[Laplace method on the simplex with KL]
\label{thm:boundary-laplace-with-kl}
Suppose $H:\mathbb{R}^{K}\to\mathbb{R}$ satisfies  (A1)-(A4) at the maximizer $\theta^* \in \Delta_{K-1}$. Let $m$ be the number of zero components of $\theta^*$. If $\gamma \in (0,1)$ then as $n \to \infty$
\begin{equation}
\label{eq:sub_KL_integral} 
\mathbb{E}_{\text{Dir}_\alpha} \left[ e^{2nH({\theta})+n^{\gamma }\text{KL}(\theta^{*}|\theta)}\mathbf{1}_{\Delta_{K-1}^\epsilon}(\theta)\right]\sim C^{\prime}_H \cdot \exp({2nH(\theta^*)})\cdot n^{-\frac{(K-1-m)}{2}}\cdot n^{-\sum_{i =1}^{m}\alpha_{i}}
\end{equation}
where $C^{\prime}_H\in(0,\infty)$ is a constant.
\end{thm}
\begin{proof}
See proof in Section~\ref{sec:appendix-proof-boundary-laplace-with-kl}.
\end{proof}
A key implication of this result is that as $n\to\infty$, the behavior of this expectation is identical to $\mathbb{E}_{\text{Dir}_\alpha} \left[ e^{2nH({\theta})} \right]$. 
The variance reduction from importance sampling will thus be driven by the factor outside the expectation in \eqref{eq:secondmoment_IS_gamman}.

 \subsection{MSE Reduction}
\label{sec:MSE-of-Importance}

The proposed IS estimator in \eqref{eq:IS_estimator} achieves a reduction in mean squared error relative to the standard MC estimator as $n\to\infty$. We show that the extent of the reduction depends on the $\gamma$ chosen.

\begin{thm}[MSE reduction]
\label{thm:boundary-mse-is} 
Suppose assumptions (A1)-(A4) hold at the maximizer $\theta^*\in \Delta_{K-1}$. Let $m$ be the number of zero components of $\theta^*$.
Then as $n\to\infty$ 
\[
\frac{\text{MSE}(\hat{p}_{\text{IS}}^{\gamma})}{\text{MSE}(\hat{p}_{\text{MC}})}=\Theta\left(n^{-\gamma\cdot \frac{K-1-m}{2}}n^{-\gamma\sum_{i=1}^{m}\alpha_i}\right).
\]
\end{thm}
\begin{proof}
See proof in Section~\ref{sec:appendix-proof-boundary-mse-is}.
\end{proof}

The reduction rate is polynomial with exponent depending on the zero components of $\theta^*$, the corresponding Dirichlet parameter values $\alpha_i$, and the scale parameter $\gamma$ of the proposal distribution.
This follows from two key results: (i) the ratio 
${\text{\ensuremath{Var}}(\hat{p}_{\text{IS}}^{\gamma})}/{\text{Var}(\hat{p}_{\text{MC}})}$ decays polynomially as $n \to \infty$; and (ii) the truncation bias of 
$\hat{p}_{\mathrm{IS}}$ is of smaller order than 
$\text{Var}(\hat{p}_{\text{MC}})$, so that the MSE comparison is 
asymptotically determined by the variance term.

\begin{thm}[Variance reduction]
\label{thm:boundary-variance-is} 
Suppose assumptions (A1)-(A4) hold at the maximizer $\theta^*\in \Delta_{K-1}$ and $\gamma\in (0,1)$. 
Let $m$ be the number of zero components of $\theta^*$.
Then as $n\to\infty$
\[
\frac{\text{\ensuremath{Var}}(\hat{p}_{\text{IS}}^{\gamma})}{\text{Var}(\hat{p}_{\text{MC}})}=\Theta\left(n^{-\gamma \cdot \frac{K-1-m}{2}}n^{-\gamma\sum_{i =1}^{m}\alpha_{i}}\right).
\]
\end{thm}

\begin{proof}
See proof in Section~\ref{sec:appendix-proof-boundary-variance-is}.
\end{proof}

This theorem is a consequence of the following two observations:
(i) 
the expectation $$\mathbb{E}_{\text{Dir}_\alpha} \left[ e^{2nH({\theta})+n^{\gamma }\text{KL}(\theta^{*}|\theta)}\mathbf{1}_{\Delta_{K-1}^\epsilon}(\theta)\right]$$ 
asymptotically behaves the same as $\mathbb{E}_{\text{Dir}_\alpha} \left[ e^{2nH({\theta})} \right]$ by Theorem~\ref{thm:boundary-laplace-with-kl}, and (ii) the factor $B(\alpha+n^\gamma\theta^{*})e^{-n^\gamma\theta^{*}\cdot \log\theta^{*}}$ decays at rate $\Theta\left(n^{-\gamma({K-1-m})/{2}}n^{-\gamma\sum_{i =1}^{m}\alpha_{i}}\right)$.

\begin{lem}[Negligible bias]
\label{lem:boundary-bias-is} 
Suppose assumptions (A1)-(A4) hold at the maximizer $\theta^*\in \Delta_{K-1}$. Then there exists $\delta>0$ (independent of $\gamma)$ such that
\[
\frac{\text{\ensuremath{Bias}}(\hat{p}_{\text{IS}}^{\gamma})^{2}}{\text{Var}(\hat{p}_{\text{MC}})}=O(e^{-2n\delta}).
\]
\end{lem}
\begin{proof}
See proof in Section~\ref{sec:appendix-proof-boundary-bias-is}.
\end{proof}

The intuition is that the truncation is constructed so as not to exclude the maximizer $\theta^*$. As $n \to \infty$, the dominant contribution to the expectation arises from a shrinking neighborhood of $\theta^*$, which lies entirely within the non-truncated region. Consequently, the contribution from the truncated portion of the domain is asymptotically negligible.

\subsubsection{Strong Efficiency}
\label{sec:strong_efficiency}
The mean-squared error rate obtained in Theorem~\ref{thm:boundary-laplace_standard} 
provides theoretical support for our proposed importance sampling estimator and our use of a Dirichlet proposal distribution. In particular, Theorem~\ref{thm:boundary-laplace_standard} implies that our IS estimator is a strongly efficient estimator as $\gamma \to 1$. 
An estimator $\widehat{Z}_n$ of a sequence $\mu_{n}$ is defined to be \textit{strongly efficient}, or to have \textit{bounded relative error}, if 
$$ \limsup_{n\to\infty} \frac{\mathbb{E}[(\widehat{Z}_n - \mu_{n})^2]}{\mu_{n}^2}<\infty.$$
This is a generalization of the usual definition~(\cite{asmussen2007stochastic}, Chapter VI) to include biased estimators, and implies that the mean-squared error grows at most as the order of the squared mean.

The standard Monte Carlo estimator is \emph{not} strongly efficient in either the interior or the boundary case, as the ratio of the variance to the squared mean grows at the rate of $\Theta(n^{\frac{K-1-m}{2} + \sum_{i =1}^{m}\alpha_{i}})$. Our importance sampling estimator reduces this rate substantially. In fact, we can show that by choosing $\gamma$ sufficiently close to 1, the IS estimator’s relative MSE achieves arbitrarily small polynomial growth and thus can be made arbitrarily close to having bounded relative error:
for any $0<\xi<1$ there exists $\gamma \in (0,1)$ such that the IS estimator achieves a relative MSE bounded by $O(n^\xi)$. 

\begin{cor}
\label{cor:relative_MSE} (Efficiency rate) As $n\to\infty$, the ratio of the mean-squared error to the square of the estimand $I(n)$ is,
\begin{align}
\frac{\text{MSE}(\hat{p}_{\text{IS}}^{\gamma})}{I(n)^2} 
&=\Theta\left(n^{(1-\gamma)\left(\frac{K-1-m}{2} + \sum_{i=1}^{m}\alpha_{i}\right)}\right).
\end{align}
\end{cor}
This follows from the asymptotic rate of the mean-squared error:
$$\text{MSE}(\hat{p}_{\text{IS}}^{\gamma})
=
\Theta\!\left(
e^{2nH(\theta^*)}
\,n^{-(1+\gamma) \left(\frac{(K-1-m)}{2} + \sum_{i=1}^{m}\alpha_{i} \right)}
\right).$$
This rate can be made arbitrarily close to the 
MSE achievable by a strongly efficient importance sampling estimator, for which the MSE is of the same order as the square of the $\mathbb{E}[e^{nH(\theta)}]$, namely $O\!\left(
e^{2nH(\theta^*)}
\,n^{-(K-1-m)}
\,n^{-2\sum_{i =1}^{m}\alpha_i}
\right).$

\section{Control Variate}  
\label{sec:control_variate}
\subsection{Introduction}  
In this section, we investigate how related ideas can be used to design and analyze control variates for variance reduction. We begin by introducing the control variate framework in our setting.

Suppose $\widehat{H}: \Delta_{K-1} \to \mathbb{R}$ is a function for which $e^{n \widehat{H}(\theta)}$ is correlated with $e^{n H(\theta)}$ and the expectation $\mathbb{E}_{\text{Dir}_{\alpha}} \left[ e^{n \widehat{H}(\theta)} \right]$ is known in closed form. An example of $\widehat{H}$ could be an approximation of $H$. 
We call $e^{n\widehat{H}(\theta)}$ a control variate.  

We can use the control variate to define a new estimator which improves upon the standard MC estimator. Define 
\begin{align}
\label{eq:CV_estimator_generic} 
\hat{p}_{\text{CV}}	&=\underbrace{\frac{1}{N}\sum_{i=1}^N e^{nH(\theta^{(i)})}}_{\hat{p}_{\text{MC}}}
+c\left({\frac{1}{N}\sum_{i=1}^N e^{n\widehat{H}(\theta^{(i)})}}-\mathbb{E}_{\text{Dir}_{\alpha}}[e^{n\widehat{H}(\theta)}]\right), 
\quad \theta ^{(i)}\stackrel{iid}{\sim} \text{Dir}_\alpha  
\end{align}
  where  $c$ is a constant. For any $c$, this is an unbiased estimator of $\mathbb{E}_{\text{Dir}_{\alpha}}[e^{n H(\theta)}]$.
 It is known that 
 the variance-minimizing choice of $c$ is given by
\begin{equation}
\label{eq:optimal_CV_coefficient}
 c^* =-\frac{\text{Cov}(e^{nH(\theta)},e^{n\widehat{H}(\theta)})}{\text{Var}(e^{n\widehat{H}(\theta)})},  
\end{equation}
and the resulting variance reduction is
\begin{align}
\label{eq:correlation}
\frac{\text{Var}(\hat{p}_{\text{CV}})}{\text{Var}(\hat{p}_\text{{MC}})} = (1-\rho_n^2), \quad  \rho_n=\frac{\text{Cov}(e^{nH(\theta)},e^{n\widehat{H}(\theta)})}{\sqrt{\text{Var}(e^{nH(\theta)})\text{Var}(e^{n\widehat{H}(\theta)})}},
\end{align}
where $\rho_n$ is the correlation between $e^{nH(\theta)}$ and $e^{n\widehat{H}(\theta)}$.

From \eqref{eq:correlation} we see that  
the effectiveness of a control variate $e^{n\widehat{H}(\theta)}$ is determined by its correlation with $e^{nH(\theta)}$.
 Since the quantities in the numerator and denominator of $\rho_n$ in \eqref{eq:correlation} both have exponential form, it is natural to consider applying the Laplace method to evaluate $\rho_n$, at least approximately.

\subsection{A KL-Based Control Variate}

 For the control variate based estimator to achieve variance reduction, the limiting correlation must not vanish. Therefore we need the leading-order contribution to the numerator to be of the same order as the leading-order contribution to the denominator. 
The KL divergence term  $\KL(\theta^*|\theta)$ which emerges from the likelihood ratio in importance sampling  gives rise to such control variate. 
As before, let $\theta^*$ denote the maximizer of $H$ at which assumptions (A1)-(A4) are satisfied. 
Define
\begin{equation}
\label{eq:lowerbound_mu_incentive}
\widehat{H}(\theta) = H(\theta^*)- \KL (\theta^*|\theta).
\end{equation}
We note a few properties of $\widehat{H}$. First, since $\widehat{H}$ involves the negative KL divergence (unlike for importance sampling), it is concave, and it is maximized at $\theta^*$. Second, the unique maximizer of the sum $H+\widehat{H}$ also remains at $\theta^*.$ Third, the the sum $H+\widehat{H}$ satisfies the negative definiteness property (A4) which allows the use of the Laplace method (shown later). Fourth, there is a closed form for the expectation 
\begin{equation}
\label{eq:control_variate_expectation}
    \mathbb{E}_{\text{Dir}_\alpha}[e^{n\widehat{H}(\theta)}]=\frac{B(\alpha+n\theta^*)}{B(\alpha)}\cdot e^{n(H(\theta^*)-\theta^*\cdot \log\theta^*)}.
\end{equation}
These properties make $\exp(n\widehat{H}(\theta))$ a promising control variate.

In the case of LDA (or any model with $\mu=1$ in the KKT conditions) with an interior $\theta^*$, $\hat{H}$ also emerges has a first-order Taylor expansion of $H$ in $\log\theta$ around $\log\theta^*$. This further suggests that our control variate should be highly correlated with the target, at least in these cases.

\subsection{Variance Reduction}

To quantify the variance reduction achieved by this CV estimator, we analyze the correlation $\rho_n$ in \eqref{eq:correlation} using the Laplace method in Theorem ~\ref{thm:boundary-laplace_standard}. When the maximizer $\theta^*$ is a boundary point, the relevant quantity, instead of the full Hessian, is the reduced Hessian on the critical cone $\mathcal{C}(\theta^*)$. If $m \le K-2$, the condition in (A4) is equivalent to the negative definiteness of the reduced Hessian
\begin{equation}
\label{eq:reduced_hessian}
U^\top \nabla^2 H(\theta^*)\, U \prec 0 ,
\end{equation}
where $U$ is defined as 
\begin{equation}
\label{eq:U_explicit}
U =
\begin{pmatrix}
0_{m\times (K-1-m)} \\
I_{K-1-m} \\
-\mathbf{1}^\top_{K-1-m}
\end{pmatrix}
\in \mathbb{R}^{K\times (K-1-m)} .
\end{equation}
The columns of $U$ form a basis of the critical cone $\mathcal{C}(\theta^*)$. With this, we discuss the limiting correlation achieved by our control variate $\hat{p}_{\text{CV}}$. 

\begin{thm}[Limiting Correlation: Boundary Case]
\label{thm:CV_correlation_boundary} 
Suppose $H:\mathbb{R}^K\to\mathbb{R}$ satisfies assumptions (A1)-(A4) at the maximizer $\theta^*$. Let $m\le K-2$ be the number of zero components of $\theta^*$ and let $\lambda$ denote the KKT multipliers of the inequality constraints $\theta_i\ge0$. Then as $n\to\infty$
\begin{align}
\label{eq:limiting_rho}
\rho^{2}_n \to  \rho^2:= 
\left( \, \frac{\left|\det\!Q\right|^\frac{1}{2} \, \left|\det\!R\right|^\frac{1}{2}}
{\left|\det\!\left(\frac{Q+R}{2}\right)\right|}\right) 
\prod_{k=1}^{m}\left(\frac{4 \lambda_k}{(\lambda_k+1)^2}\right)^{\alpha_k}
\end{align}
where $Q=U^{\top}\nabla^{2}H(\theta^{*})U$ and $R= -U^{\top}\nabla^2\text{KL}(\theta^*|\theta)\big|_{\theta=\theta^*}U$. 
\end{thm}

\begin{proof}
See proof in Section~\ref{sec:CV_limiting_correlation_boundary}. The proof relies on the fact that the numerator and denominator of $\rho_n^2$ have matching exponential orders, by the properties of $\widehat{H}$ and $H+\widehat{H}$ discussed previously. The polynomial factors $n^{-(K-1-m)/2}$ cancel so that the variance reduction is determined by the remaining constants in \eqref{eq:limiting_rho}.\end{proof}

\begin{remark}
We can identify $\rho^2$ in (\ref{eq:limiting_rho}) as involving 1) the ratio of the determinants of the geometric mean to the arithmetic mean of the of the corresponding reduced Hessians and 2) a function of the KKT multipliers. 

The determinant factor on the left is between 0 and 1 by the log-concavity of the determinant on the cone of $(K-1-m)\times (K-1-m)$ symmetric positive definite matrices $\mathbb{S}_{++}^{K-1-m}$:
\begin{align}
\label{eq:log_determinant_inequality}
\det \left(\frac{Q+R}{2} \right)
&\ge  \sqrt{\det Q \det R},\quad \text{for all } Q,R\in\mathbb{S}^{K-1-m}_{++}.
\end{align}
The KKT factor on the right is less than 1 since $(\lambda_k + 1)^2 \ge 4\lambda_k$ for any $\lambda_k$. Therefore $\rho^2$ lies between $0$ and $1$, consistent with the fact that square of a correlation must take values in $[0,1]$. 

It remains of interest to identify when $\rho^2$ is close to 1, as this corresponds to maximal variance reduction. We note that $\rho^2=1$ when both the Hessian and KKT factors equal one. This happens when $Q=R$ (by the strict concavity of $\log\det$) and when $\lambda_k=1$ for $k\le m$. More generally, we can expect the geometric mean to be close to the arithmetic mean when two matrices $Q,R$ are close. Hence, the variance reduction is near optimal when the corresponding reduced Hessians are close and the $\lambda_k$'s are near 1.
\end{remark}

When $\theta^*$ is a vertex point on the simplex ($m=K-1$), the critical cone reduces to $\mathcal C(\theta^*)=\{ \mathbf{0}\}$. Therefore the negative definiteness of the Hessian condition (A4) is vacuous and holds automatically. The result in Theorem \ref{thm:CV_correlation_boundary} still holds except the determinant factor in \eqref{eq:limiting_rho} reduces to 1. 
On the other hand when $\theta^*$ is an interior point ($m=0$), the KKT factor disappears and only the Hessian factor remains. When $H$ is the log-likelihood function of LDA, we get a simpler expression for $\rho^2$. Instead of the reduced Hessians, it suffices to look at the ratio of the full Hessians. 

\begin{cor}[Limiting Correlation: LDA Interior Case]
\label{cor:CV_correlation_interior} 

Suppose $H:\mathbb{R}^K\to\mathbb{R}$
is the log-likelihood function of LDA defined in \eqref{eq:entropy} with an interior maximizer $\theta^*$.
Then as $n\to\infty$ 
\begin{align}
\label{eq:limiting_rho_interior}
\rho^{2}_n \to  \rho^2:= 
\frac{|\det Q|^{1/2}\; |\det R|^{1/2}}
     {\left|\det\left(\frac{Q+R}{2}\right)\right|}\,
\end{align}
where $Q=\nabla^2 H(\theta^*)$ and $R= -\nabla^2\text{KL}(\theta^*|\theta)\big|_{\theta=\theta^*}$.
\end{cor}

\begin{proof} 
See proof in Section~\ref{sec:appendix-proof-CV_correlation_interior}.
\end{proof}

In the setting of Corollary \ref{cor:CV_correlation_interior}, the matrix $R$ (the Hessian of $\KL$
evaluated at $\theta=\theta^*$) takes the form
\begin{equation*}
  \nabla^2_\theta \mathrm{KL}(\theta^*|\theta)\big|_{\theta=\theta^*}
  = \sum_{i=1}^{K}
  \frac{1}{\theta_i^*} e_i e_i^\top,
  \qquad e_i \in \mathbb{R}^K .
\end{equation*}
In particular, $R$ is a diagonal matrix with nonzero entries on the support of $\theta^*$. 
Thus, for $Q$ to be close to $R$, the Hessian of $H$ at the maximizer must be approximately diagonal, with its dominant contribution along the diagonal entries. 
This indicates that the proposed control variate is particularly effective in problems where the Hessian at the maximizer is sparse or nearly diagonal. We will see that this can naturally arise in the case of LDA.

\vspace{1em}

\subsection{LDA: Almost Mutually Orthogonal Case}
\label{subsection: mutually_orthogonal_case}

In LDA, a topic vector $\phi_k$ typically puts most of its mass on a small subset of words in the vocabulary, which comprise the core words in the topic. Different topics concentrate mass on different subsets. These properties make the topic vectors nearly orthogonal. We will formulate this property precisely and show that it leads to a limiting correlation close to 1.

We assume that each element of the vocabulary has a most-likely topic, in the following sense:
\begin{assumplistB}
    \item For every $v\in V$, the probability $\phi_k(v)$ is maximized by a unique topic index, defined by
$$ k(v)=\arg\max_k \phi_k (v).$$
\end{assumplistB}
If the topic vectors $\phi_1,\dots,\phi_K$ are sampled independently from a Dirichlet distribution, (B1) holds almost surely. Even if $\phi$ is estimated deterministically (e.g., via variational inference), (B1) is not restrictive because individual words tends to be important for a small number of topics.

We call an LDA topic model $\phi$ \textit{$\varepsilon$-sparse} if
\begin{equation}
\label{eq:sparsity}
\varepsilon = \max_v \frac{\sum_{j\neq k(v)} \phi_j (v)}{\phi_{k(v)}(v)}. 
\end{equation}
The ratio measures the total mass assigned to non-dominant topics relative to the mass on its preferred topic $k(v)$ for each vocabulary element $v$. 
Thus, smaller values of $\varepsilon$ correspond to stronger concentration on $k(v)$ for every $v$, and the vectors $\phi(v)$ approach 0--1 vectors as $\varepsilon \to 0$.

Next, we define a function that will be used to define a threshold for $\varepsilon$. Suppose $\theta^*$ is an interior point. Define 
\begin{equation}
\label{eq:F_epsilon}
F(\varepsilon):=4 (C_{\max}^{(2)})^2 K+K(K-1)\left(2C_{\max}^{(1)}+C_{\max}^{(2)}\varepsilon\right)^{2}  
\end{equation}
where
\begin{equation}
         \label{eq:constants_max}
         C_{\max}^{(1)}:=\frac{(\theta_{\max}^{*})^{1/2}}{(\theta_{\min}^{*})^{3/2}}, \quad C_{\max}^{(2)}:=\frac{\theta_{\max}^{*}}{(\theta_{\min}^{*})^{2}},\quad \theta_{\min} = \min_{k}\theta_k^{*}, \quad \theta_{\max} = \max_{k}\theta_k^{*}. 
\end{equation}
We note that $\varepsilon \sqrt{F(\varepsilon)}$ is strictly increasing for $\varepsilon>0$ and vanishes at $\varepsilon=0$. Let $\varepsilon_0>0$ be the unique root of $\varepsilon \sqrt{F(\varepsilon)}=1/2$.

We can show that the squared correlation $\rho^2$ is close to $1$ if the topic vectors $\phi$ are $\varepsilon$-sparse with $\varepsilon>0$ small enough such that $\varepsilon < \varepsilon_0$. 

\begin{thm}
\label{thm:epsilon_main_result}
Suppose $H:\mathbb{R}^K\to\mathbb{R}$
is the log-likelihood function of LDA defined in \eqref{eq:entropy} with an interior maximizer $\theta^*$. Assume that $K$ (the number of topics) is greater than 3. If the topic vectors $\phi$ satisfies (B1) and are $\varepsilon$-sparse with
$\varepsilon < \varepsilon_0$, 
then the squared correlation in \eqref{eq:limiting_rho_interior} satisfies the lower bound
\begin{equation}
\label{eq:epsilon_main_result}
     \rho^2 \ge \exp\left(-C^2\varepsilon^2\right), 
\end{equation}
where $C>0$ is a constant that depends on $\theta^*$, $K$, and $\varepsilon_0$. 
\end{thm}

\begin{proof}
    See proof Section in \ref{sec:appendix-proof-epsilon_main_result}. 
\end{proof}

\begin{remark}
The choice of $1/2$ in the definition of $\varepsilon_0$ is arbitrary; 
any fixed $\nu \in (0,1)$ yields an analogous bound with 
$C = \sqrt{F(\varepsilon_0)/(1-\nu)}$. 
Larger $\nu$ relaxes the sparsity requirement but increases $C$.
 \end{remark}

The proof, given in Section~\ref{sec:appendix-proof-epsilon_main_result}, 
relies on properties of the $\log\det$ function to quantify how the ratio between the geometric and arithmetic means of the determinants of the Hessians approaches $1$ under $\varepsilon$-sparsity.


\begin{figure}[t]
\centering

    \includegraphics[height = 3.3in]{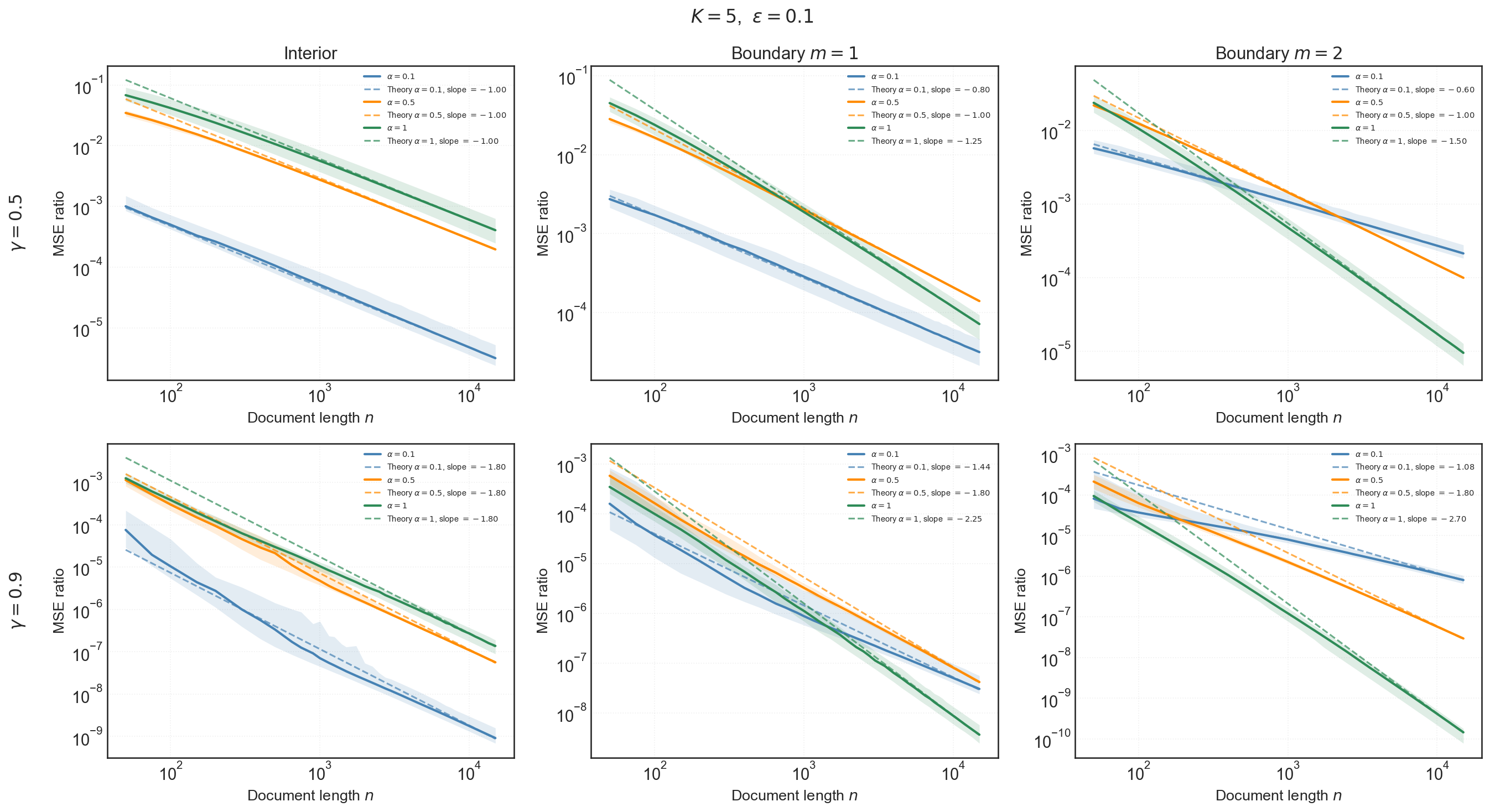}    
\caption{\label{fig:imp_sampling_epsilon_0p1_K_5} 
    $\text{MSE}(\hat{p}_{\text{IS}})/\text{MSE}(\hat{p}_{\text{MC}})$ versus $n$ on a log-log scale for $K=5$ topics and $\epsilon=0.1$. 
Top: $\gamma=0.5$. Bottom: $\gamma=0.9$. 
Results are computed over 100 independent instances of $(\phi,\theta^*,p)$; solid lines denote the empirical median and shaded regions indicate the interquartile range.
Dashed lines represent the theoretical asymptotic reduction rate.
}
\end{figure}

\section{Numerical Experiments}
\label{sec:numerical}

In this section, we present numerical experiments designed to illustrate the performance of the proposed estimators. We first study the estimators in a synthetic setting, and then evaluate them on a dataset of news articles.

\subsection{Synthetic Data Experiments}
\label{sec:synthetic_experiments}

We evaluate the proposed estimators using synthetic instances with vocabulary size $V = 1000$ and topic size $K=5$. In all experiments, the topic-word distributions are sampled independently as $\phi_k \sim \text{Dir}_{\beta}$ with $\beta = 0.1\,\mathbf{1}_V$.

To generate an instance of an interior case, we sample $\theta_{\text{true}} \sim \mathrm{Dir}_{\alpha_\text{true}}$ with $\alpha_{\text{true}} = \,\mathbf{1}_K$, which lies in the interior of the simplex almost surely. We then set $p_v = \sum_{k=1}^K \theta_{\text{true},k} \phi_k (v)$, so that the maximizer satisfies $\theta^*=\theta_{\text{true}}$. For boundary cases, we prescribe a boundary point $\theta^*$ on a face of the simplex with $m$ zero components for each $m\in \{ 0,1,2\}$ and construct $p=(p_v)$ such that $\theta^*$ is the maximizer of $H$  and the associated KKT multipliers satisfy strict complementarity. The details of this sampling method are in Supplemental Material \ref{sec:simulation_boundary_cases}.

Each such construction yields a single instance of $(\phi,p,\theta^*)$. In the experiments below, we generate multiple independent instances according to this procedure to assess variability across problem configurations. 

\subsubsection{Importance Sampling Estimator}

  Preliminary numerical results limited to the interior case $(m=0)$ with $\gamma=0.5$ and $\alpha=0.1$ were reported in \cite{glee}. Here we provide more comprehensive numerical experiments comparing $\gamma\in \{0.5,0.9\}$, $m\in\{0,1,2\}$, and Dirichlet prior parameters $\alpha\in\{0.1,0.5,1\}\mathbf{1}_K$ to illustrate the more general results proved here. For each configuration $(K,m)$, we generate $100$ independent instances of $(\phi,p,\theta^*)$. 
For importance sampling, we set $\epsilon=0.1$ and apply truncation coordinate-wise: we retain a sample $\theta$ if $\theta_i \ge \epsilon \theta_i^*$ for all $i$ with $\theta^*_i>0$.
  We consider document lengths $n$, ranging from $50$ to $1.5\times 10^4$. For the standard Monte Carlo estimator, its MSE equals its variance. Additional details of the implementation are discussed in Supplemental Material~\ref{sec:degeneracy_plainMC}. For estimating the moments of the standard MC estimator and IS estimator we use $N=10^6$ samples.

Since $\hat{p}_{\text{IS}}$ has a small bias, we report the log MSE ratio $\log (\text{MSE}(\hat{p}_{\text{IS}}))/(\text{MSE}(\hat{p}_{\text{MC}}))$. For each $n$, the results are computed over 100 instances of $(\phi,\theta^*, p)$; we plot the empirical medians (solid lines) with interquartile-range bands
(Figure~\ref{fig:imp_sampling_epsilon_0p1_K_5}). Consistent with Theorem~\ref{thm:boundary-mse-is}, the decay of the MSE ratio depends on the number of zero components $m$ and the corresponding Dirichlet parameters $\alpha_i$. The slope of the dashed line corresponds to the exponent in Theorem ~\ref{thm:boundary-mse-is}, i.e., $-\gamma\big((K-1-m)/2+\sum_{i=1}^m \alpha_i\big)$. 
To facilitate comparison, the intercept for each dashed line is selected by fitting the line to the simulation values using the five largest values of $n$.

As expected, we see greater MSE reduction at $\gamma=0.9$ than at $\gamma=0.5$. We also confirm that $\text{Bias}^2(\hat{p}_{\text{IS}})/\text{MSE}(\hat{p}_{\text{MC}})$ decays rapidly with $n$, in agreement with
  the $O(e^{-2n\delta})$ upper bound in Lemma~\ref{lem:boundary-bias-is} (see Figure~\ref{fig:bias_and_gamma_1}). With $\gamma=1$, we have observed that the MSE ratio diverges as $n$ increases and becomes unstable even at small $n$. 
Overall, the proposed IS estimator achieves substantial variance reduction with negligible bias relative to the standard MC estimator when $\gamma<1$.

\begin{figure}[t]
\centering
  \includegraphics[width=\linewidth]
{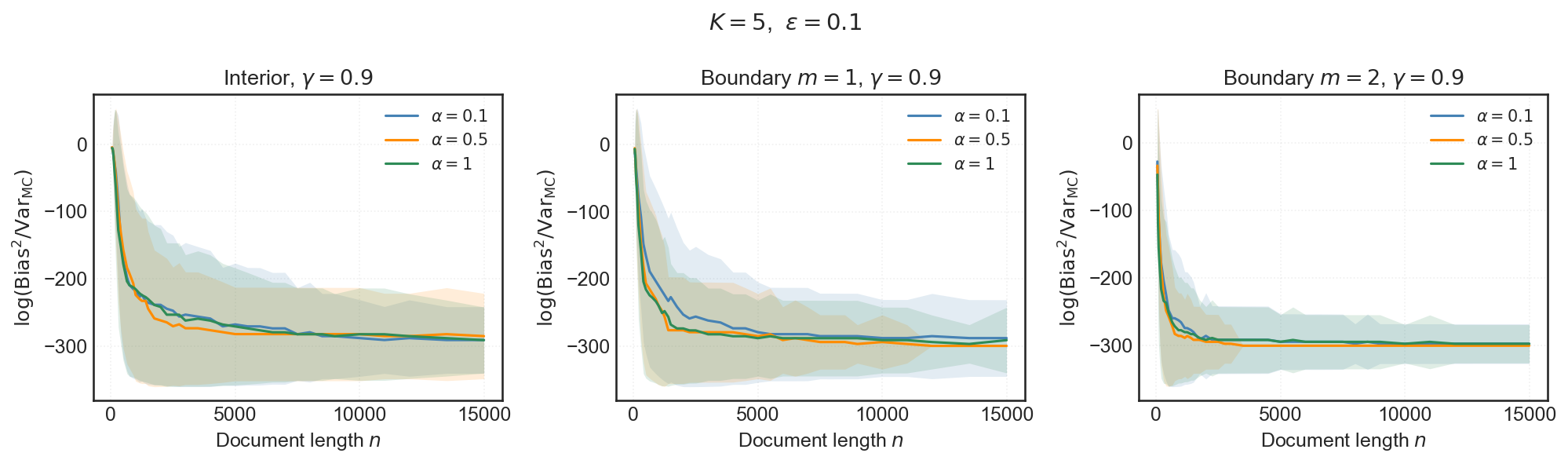}
\hfill
\caption{\label{fig:bias_and_gamma_1}
Plot of $\log(\text{Bias}(\hat p_{\text{IS}})^2/\text{MSE}(\hat p_{\text{MC}}))$ versus $n$ for $K=5$, $\epsilon=0.1$, and $\gamma=0.9$, for $\alpha \in \{0.1,0.5,1\}\mathbf{1}_K$. The solid lines indicate the average over 100 instances and shaded bands show $\pm 1$ standard deviation.  The normalized bias decays rapidly with $n$, consistent with
  the $O(e^{-2n\delta})$ upper bound of Lemma~\ref{lem:boundary-bias-is}.  Instances with zero truncation bias (i.e., no samples were truncated)
  are assigned a floor of $10^{-300}$. By $n = 1000$ approximately $70$--$95\%$ of instances have zero bias, rising to $95$--$100\%$       by $n = 15{,}000$. }

\end{figure}

\subsubsection{Control Variate Estimator}

    As before, for each configuration $(K,m)$, we generate $100$ independent instances of $(\phi,p,\theta^*)$.
    For each instance, we vary the document length $n$ over the range $n=50$ to $1.5\times 10^4$ and consider Dirichlet priors $\alpha\in \{ 0.1,2\}\mathbf{1}_K$. For each instance,
we compute the log ratio    $\log(1 - \hat{\rho}_n^2) - \log(1 - \rho^2)$ between the log variance reduction and its theoretical limit. 
Results are summarized across the $100$ independent instances by plotting the median and interquartile range (see Figure~\ref{fig:oracle_theory_log_ratio}). 
The results illustrate the convergence of the variance reduction to its theoretical limit in Theorem \ref{thm:CV_correlation_boundary}.

\begin{figure}[t]
\centering

\includegraphics[width=0.48\linewidth]
{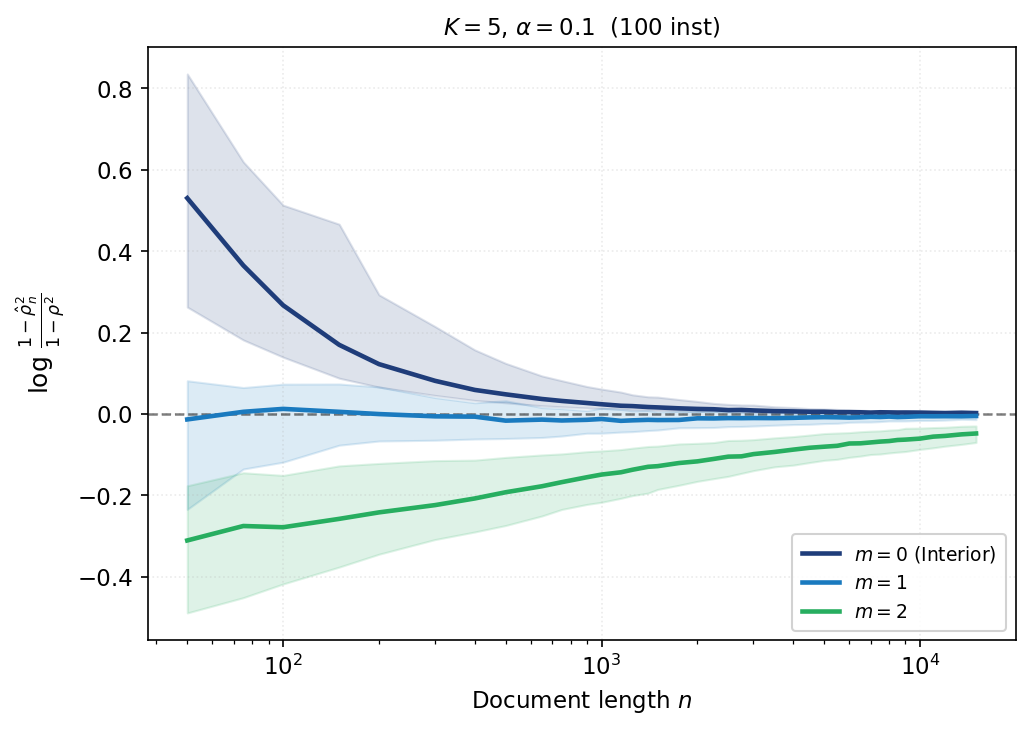}
\includegraphics[width=0.48\linewidth]
{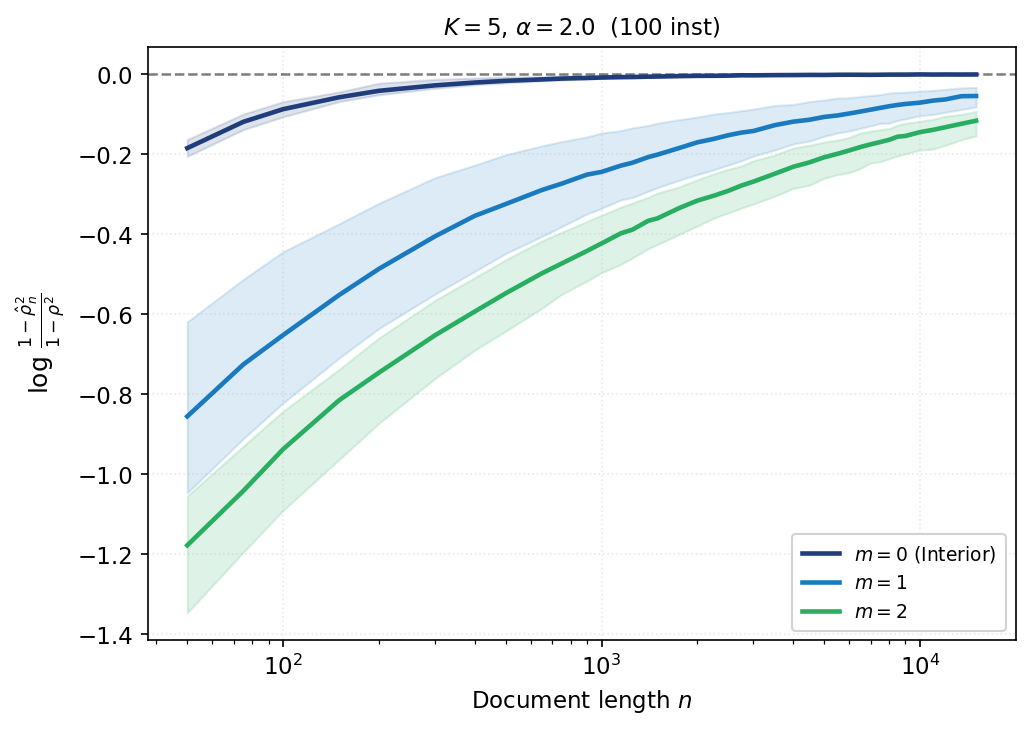}

\caption{\label{fig:oracle_theory_log_ratio}
Log-ratio of the variance reduction to its theoretical limit versus $n$ for $K=5$ topics. Left: $\alpha=0.1 \mathbf{1}_K$. Right: $\alpha =2 \mathbf{1}_K$. Results are computed over 100 independent instances of $(\phi,\theta^*,p)$; solid lines denote the empirical medians and shaded regions indicate
the interquartile range.
The dashed line at $0$ indicates agreement with theory.}

\end{figure}

  We also test Theorem~\ref{thm:epsilon_main_result} by examining how the empirical $\hat{\rho}^2_n$ varies with sparsity level $\varepsilon$. For each $\varepsilon$ in the range $10^{-7}$ to $5$, we generate  
  200 synthetic topic matrices $\phi$ with $K=10$ topics and $V=1000$ vocabulary, constructing each $\phi$ so that its sparsity level matches the target. For each $\phi$, we draw 20 synthetic documents of length 
  $n=1000$ from the LDA generative model, retaining only those with interior maximizer $\theta^*$, and estimate $\rho^2_n$ via Monte Carlo with $10^5$ draws from        
  $\text{Dir}_\alpha$ where $\alpha=0.1\mathbf{1}_K$. Figure~\ref{fig:sparsity_test} shows that $\hat{\rho}^2_n \to 1$ as $\varepsilon \to 0$, consistent with the lower bound in Theorem~\ref{thm:epsilon_main_result}.

\begin{figure}[t]
\centering
\includegraphics[height = 2.4in]
{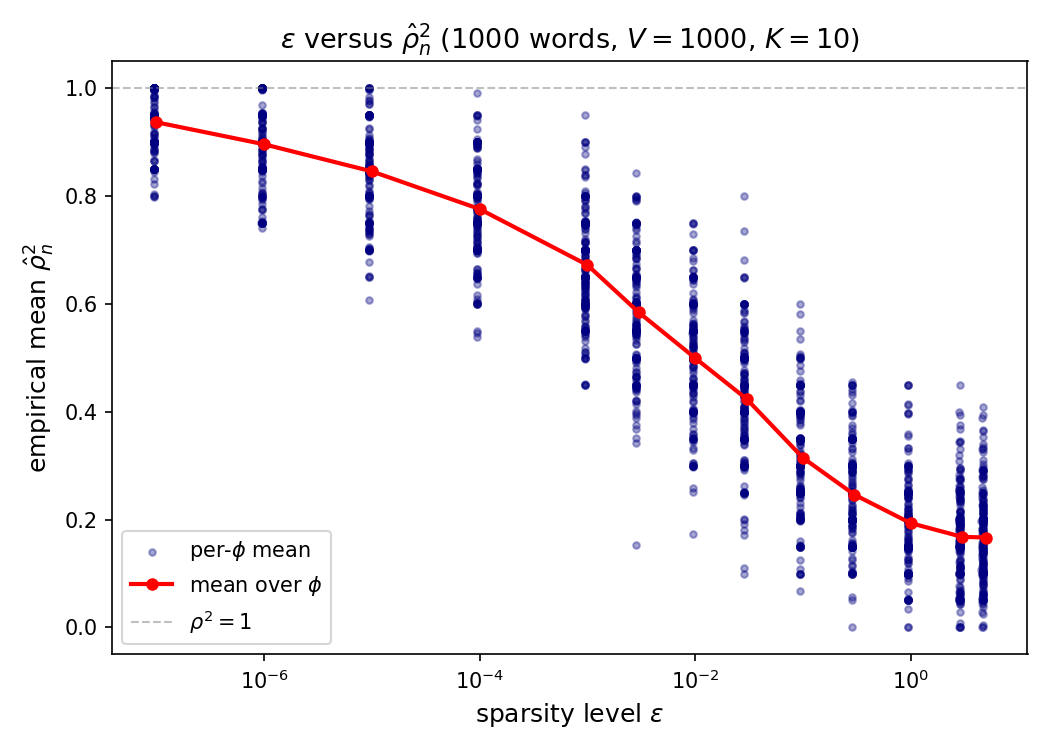}

\caption{\label{fig:sparsity_test}
A scatter plot of the empirical mean $\hat{\rho}^2_n$ across 200 controlled $\phi$ matrices per sparsity level $\varepsilon$, with $K=10$, $V=1000$, and $n=1000$ words. The red line connects the group means. As $\varepsilon \to 0$ (sparser $\phi$), $\hat{\rho}^2_n$ approaches 1, consistent with the theoretical prediction in Theorem~\ref{thm:epsilon_main_result}.
}

\end{figure}

\subsection{Real Dataset: Reuters Corpus}

We next evaluate our estimators on the Reuters-21578 corpus~\citep{lewis1997reuters21578}, a standard text classification dataset, accessed through the Natural Language Toolkit (NLTK) interface~\citep{BirdKleinLoper2009}. In the NLTK distribution, the dataset contains 10,788 documents (about 1.3 million words), partitioned into train and test datasets. The vocabulary size $V$ is $14,838$ words. For our
model, we choose $K = 10$ topics and use variational inference to extract the topic-word distributions $\phi$
from 7,770 training documents, with a default prior of $\alpha = 0.1 \cdot \mathbf{1}_{K}$. Then, for each of the remaining 3,018 test documents, we calculate the empirical word frequencies $p_{v}$ and evaluate the MSE ratio between our estimators and  the standard Monte Carlo estimator using $N=10^5$ samples each.

This setting does not strictly fall within the scope of our theoretical analysis because here we cannot vary $n$ while holding $p_v$ fixed; we simply have documents of different lengths. Also, we cannot guarantee (A1).
We can nevertheless check for error reduction through the MSE ratio. 

The left panel in Figure~\ref{fig:hist_imp_sampling_reuters} 
shows the distribution of log-MSE ratios across the test documents for the proposed importance sampling estimator with $\gamma=0.9$ and $\epsilon=0.1$. The results show substantial error reduction across the corpus. For 99.7\% of test documents, the importance sampling estimator achieves variance reduction over the standard MC estimator. We observe one outlier document with high MSE ratio and short length ($n=32$), which is not depicted in the figure. Furthermore, 95\% of test documents show a log-MSE ratio less than $-1.63$ (equivalent to an MSE ratio of $0.023$), with a median log-MSE ratio of $-2.94$ (equivalent to an MSE ratio of $0.001$). In other words, for at least half of the documents, the estimator achieves more than a $863\times$ reduction in MSE.  We also observe a statistically significant negative correlation between document length and MSE ratio (Spearman $\rho = -0.458$, $p$-value $= 3.3  
  \times 10^{-156}$), indicating that importance sampling gives larger variance reductions for longer documents.  

The importance sampling estimator requires the computation of the maximizer $\theta^*$, which we use the recursive algorithm in \cite{cover1984algorithm} to calculate. We note that the overhead incurred by this step is negligible. 
Running the algorithm for $100$ steps takes 21.0 ms on an AMD EPYC 7B12 CPU (2.25 GHz), which is approximately $290$ times faster than sampling $N = 10^{6}$ points from the Dirichlet distribution and computing $\frac{1}{N}\sum_{i=1}^{N} e^{nH(\theta_{i})}$ (8.99 s).
Therefore it has no practical impact on the overall runtime.

Next we turn to the control variate estimator.
The right panel of Figure~\ref{fig:hist_imp_sampling_reuters}
shows the distribution of log-MSE ratios for $\hat{p}_{\text{CV}}.$
Improvement is again observed across all documents, although
the variance reduction is now smaller, with a median log MSE ratio of $-0.19$ (equivalent to an MSE ratio of 0.65).

\begin{figure}[t]
\centering
\includegraphics[width=0.48\linewidth]{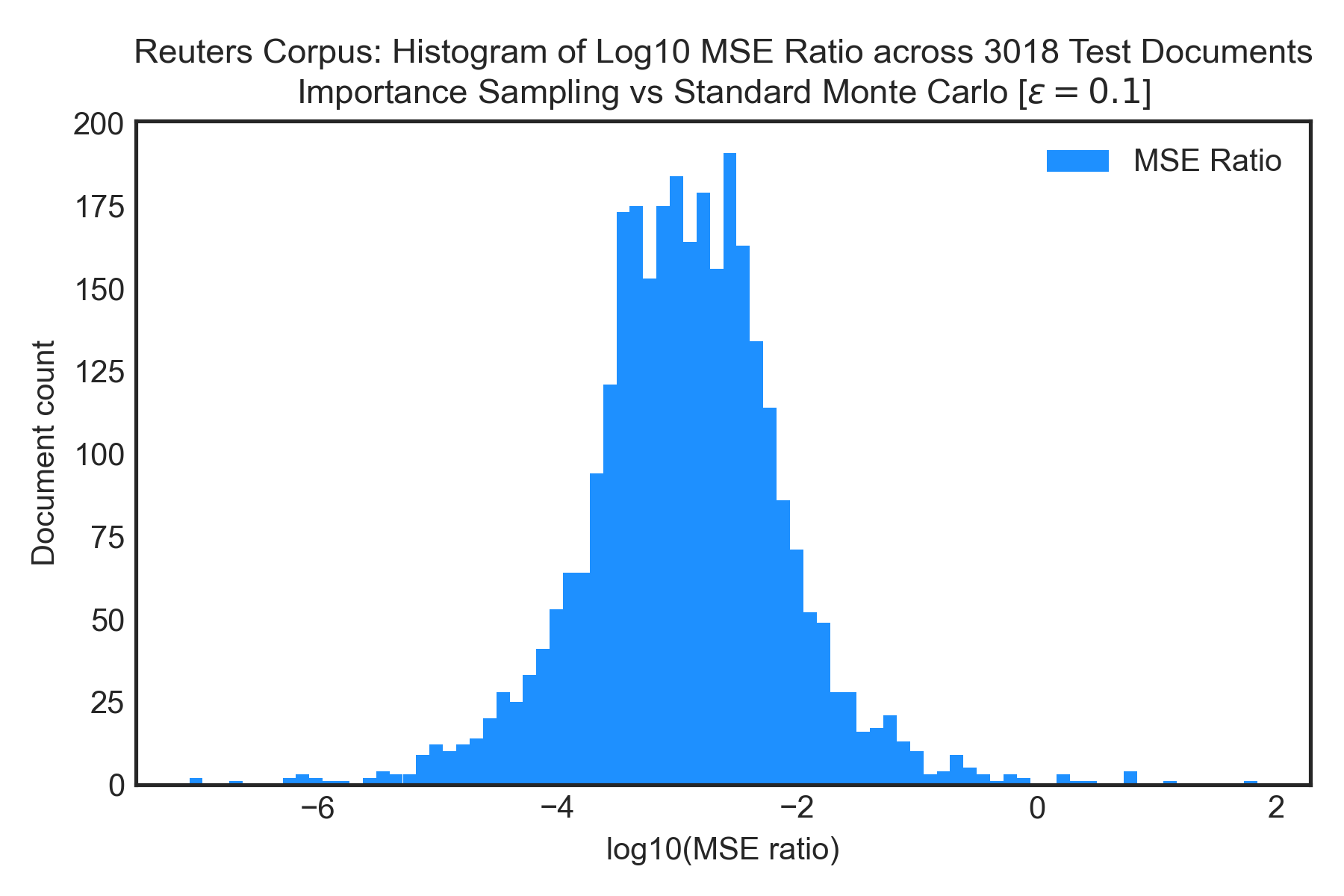}
\includegraphics[width=0.48\linewidth]{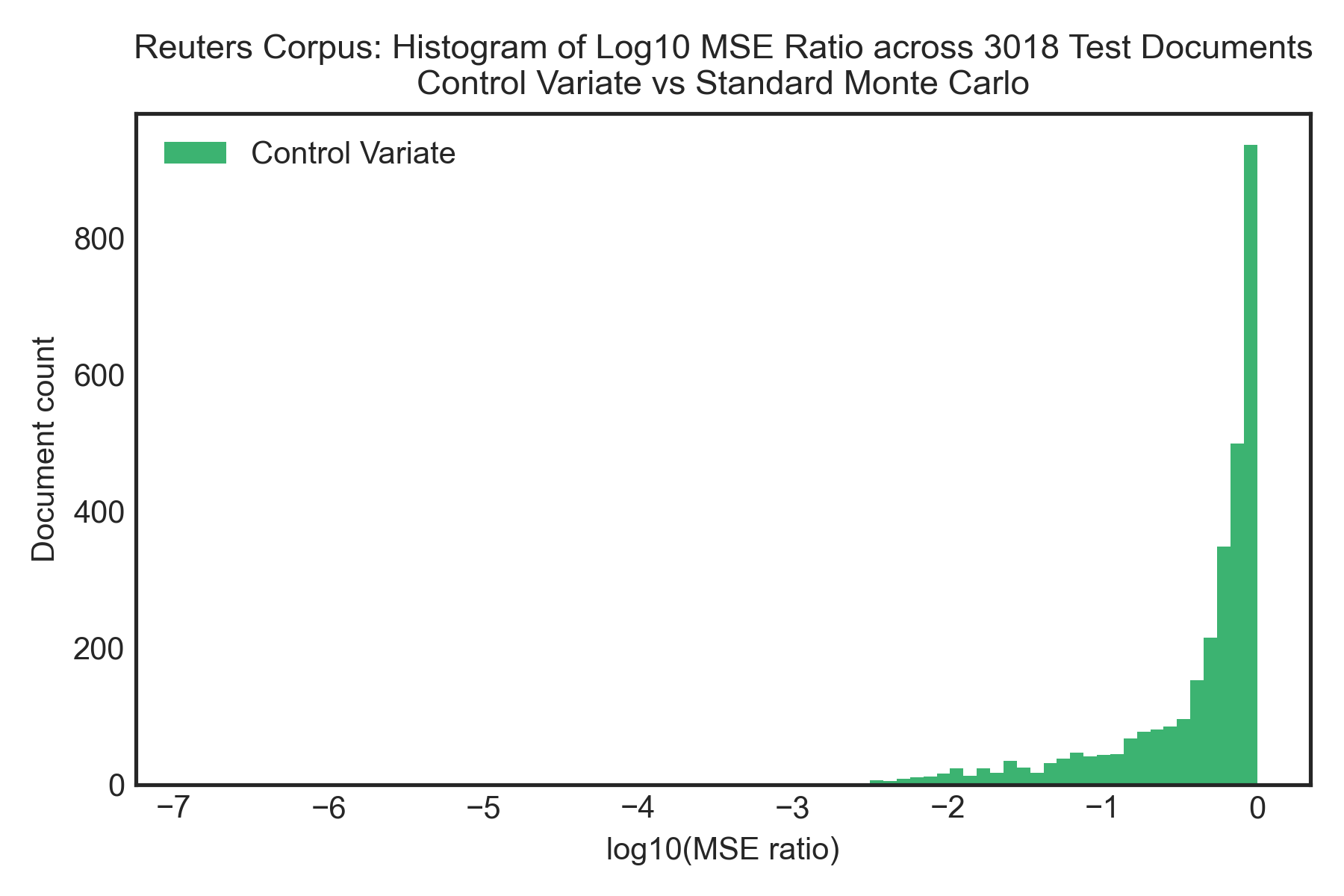}
\hfill
    
    \caption{\label{fig:hist_imp_sampling_reuters} Reuters-21578 text corpus. Left: histogram of the log MSE ratio $\log(\text{MSE}(\hat{p}_{\text{IS}}^{\gamma})/\text{MSE}(\hat{p}_{\text{MC}}))$ of the importance sampling estimator with $\gamma=0.9$ and $\epsilon=0.1$ across 3,018 
    test documents using $K=10$ topics. Right: histogram of the log MSE ratio of the control variate estimator $\log(\text{MSE}(\hat{p}_{\text{CV}})/\text{MSE}(\hat{p}_{\text{MC}}))$ across the same test documents. 
    }
    
\end{figure}

\section{Conclusion}
\label{sec:conclusion}

In this work, we study variance reduction techniques for estimating expectations of the form $\mathbb{E}[\exp(nH(\theta))]$ under Dirichlet distributions. We propose an importance sampling and control variate estimator and analyze their statistical efficiency for large $n$. Our analysis is based on novel extensions of the Laplace method for sparse maximizers $\theta^*$ that illustrate how sparsity influences the constant and polynomial terms in the asymptotics, which inform the asymptotic mean-squared error of the estimators.

Our analysis shows that these estimators are capable of substantial variance reduction compared to the plain Monte Carlo estimator for large $n$. The importance sampling estimator reduces the relative mean-squared error from $\Theta(n^{c})$ to $\Theta(n^{(1-\gamma)c})$ for any $\gamma \in (0,1)$ (where $c>0$ is a problem-specific exponent), achieving near-bounded relative error. We show that the control-variate estimator, guaranteed to achieve variance reduction, reduces the MSE by constant factor based on the Hessian of $H$.

\appendix
\label{appendix}
\section{Appendix}

\subsection{The Projected Simplex}
\label{subsec:appendix_assumptions}

In light of the discussion in Section \ref{subsec: definitions}, the expectation with respect to $\text{Dir}_\alpha$ can be written as \[\mathbb{E}_{\text{Dir}_{\alpha}}[f(\theta)] =\int_{\tilde{\Delta}_{K-1} } f(T(y)) \text{Dir}_{\alpha}^{(K-1)}({y}) \, d{y}\]
where 
$T:\tilde{\Delta}_{K-1} \to {\Delta}_{K-1}$ is the coordinate map (bijective and affine)
\begin{equation} 
\label{eq:T_map} 
T(y):=(y_{1},\dots,y_{K-1},1-\sum_{i=1}^{K-1}y_{i})=Ay+b,\quad A=\left[\begin{array}{c} I_{K-1}\\ -\mathbf{1}^{\top}_{K-1} \end{array}\right],\quad b=e_{K} 
\end{equation} 
where the $\mathbf{1}_{K-1}$ is the column vector of ones in $\mathbb{R}^{K-1}$. 
Clearly,
\begin{equation}
\label{eq:A_kernel}
    A^\top \mathbf{1}_K=0.
\end{equation}

Note that $T$ is an affine bijection between the truncated simplex defined in \eqref{eq:truncated_simplex} and the projected truncated simplex
\begin{equation}
\label{eq:epsilon_projected_simplex}
\tilde{\Delta}_{K-1}^\epsilon 
= \left\{ y \in \tilde{\Delta}_{K-1} \,\middle|\, 
y_i \ge \epsilon \enspace \text{for all } i\le K-1 \text{ such that } \theta^*_i>0 \right\}.
\end{equation}

\subsection{Analysis of Laplace Method on the Simplex}

We state an auxiliary lemma which will be used to establish an integrable bound in the proof of the Laplace method. 

\subsubsection{Bound around Maximum Lemma}

Under (A1)-(A4), the composed function $H(T(\cdot))$ on the projected simplex can be bounded above by a decreasing envelope centered around the maximizer in the projected simplex. If $\theta^*$ lies at the boundary of the simplex, this envelope is a linear quadratic expression that decays linearly in the first $m$ (active) coordinates and quadratically in the remaining coordinates. This follows from the first and second-order optimality conditions at $\theta^*$ and serves as the upper bound required for the Laplace method on the simplex.  

\begin{lem}
\label{lemma: bounds on maximum}
Suppose $H:\Delta_{K}\to\mathbb{R}$ satisfies assumptions (A1)-(A4) at the maximizer $\theta^*$. Let $m$ be the number of zero components of $\theta^*$ and let $T:\tilde{\Delta}_{K-1}\to \Delta_{K-1}$ be the coordinate map defined in \eqref{eq:T_map}.
Then there exists $c_1, c_2>0$ such that for any $y\in \tilde{\Delta}_{K-1}$,
\begin{equation}
\label{eq:H_hess_bound} 
H(T(y)) \le  H(\theta^*) -c_{1} \sum_{i=1}^{m}|y_i-{\theta}^*_i| - c_{2}\sum_{i=m+1}^{K-1}|y_i-{\theta}^*_i|^2.
\end{equation}
\end{lem}

\begin{remark}
    Linear terms appear in the envelope due to the gradient not being zero at the boundary of the simplex. This lemma is used in the proof of the Laplace theorem (Theorem \ref{thm:boundary-laplace_standard}) to show an integrable upper bound.
\end{remark}

\begin{proof}
    Refer to the Supplementary Material \ref{sec:bound_around_maximum_lemma}
\end{proof}

\subsubsection{ Proof of Theorem \ref{thm:boundary-laplace_standard} (Laplace Method on the Simplex)}
\label{sec:appendix-proof-boundary-laplace_standard}

We first introduce some identities and notation. Recall the equivalent formulation of (A4) discussed in \eqref{eq:reduced_hessian}. 
The matrix $U$ in \eqref{eq:U_explicit}
satisfies the identity
\begin{equation}
\label{eq:U_basis}
U = A V_m \in \mathbb{R}^{K\times (K-1-m)}, 
\qquad
V_m =
\begin{bmatrix}
0_{m\times (K-1-m)} \\
I_{K-1-m}
\end{bmatrix},
\end{equation}
with $A$ the gradient of map $T$ defined in \eqref{eq:T_map}. When $\theta^*$ is an interior point $(m=0)$, $U$ coincides with $A$. When $\theta^*$ is a boundary point $(m>0)$, the matrix $V_m$ extracts columns of $A$ associated with strictly positive components of $\theta^*$, i.e., the last $K-1-m$ columns of $A$. Each extracted column spans a tangent direction in the critical cone at $\theta^*$, which determines the quadratic contribution in the Laplace approximation.

Furthermore for any $1 \le p \le K-1$, define 
\begin{equation}
\label{eq:projection_matrix}
P_p = \sum_{j=1}^p e_j e_j^\top \in \mathbb{R}^{(K-1)\times(K-1)},
\quad e_j \in \mathbb{R}^{K-1}
\end{equation}
which is the orthogonal projection onto the first $p$ coordinate directions. Together with $V_m$, we have the identity
\begin{equation}
    \label{eq:VV_I_P_identity}
    V_m V_m^\top = I_{K-1}-P_m
\end{equation}
which will be used to rewrite an integral later. 

We also introduce the partial Dirichlet factor
\begin{equation} 
\text{Dir}_{{\alpha},m}(\theta)  =\frac{1}{B(\alpha)}\prod_{j=m+1}^{K}(\theta_j)^{\alpha_{j}-1} ,\quad  \theta \in \Delta_{K-1}.
\end{equation}
Evaluated at $\theta^*$, this is the subproduct of $\text{Dir}_\alpha (\theta^*)$ over the inactive coordinates. These will be used in the limiting argument.

\begin{proof}[Proof of Theorem \ref{thm:boundary-laplace_standard}]
The outline of the proof is as follows.

 \begin{enumerate}[itemsep=0pt, topsep=0pt, parsep=0pt, partopsep=0pt]
 
  \item We restrict the integration domain to the truncated simplex which  contains $\theta^*$, as this gives the same asymptotic rate by a localization lemma. We re-parameterize the integrating variable around $\theta^*$. 
\item We apply a scaling factor on the integrand to obtain a pointwise limit.
  \item We leverage Lemma~\ref{lemma: bounds on maximum} (bound around maximum) and domain truncation to obtain an integrable upper bound of the integrand. 
  \item By the Dominated Convergence Theorem, we get the desired limit.
\end{enumerate}

\paragraph{1. Reparameterization}

By the Localization lemma (cf. Supplementary Material \ref{sec:localization_lemma}), we can restrict the integration domain to the truncated simplex $\Delta_{K-1}^\epsilon$ in \eqref{eq:truncated_simplex} . Using the reparameterization $\beta=n^{1/2}$ and map $T$ defined in \eqref{eq:T_map}, we can write $\mathbb{E}_{\text{Dir}_\alpha} \left[ e^{nH({\theta})}\right]$ as
\begin{equation}
\label{eq: main_integral_rewritten}
    I(\beta)=\int_{\Omega} e^{\beta^2 H(T(x))}
    \text{Dir}_{\alpha}^{(K-1)}\left(T(x)\right)
    \mathbf{1}\{x\in\tilde{\Delta}_{K-1}^\epsilon\} dx
\end{equation}
where $\Omega\equiv[0,\infty)^{m}\times\mathbb{R}^{K-1-m}\subseteq\mathbb{R}^{K-1}$ and $\tilde{\Delta}_{K-1}^{\epsilon}$ is the projected truncated simplex defined in \eqref{eq:epsilon_projected_simplex}.
Define $\ttheta=T^{-1}(\theta^*)$ to be the corresponding maximizer in the projected simplex $\tilde{\Delta}_{K-1}$. For any fixed $\beta>0$, we apply the following change of variables around $\tilde{\theta}^*$, $$x(u)=h_{\beta}(u)+\tilde{\theta}^*$$ where
\begin{align*}
h_{\beta}(u) & =\begin{cases}
\beta^{-2}u_{k}, & \text{if }k=1,...,m;\\
\beta^{-1}u_{k}, & \text{otherwise}.
\end{cases} 
\end{align*}
This change of variables can be expressed with the shorthand notation 
$$
h_{\beta}(u)=(\beta^{-2}P_m+\beta^{-1}({I}_{K-1}-P_m))u,
$$ 
with $P_m$ as defined in \eqref{eq:projection_matrix}.
With this, \eqref{eq: main_integral_rewritten} can be re-written as
\begin{align}
    I(\beta) 
    &=\beta^{-(K-1+m)}  \int_{\Omega} e^{\beta^2 H(T(h_{\beta}(u)+\tilde{\theta}^*))}
    \text{Dir}_{\alpha}^{(K-1)}\left(T(h_{\beta}(u)+\tilde{\theta}^*)\right)
    \mathbf{1}\{h_{\beta}(u)+\tilde{\theta}^*\in\tilde{\Delta}_{K-1}^\epsilon\} du \nonumber \\
    &= \beta^{-(K-1+m)} \beta^{-\sum_{k=1}^{m}2(\alpha_{k}-1)} \exp \left(\beta^{2}H(T(\ttheta)) \right)\int_{\Omega}  g_{\beta}(u)\mathbf{1}\{h_{\beta}(u)+\tilde{\theta}^*\in\tilde{\Delta}_{K-1}^\epsilon \} du
\end{align}
where 
\[
g_{\beta}(u)\equiv\underbrace{\beta^{\sum_{k=1}^{m}2(\alpha_{k}-1)} \text{Dir}_{\alpha}^{(K-1)}\left(T(h_{\beta}(u)+\tilde{\theta}^*)\right)}_{(a)}\underbrace{\exp{ \left( \beta^2   H(T(h_{\beta}(u)+\tilde{\theta}^*))- \beta^2 H(T(\ttheta))  \right)} }_{(b)}.
\]
Clearly $\mathbf{1}\{h_{\beta}(u)+\tilde{\theta}^*\in\tilde{\Delta}_{K-1}\}\to 1$ as $\beta\to1$ since for any $u\in\mathbb{R}^{K-1}$, $\hb \to \ttheta $. 

The rest of the proof proceeds as follows. We first establish the pointwise
limit of the integrand $g_{\beta}(u)$. Then we show an integrable upper bound of $g_{\beta}(\cdot)\mathbf{1}\{h_{\beta}(u)+\tilde{\theta}^*\in\tilde{\Delta}_{K-1}\}$ and apply the Dominated Convergence Theorem.

\paragraph{2. Pointwise limit of $g_{\beta}(u)$.} 

For the point-wise limit of (a), 
\begin{align*}
 & \lim_{\beta\to\infty}\beta^{\sum_{k=1}^{m}2(\alpha_{k}-1)} \text{Dir}_{\alpha}^{(K-1)}\left(T(h_{\beta}(u)+\tilde{\theta}^*)\right)\\
 & =\frac{1}{B(\alpha)}\prod_{k=1}^{m}u_{k}^{\alpha_{k}-1}\prod_{j=m+1}^{K-1}(\ttheta_j)^{\alpha_{j}-1}(1-\sum\nolimits_{k=1}^{K-1} \ttheta_k)^{\alpha_{K}-1}\\
 & =\frac{1}{B(\alpha)}\prod_{k=1}^{m}u_{k}^{\alpha_{k}-1}\text{Dir}_{\alpha,m}(\rtheta)\\
 & >0.
\end{align*}
 For the exponential term (b), we first note that $A$, the gradient of the map $T$, satisfies 
 $$A^\top \nabla H(\theta^*)= -\sum_{i=1}^{m} \lambda_i e_i$$
 which follows from $A^\top 1_K=0$ in \eqref{eq:A_kernel} and the KKT equation \eqref{eq:KKT_conditions}. Since $P_m$ is a orthogonal projection that keeps only first $m$ coordinates, we have that
$$ P_m A^\top\nabla H(\rtheta)=A^\top\nabla H(\rtheta) \quad \mbox{and} \quad (I_{K-1}-P_m)A^\top\nabla H(\rtheta)=0.$$
With this, by differentiability of $H$ near $\theta^*$ and L'Hospital's rule: 
\begin{align*}
 & \lim_{\beta\to\infty}\beta^{2} \left(H(T(\hb))-H(T(\ttheta))\right)\\
 & =u^{\top}P_m A^{\top}\nabla H(T(\ttheta))+\lim_{\beta\to\infty}\frac{u^{\top}(I_{K-1}-P_m )A^{\top}\nabla_{\theta}H(T(\hb)))}{2\beta^{-1}} \\
&= u^{\top} A^{\top}\nabla H(\rtheta)+\frac{1}{2}u^{\top}(I_{K-1}-P_m)A^{\top}\nabla^{2}H(\rtheta)A(I_{K-1}-P_m)u.
\end{align*}
Thus, we have that 
\begin{eqnarray}
\lefteqn{
\lim_{\beta\to\infty}\exp \left(\beta^{2}H(T(\hb))-\beta^{2}H(T(\ttheta)) \right) 
} && \nonumber \\
&=&\exp\left(u^{\top} A^{\top}\nabla H(\rtheta) \right)
\cdot \exp \left(\frac{1}{2}u^{\top}(I_{K-1}-P_m)A^{\top}\nabla^{2}H(\rtheta)A(I_{K-1}-P_m)u\right).
\label{eq:quadratic_fluctuation}
\end{eqnarray}

\paragraph{3. Integrable bound for $g_{\beta}$.} Next we will show that there exists a integrable function $G(u)\geq0$ such that for all $u\in\Omega$
\[
g_{\beta}(u)\mathbf{1}\{\hb\in\tilde{\Delta}_{K-1}^\epsilon \}\leq G(u).
\]

For the bound on (a), we first fix the notation
\begin{equation}
\label{eq:reduced_s_b}
    S_\beta := \{u \in \Omega : h_{\beta}(u)+\tilde{\theta}^*\in \tilde{\Delta}_{K-1}^\epsilon\}.
\end{equation}
On $S_\beta$, the vector $T(h_{\beta}(u)+\tilde{\theta}^*)$ lies in the truncated simplex $\Delta_{K-1}^\epsilon$, which restricts each component indexed by $k\ge m+1$ to lie between $\epsilon$ and $1$. Hence, for every such $k$,
$$
\left(T(h_\beta(u)+\tilde{\theta}^*)_k\right)^{\alpha_k-1}
\le 
\begin{cases}
1, & \alpha_k \ge 1,\\[4pt]
\epsilon^{\alpha_k-1}, & \alpha_k < 1.
\end{cases}
$$
Therefore the following partial product of the last $K-m$ components in $\text{Dir}_{\alpha}^{(K-1)}\left(T(h_{\beta}(u)+\tilde{\theta}^*)\right)$ can be upper bounded by
\begin{align}
    \prod_{j=m+1}^{K-1}(\beta^{-1}u_{j}+\theta_{j}^{*})^{\alpha_{j}-1}(1-\sum_{k=1}^{m}(\beta^{-2}u_{k})-\sum\nolimits_{j=m+1}^{K-1}(\beta^{-1}u_{j}+\theta_{j}^{*}))^{\alpha_{K}-1} 
  \le \prod_{k=m+1}^{K}\max (1,\epsilon^{\alpha_k -1})
\end{align}
Thus we have the following bound on $(a)$:
\begin{equation}
\label{eq:upperbound_1}
\beta^{\sum_{k=1}^{m}2(\alpha_{k}-1)} \text{Dir}_{\alpha}^{(K-1)}\left(T(h_{\beta}(u)+\tilde{\theta}^*)\right)\le \frac{1}{B(\alpha)}	\prod_{k=1}^{m}u_{k}^{\alpha_{k}-1} \prod_{k=m+1}^{K}\max (1,\epsilon^{\alpha_k -1}).
\end{equation}

For the bound on (b), the bound around maximum lemma ( \ref{lemma: bounds on maximum}) gives constants $C_1,C_2>0$ such that
\begin{align}
\label{eq:upperbound_2}
 \beta^{2}\left(H(T(\hb))-H(T(\ttheta))\right)\nonumber &\le \beta^{2}\left( -C_{1}||P_m h_{\beta}(u)||_{1} -C_{2}||(I_{K-1}-P_m)h_\beta (u)||_{2}^{2} \right) \\
&= -C_1|| P_m u ||_1
- C_2\|(I_{K-1}-P_m) u \|_{2}^{2}.  
\end{align}
Putting the two bounds \eqref{eq:upperbound_1} and \eqref{eq:upperbound_2} together, we have 
\begin{align*}
\label{eq:int_F_b}
 & g_{\beta}(u)\mathbf{1}\{h_{\beta}(u)+\tilde{\theta}^*\in\tilde{\Delta}_{K-1}^\epsilon \} \\
    & \leq \frac{1}{B(\alpha)}	\prod_{k=1}^{m}u_{k}^{\alpha_{k}-1} \prod_{k=m+1}^{K}\max (1,\epsilon^{\alpha_k -1})\exp(-C_1 \|P_m u\|_{1})\exp(-C_2\|(I_{K-1}-P_m)u\|_{2}^{2}) 
\end{align*}
which is separable and integrable on $\Omega$.
Integrability holds since for any $C_1,C_2>0$ and $\alpha_k>0,\quad k=1,\dots,m$
\begin{equation}
\label{eq:Gammabound_1}
\int_{0}^{\infty} u^{\alpha_k-1}\exp\!\left(-C_1\,u\right)\,du < \infty   
\end{equation}
and 
\begin{equation}
\label{eq:Gammabound_2}
    \int_{\mathbb{R}^{K-1-m}} \exp\!\left(-C_2\,\|v\|_2^{\,2}\right)\,dv < \infty
\end{equation}
and by applying Tonelli's Theorem. 

\paragraph{4. Limit for $I(\beta)$.} By the Dominated Convergence Theorem, we have that
\begin{align*}
 & \lim_{\beta\to\infty}\beta^{(K-1+m)}\beta^{\sum_{k=1}^{m}2(\alpha_{k}-1)}e^{-\beta^{2}H(T(\ttheta))}I(\beta)\\
  &=   \left(\int_{\Omega} \lim_{\beta\to\infty} g_{\beta}(u)\mathbf{1}\{h_{\beta}(u)+\ttheta\in\tilde{\Delta}_{\epsilon}\}du\right)\\
 & =\frac{\prod_{j=m+1}^{K}(\theta^*_j)^{\alpha_{j}-1}}{B(\alpha)}\int_{\Omega}\prod_{k=1}^{m}u_{k}^{\alpha_{k}-1}\exp\left(u^{\top} A^{\top}\nabla H(\rtheta)+\frac{1}{2}u^{\top}(I_{K-1}-P_m)A^{\top}\nabla^{2}H(\rtheta)A(I_{K-1}-P_m)u\right)du.
\end{align*}

We can further simplify the last integral by using the identity $V_m V_m^\top = I_{K-1}-P_m $ discussed in \eqref{eq:VV_I_P_identity}.
Noting that $V_m^\top: \mathbb{R}^{K-1}\to \mathbb{R}^{K-1-m}$ drops the first $m$ coordinates, 
and recalling that $A^\top \nabla H(\theta^*)= -\sum_{i=1}^{m} \lambda_i e_i$,
we can write
\begin{align}
    \label{eq:intergral_factorization}
&\int_{\Omega}\prod_{k=1}^{m}u_{k}^{\alpha_{k}-1}
    \exp\left(
      u^{\top} A^{\top}\nabla H(\rtheta)
      +\frac{1}{2}u^\top(I_{K-1}-P_m)A^{\top}\nabla^{2}H(\rtheta)A(I_{K-1}-P_m)u
    \right)du \nonumber \\
&=\int_{\Omega}\prod_{k=1}^{m}u_{k}^{\alpha_{k}-1}\exp\left(u^{\top} A^{\top}\nabla H(\rtheta)+\frac{1}{2}(V_m^\top u)^\top (V_m^\top A^{\top}\nabla^{2}H(\rtheta)AV_m) (V_m^\top u)\right)du \nonumber \\
=&\int_{[0,\infty)^m}\prod_{k=1}^{m}u_{k}^{\alpha_{k}-1}
   \exp\left(-\sum_{i=1}^m \lambda_i u_i \right) \,du \times 
  \int_{\mathbb{R}^{K-1-m}}
   \exp \left(\frac{1}{2}\xi^\top (U^{\top}\nabla^{2}H(\rtheta)U)\xi\right) d\xi,
\end{align}
with $\xi := V_m^\top u \in \mathbb{R}^{K-1-m}$.
If $m\le K-2$, evaluating the integrals gives the constant
\begin{equation}
\label{eq:C_H_constant}
C_{H}\equiv
\text{Dir}_{{\alpha},m}(\theta^*)  \prod_{k=1}^{m}\lambda_{k}^{-\alpha_{k}}\Gamma(\alpha_{k})
\frac{(2\pi)^{\frac{K-1-m}{2}}}{|\det(U^\top \nabla^2 H(\theta^*)U)|^{1/2}}
\end{equation}    
where $U^\top \nabla^2 H(\theta^*)U$ is the reduced Hessian discussed in \eqref{eq:reduced_hessian}. 

If $m = K-1$, then the last integral in \eqref{eq:intergral_factorization} is absent
and the constant reduces to
\begin{equation}
\label{eq:C_H_constant_alternative}
C_{H}\equiv
\frac{1}{B(\alpha)}{\theta^*_{K}}^{\alpha_{K}-1}  \prod_{k=1}^{K-1}\lambda_{k}^{-\alpha_{k}}\Gamma(\alpha_{k}).
\end{equation}

\end{proof}

\subsection{Results involving $\KL$ Divergence} 

\subsubsection{Derivatives of $\KL$ Divergence}
\label{sec:derivatives_KL}

We first note several facts about the derivatives of $\KL$  defined in \eqref{eq:KL_divergence}. 
The gradient is
\begin{equation}
\label{eq:KL_gradient}
\nabla_\theta \KL(\theta^* |\theta)
=
-\left(
\frac{\theta_1^*}{\theta_1},\dots,\frac{\theta_K^*}{\theta_K}
\right)^\top, \qquad
\theta \in \Delta_{K-1}.
\end{equation}
The Hessian is given by
\begin{equation}
\label{eq:KL_Hessian}
\nabla_\theta^2 \KL(\theta^* |\theta)
=
\operatorname{diag}\!\left(
\frac{\theta_1^*}{\theta_1^2},\dots,
\frac{\theta_K^*}{\theta_K^2}
\right),
\qquad
\theta \in \Delta_{K-1}.
\end{equation}
For indices $i$ such that $\theta_i^* = 0$, the corresponding entries in
the gradient and Hessian are defined to be zero.

\subsubsection{Proof of Theorem ~\ref{thm:boundary-laplace-with-kl} (Laplace Method with KL Factor)}
\label{sec:appendix-proof-boundary-laplace-with-kl}

\begin{proof}[Proof of Theorem \ref{thm:boundary-laplace-with-kl}]

The proof is similar to the proof of Theorem \ref{thm:boundary-laplace_standard} except that $g_\beta (u)$ has an extra factor. Using the reparameterization $\beta=n^{1/2}$, we can write $\mathbb{E}_{\text{Dir}_\alpha} \left[ e^{2nH({\theta})+n^{\gamma }\text{KL}(\theta^{*}|\theta)}\mathbf{1}_{\Delta_{K-1}^\epsilon}(\theta)\right]$ as
\begin{equation}
\label{eq: KL_integral_rewritten}
    I(\beta)=\int_{\Omega} e^{2\beta^2 H(T(x))+\beta^{2\gamma}\text{KL}\left (T(\ttheta)|T(x) \right)}
    \text{Dir}_{\alpha}^{(K-1)}\left(T(x)\right)
    \mathbf{1}\{x\in\tilde{\Delta}_{K-1}^\epsilon\} dx
\end{equation}
where 
$\Omega\equiv[0,\infty)^{m}\times\mathbb{R}^{K-1-m}\subseteq\mathbb{R}^{K-1}$. Let $\ttheta=T^{-1}(\theta^*)$. With the change of variables $x(u)=h_{\beta}(u)+\tilde{\theta}^*$, we rewrite $I(\beta)$ as
\begin{align*}
    I(\beta) 
    &= \beta^{-(K-1+m)} \beta^{-\sum_{k=1}^{m}2(\alpha_{k}-1)} \exp \left(2\beta^{2}H(T(\ttheta))  \right)\int_{\Omega}  g_{\beta}(u)\mathbf{1}\{h_{\beta}(u)+\tilde{\theta}^*\in\tilde{\Delta}_{K-1}^\epsilon \} du
\end{align*}
where 
\begin{align*}
g_{\beta}(u) &\equiv\underbrace{\beta^{\sum_{k=1}^{m}2(\alpha_{k}-1)} \text{Dir}_{\alpha}^{(K-1)}(T(h_{\beta}(u)+\tilde{\theta}^*))}_{(a)}\cdot \underbrace{\exp{ \left( 2\beta^2  H(T(h_{\beta}(u)+\tilde{\theta}^*))-2\beta^2 H(T(\ttheta))  \right)} }_{(b)} \\ &\cdot \underbrace{\exp\left(\beta^{2\gamma}\text{KL}\left (T(\ttheta)|T(h_{\beta}(u)+\tilde{\theta}^*) \right)\right)}_{(c)}.
\end{align*}
The rest of the proof goes as follows.
\begin{enumerate}[itemsep=0pt, topsep=0pt, parsep=0pt, partopsep=0pt]
  \item  We show that the factor (c) has limit value of 1 as $\beta\to \infty$. 
  \item  We show an upper bound of the factor (c) on  $S_\beta$ defined in \eqref{eq:reduced_s_b}.
  \item We combine with the previous upper bound in the proof of Theorem~\ref{thm:boundary-laplace_standard} and show an integrable upper bound of $g_{\beta}(u) \mathbf{1}\{ u \in S_\beta\}$. 
  \item By the Dominated Convergence Theorem, we get the desired limit for $I(\beta)$.
\end{enumerate}

\paragraph{1. Pointwise limit of (c)}
    We first show that the pointwise limit of factor $(c)$ is $1$, which yields the same pointwise limit for
$g_{\beta}(u)\mathbf{1}\{h_{\beta}(u)+\tilde{\theta}^*\in\tilde{\Delta}_{K-1}\}$
in Theorem~\ref{thm:boundary-laplace-with-kl}.
By L'Hospital's rule,
\begin{equation}\label{eq:KL_lhopital_conclusion_start}
\lim_{\beta\to\infty}
\frac{\KL\!\left(T(\ttheta)\,\middle|\,T(h_\beta(u)+\ttheta)\right)}{\beta^{-2\gamma}}
=
-\frac{1}{2\gamma}\,
\lim_{\beta\to\infty}\beta^{2\gamma+1}\,
\frac{d}{d\beta}\KL\!\left(T(\ttheta)\,\middle|\,T(h_\beta(u)+\ttheta)\right).
\end{equation}
For the derivative of \KL, we can write
\[
\frac{d}{d\beta}\,\KL\!\left(T(\ttheta)\,\middle| \,T(h_\beta(u)+\ttheta)\right)
=
\beta^{-2}B_\beta(u)+R_\beta(u),
\qquad
R_\beta(u)=O(\beta^{-3}).
\]
where 
\[
B_\beta(u)
:=
\sum_{k=m+1}^{K-1}
\frac{\tilde{\theta}_k^*u_k}{\beta^{-1}u_k+\tilde{\theta}_k^*}
-
\sum_{k=m+1}^{K-1}
\frac{(1-\mathbf {1}^\top\tilde{\theta}^*)u_k}{D_\beta(u)}
\]
and
\[
D_\beta(u):=
1-\sum_{i=1}^{m}\beta^{-2}u_i-\sum_{r=m+1}^{K-1}\bigl(\beta^{-1}u_r+\tilde{\theta}_r^*\bigr).
\]
Since $B_\beta(u)=O(\beta^{-1})$ (shown in the Supplemental Material \ref{sec:KL_derivative}), it follows that 
\begin{equation}\label{eq:KL_lhopital_conclusion_rate}
\beta^{2\gamma+1}\left( \beta^{-2}B_\beta(u)+R_\beta(u)
\right)=O(\beta^{2\gamma-2}).
\end{equation}
Therefore \eqref{eq:KL_lhopital_conclusion_start} satisfies
\[
\lim_{\beta\to\infty}
\frac{\KL\!\left(T(\ttheta)\,\middle| T(h_\beta(u)+\ttheta)\right)}{\beta^{-2\gamma}}
=0,
\qquad \text{whenever } 0<\gamma<1.
\]
Hence the point-wise limit of $(c)$ is 
\begin{align*}
\lim_{\beta\to\infty}\exp\left(\beta^{2\gamma}\text{KL}(T(\ttheta)|T(h_{\beta}(u)+\tilde{\theta}^*))\right)
&= \exp\left( \lim_{\beta\to\infty} \beta^{2\gamma}\text{KL}(T(\ttheta)|T(h_{\beta}(u)+\tilde{\theta}^*))\right) \\
 &=1.
\end{align*}

\paragraph{2. Upper bound on (c) in $\{ u \in \Omega: h_{\beta}(u)+\tilde{\theta}^*\in \tilde{\Delta}_{K-1}^\epsilon \}$.} 

We show an upper bound of $(c)$ on $S_\beta = \{ u \in \Omega: h_{\beta}(u)+\tilde{\theta}^*\in \tilde{\Delta}_{K-1}^\epsilon \}$.

On the truncated simplex $\Delta_{K-1}^\epsilon$, the Hessian of the
$\KL$ divergence in \eqref{eq:KL_Hessian} admits the uniform bound
\begin{equation}
\label{eq:KL_truncated_bound}
\nabla_\theta^2 \KL(\theta^*|\theta)
\preceq M I_K,
\qquad
\theta \in \Delta_{K-1}^\epsilon
\end{equation}
where
$$M := \max_{i\ge m+1}\frac{\theta_i^*}{\epsilon^2}$$
which follows from $\theta_i\ge \epsilon$ for all
$i\ge m+1$ on $\Delta_{K-1}^\epsilon$. Since the gradient of map $T$ is $A$ and $T(\ttheta)=\theta^*$,  
its Hessian in the projected coordinates is
$$
\nabla_y^2 \KL(T(\ttheta)|T(y))
=
A^\top \left.\nabla_\theta^2 \KL(\theta^*| \theta)\right|_{\theta=T(y)} A.
$$
Combined with the upperbound \eqref{eq:KL_truncated_bound} in $\mathbb{R}^K$, we have the upper bound in $\mathbb{R}^{K-1}$
$$\nabla_y^2 \KL(T(\ttheta)|T(y))
\preceq
M A^\top A.$$
We can further bound this by noting the properties of $A^\top A$
$$A^\top A = I_{K-1} + \mathbf{1}_{K-1} \mathbf{1}_{K-1} ^\top,
\qquad
\lambda_{\max}(A^\top A)=K,$$
from which we have
\begin{equation}
\label{eq:KL_projected_bound}
\nabla_y^2 \KL(T(\ttheta)|T(y))
\preceq
(MK) I_{K-1}.
\end{equation}
Therefore, with the inequality (9.13) in \cite{boyd2004convex},
\begin{align}
\label{eq:KL_upperbound}
    \text{KL}(T(\ttheta)|T(y))	&\le\text{KL}(T(\ttheta)|T(\tilde{\theta^{*}}))+\left(\nabla_y \text{KL}(T(\ttheta)|T(y))|_{y=\ttheta}\right)^{\top}\left(T(y)-T(\ttheta)\right) \nonumber \\
    &+\frac{MK}{2}||y-T(\ttheta)||_2^{2} \nonumber \\
	&=\sum_{i=1}^{m}(y_{i}-\tilde{\theta}_{i}^{*})+\frac{MK}{2}||y-\tilde{\theta}^{*}||_2^{2} \nonumber \\
    &= \sum_{i=1}^{m}y_{i}+\frac{MK}{2}||y-\tilde{\theta}^{*}||_2^{2}
\end{align}
where the first equality follows from $\text{KL}(T(\ttheta)|T(\tilde{\theta^{*}}))=0$ and  
\begin{equation*}
    (A^\top \nabla_{\theta}\KL(\theta^* |\theta)\big|_{\theta=\theta^*})_j
=
\begin{cases}
1, & j=1,\dots,m,\\
0, & j=m+1,\dots, K-1
\end{cases}
\end{equation*}
which follows from the gradient expression in  \eqref{eq:KL_gradient}. 
Letting $y=h_{\beta}(u)+\ttheta$ in \eqref {eq:KL_upperbound} and using the fact that $(h_{\beta}(u)+\ttheta)_i\ge0$ for $i\le m$ on $\Omega$,
\begin{align}
\label{eq:KL_bound_1}
    \beta^{2\gamma}\text{KL}(T(\ttheta) \, | \, T(h_{\beta}(u)+\ttheta))  
    &\le\beta^{2(\gamma-1)}||P_{m}u||_1 + \beta^{2(\gamma-2)}\frac{MK}{2}||u||_2^{2}
\end{align}
Using the orthogonality relations $(P_{m}u)\perp(I-P_{m})u$, the second term in \eqref{eq:KL_bound_1} can be rewritten as
\begin{equation}
\label{eq:KL_triangle_inequality}
    \beta^{2(\gamma-2)}\frac{MK}{2}||u||_2^{2} =   \beta^{2(\gamma-2)}\frac{MK}{2}||P_{m}u||_2^{2}+\beta^{2(\gamma-1)}\frac{MK}{2}||(I_{K-1}-P_{m})u||_2^{2}.
\end{equation}
Note that if $u\in S_\beta$, then $h_{\beta}(u)+\ttheta \in \tilde{\Delta}_{K-1}$. Hence, for $i\le m$, $0\le \beta^{-2}u_i +\ttheta_i \le 1$. Since $\ttheta_i=0$, it follows that $0\le u_i\le \beta^2.$
This implies
\begin{equation}
\label{eq:S_beta_bound}
     \| P_m u\|_2^2 \le \| P_m u\|_1 \| P_m u\|_\infty \le \beta^2 \| P_m u\|_1 ,\quad u\in S_\beta,
\end{equation}
which multiplied with $\beta^{2(\gamma-2)} \frac{MK}{2}$ yields
\begin{equation}
\label{eq:KL_bound_2}
    \beta^{2(\gamma-2)}\frac{MK}{2}||P_{m}u||_2^{2} \le \beta^{2(\gamma-1)}\frac{MK}{2}||P_{m}u||_1.
\end{equation}
By combining the bounds \eqref {eq:KL_triangle_inequality} and \eqref{eq:KL_bound_2}, and then applying them in \eqref{eq:KL_bound_1}, we achieve the following upper bound of $(c)$ on $S_\beta$
\begin{equation}
\label{eq:upperbound_on_c}
    (c) \le  \exp \left( \beta^{2(\gamma-1)} \left(1+\frac{MK}{2} \right)||P_{m}u||_1 + \beta^{2(\gamma-1)}\frac{MK}{2}||(I_{K-1}-P_{m})u||_2^{2}\right),\quad u\in S_\beta.
\end{equation}

\paragraph{3. Integrable upper bound for $g_{\beta}(u)$.}
We will show that there exists $\beta_0<\infty$ and a $\beta$ independent function $G_* \in L^1 (\Omega)$
such that for all $\beta \ge \beta_0$
$$ g_{\beta}(u) \mathbf{1}\{ u \in \Omega: h_{\beta}(u)+\tilde{\theta}^*\in \tilde{\Delta}_{K-1}^\epsilon\} \le G_* (u),\quad \forall u\in \Omega $$
We combine \eqref{eq:upperbound_on_c} with the bound \eqref{eq:upperbound_2} in Theorem $\ref{thm:boundary-laplace_standard}$ proof (taking $2H$ instead). There exist positive constants $C_1,C_2>0$ such that
\begin{align*}
\label{eq:int_F_b_KL}
g_{\beta}(u)\mathbf{1}\{h_{\beta}(u)+\tilde{\theta}^*\in\tilde{\Delta}_{K-1}^{\epsilon}\}
&\le \frac{1}{B(\alpha)}	\prod_{k=1}^{m}u_{k}^{\alpha_{k}-1} \prod_{k=m+1}^{K}\max (1,\epsilon^{\alpha_k -1})   \\
&\times \exp \left(-C_1 \|P_m u\|_{1} +\beta^{2(\gamma-1)} \left(1+\frac{MK}{2} \right)||P_{m}u||_1 \right) \\
&\times \exp\left(-C_2\|(I_{K-1}-P_m)u\|_{2}^{2} + \beta^{2(\gamma-1)}\frac{MK}{2}||(I_{K-1}-P_{m})u||_2^{2} \right) \\
&:= G(\beta,u).
\end{align*}
Since $\gamma \in (0,1)$ and therefore $\beta^{2(\gamma-1)}\to 0$, there exists $\beta_0$ such that for all $\beta \ge \beta_0$
\begin{align*}
    (1+\frac{MK}{2}) \beta^{2(\gamma-1)} & \le \frac{C_1}{2} \\
  \frac{MK}{2}\beta^{2(\gamma-1)} &\le \frac{C_2}{2}.
\end{align*}
 This implies that when $\beta \ge \beta_0$, $G(\beta,u)$ is bounded by 
\begin{equation}
\label{eq:G_asterik_upperbound}
    G_* (u):= C \left( \prod_{k=1}^m u_{k}^{\alpha_k-1} \right) \exp \left(-\frac{C_1}{2} \|P_m u\|_{1} \right) \exp \left( -\frac{C_2}{2}\|(I_{K-1}-P_m)u\|_{2}^{2}\right)
\end{equation}
where $C>0$. 

The bound $G_*(u)$ is integrable on $\Omega$ since it factorizes over $[0,\infty)^m \times \mathbb{R}^{K-1-m}$. Using the integrability of each factor in \eqref{eq:Gammabound_1} and \eqref{eq:Gammabound_2} with $C_1/2$ and $C_2/2$, and Tonelli's Theorem, we have that $G_*\in L^1(\Omega)$.

\paragraph{4. Limit for $I(\beta)$.} 
Recall that we have the following pointwise limit results:
\begin{align*}
\beta^{-\sum_{k=1}^{m}2(\alpha_{k}-1)}\text{Dir}_{\alpha}^{(K-1)}(h_{\beta}(u)+\ttheta) & =\frac{1}{B(\alpha)}\prod_{k=1}^{m}u_{k}^{\alpha_{k}-1}\prod_{j=m+1}^{K}(\ttheta_{j})^{\alpha_{j}-1}(1-\sum\nolimits_{k=1}^{K-1}\ttheta_{k})^{\alpha_{K}-1}\\
\lim_{\beta\to\infty}e^{-2\beta^{2}H(T(\ttheta))}e^{\beta^{2}H(T(h_{\beta}(u)+\ttheta))} & =\exp\left(2u^{\top}PA^{\top}\nabla H(\theta^{*})+u^{\top}(I-P)A^{\top}\nabla^{2}H(\theta^{*})A(I-P)u\right)\\
\lim_{\beta\to\infty}e^{\beta^{2\gamma}\text{KL}(T(\ttheta)|T(h_{\beta}(u)+\ttheta))} & =1
\end{align*}
By the Dominated Convergence Theorem, we have that
\begin{align}
\label{eq:integral_factorization_2}
 & \lim_{\beta\to\infty}\beta^{(K-1+m)}\beta^{\sum_{k=1}^{m}2(\alpha_{k}-1)}e^{-2\beta^{2}H(T(\ttheta))}I(\beta)\nonumber \\
=&\left(\int_{\Omega} \lim_{\beta\to\infty} g_{\beta}(u)\mathbf{1}\{h_{\beta}(u)+\ttheta\in\tilde{\Delta}_{\epsilon}\}du\right) \nonumber \\
&=\frac{\prod_{j=m+1}^{K}(\theta^*_j)^{\alpha_{j}-1}}{B(\alpha)}\int_{[0,\infty)^m}\prod_{k=1}^{m}u_{k}^{\alpha_{k}-1}
   \exp\left(-2\sum_{i=1}^m \lambda_i u_i \right) \,du \times 
  \int_{\mathbb{R}^{K-1-m}}
   \exp \left(\xi^\top (U^{\top}\nabla^{2}H(\rtheta)U)\xi\right) d\xi
\end{align}
where $\xi := V_m^\top u \in \mathbb{R}^{K-1-m}$. If $m\le K-2$, evaluating the integrals gives the constant
\begin{equation}
\label{eq:C_H_constant_KL}
C^{\prime}_H\equiv
 \text{Dir}_{{\alpha},m}(\theta^*)  \prod_{k=1}^{m}(2\lambda_{k})^{-\alpha_{k}}\Gamma(\alpha_{k})
\frac{(\pi)^{\frac{K-1-m}{2}}}{|\det(U^\top \nabla^2 H(\theta^*)U)|^{1/2}}
\end{equation}    
If $m = K-1$, then the last integral in \eqref{eq:integral_factorization_2} is absent
and the constant reduces to
\begin{equation}
\label{eq:C_H_constant_alternative}
C^{\prime}_H \equiv
\frac{1}{B(\alpha)}{\theta^*_{K}}^{\alpha_{K}-1}  \prod_{k=1}^{K-1}(2\lambda_{k})^{-\alpha_{k}}\Gamma(\alpha_{k}).
\end{equation}  
We note that these constants coincide with the constant $C_H$ in Theorem \ref{thm:boundary-laplace_standard} applied to the second moment $\mathbb{E}[e^{2nH(\theta)}]$. 
\end{proof}

\subsection{Asymptotics of the Beta Function: An Application of Theorem \ref{thm:boundary-laplace_standard}}
\label{sec:reduced_KL}

Fix any $\gamma>0$. We will characterize the asymptotics of the Beta function $B(\alpha + n^\gamma \theta^*)$ by applying Theorem \ref{thm:boundary-laplace_standard} to $\mathbb{E}[e^{n\widehat{H}(\theta)}]$ where $\widehat{H}$ is defined in \eqref{eq:lowerbound_mu_incentive}. The analysis in this section will be used later in the analysis of the control variate based estimator in \eqref{sec: control_variate_analysis}. 

To apply Theorem \ref{thm:boundary-laplace_standard}, we need to verify conditions (A1)--(A4).

\paragraph{(A1)(A2) Maximizer and Differentiability} 
We note that $\widehat{H}(\theta)$ is uniquely maximized at $\theta^*$, implying (A1). This follows from the fact that $\text{KL}(\theta^*|\theta)\ge0$ and equals zero if and only if $\theta=\theta^*$. 
Next, by definition, $\nabla \widehat{H}= -\nabla \text{KL}(\theta^*|\theta)$ and $\nabla^2 \widehat{H}= -\nabla^2 \text{KL}(\theta^*|\theta)$. It is not hard to see that there exists an open neighborhood around $\theta^*$ in $\mathbb{R}^K$ on which $\KL(\theta^*|\theta)$ is twice differentiable. 
It is not continuous on $\Delta_{K-1}$ (if $\theta_i=0$ for some $i$ with $\theta^*_i>0$, then $\KL(\theta^*|\theta)=+\infty$, and hence  $\widehat{H}=-\infty$)
but it satisfies the strict gap condition \eqref{eq:continuity_relaxation}, which follows from the lower semi-continuity of $\KL(\theta^*|\theta)$.

\paragraph{(A3) KKT Multipliers} 

The problem of maximizing $\widehat{H}$ over the simplex admits KKT multipliers at $\theta^*$. Since
\begin{equation}
\label{eq:gradient_of_wide_hat_H}
\nabla \widehat{H}(\theta^*)=-\nabla_\theta \KL(\theta^* |\theta)\big|_{\theta=\theta^*}=\mathbf{1}_K-\sum_{i=1}^{m} e_i
\end{equation}
the KKT stationary condition $\nabla \widehat{H}(\theta^*) = \lambda\,\mathbf{1}_K - \mu$ is satisfied with $\lambda = 1$ and $\mu_k = 1 > 0$ for each $k$ with $\theta_k^* = 0$, establishing strict complementarity (A3).

\paragraph{(A4) Negative Definiteness} 

Let $U$ be the basis of the critical cone $\mathcal C(\theta^*)$ defined in \eqref{eq:U_basis}. The reduced Hessian of $\KL$ evaluated at $\theta^*$ has the expression
\begin{equation}
\label{eq:reduced_hessian_KL}
U^\top \nabla_\theta^2 \KL(\theta^* |\theta)\big|_{\theta=\theta^*} U
=
\operatorname{diag}\!\Bigl(
\frac{1}{\theta_{m+1}^*},\dots,\frac{1}{\theta_{K-1}^*}
\Bigr)
+
\frac{1}{\theta_K^*}\,
\mathbf{1}_{K-1-m}\,\mathbf{1}_{K-1-m}^\top.    
\end{equation}
Its determinant admits the explicit expression
\begin{align}
\label{eq:det_of_projected_KL}
\det\bigl(U^\top \nabla_\theta^2 \KL(\theta^*| \theta)\big|_{\theta=\theta^*} U\bigr)
&= \left(
\prod_{i=m+1}^{K-1} \frac{1}{\theta_i^*}
\right)
\left(
1 + \frac{1}{\theta_K^*}
\sum_{i=m+1}^{K-1} \theta_i^*
\right) \nonumber \\
&=\prod_{i=m+1}^{K}\frac{1}{\theta_i^*}
>0.
\end{align}
Since \eqref{eq:det_of_projected_KL} is positive, we can conclude that  $ U^\top \nabla_\theta^2 \KL(\theta^* |\theta)\big|_{\theta=\theta^*} U \succ 0$
and \begin{equation}
\label{eq:reduced_widehat}
    U^\top \nabla^2\widehat{H}(\theta^*)U \prec 0.
\end{equation} 
Therefore $\widehat{H}$ satisfies condition (A4). 

We can now apply Theorem \ref{thm:boundary-laplace_standard} with $n$ replaced by $n^\gamma$, yielding
\begin{equation}
\label{eq:main_widehat_H}
\mathbb{E}_{\text{Dir}_\alpha}\!\bigl[e^{n^\gamma\widehat{H}(\theta)}\bigr]
\;\sim\; C_B\;\cdot\; e^{n^\gamma H(\theta^*)}\;\cdot\; n^{-\gamma\frac{K-1-m}{2}}\;\cdot\; n^{-\gamma\sum_{i=1}^m\alpha_i}
\end{equation}
where
\[
C_B \;=\; \frac{(2\pi)^{\frac{K-1-m}{2}}}{B(\alpha)}\;\prod_{i=1}^m\Gamma(\alpha_i)\;\prod_{j=m+1}^K(\theta_j^*)^{\alpha_j-\frac{1}{2}}.
\]

\subsubsection{Beta Function Approximation}

Since $e^{-n^\gamma\KL(\theta^*|\theta)}=e^{-n^\gamma\theta^*\cdot\log\theta^*}\prod_{k=m+1}^K\theta_k^{n^\gamma\theta_k^*}$, we have that
\begin{align}
\mathbb{E}_{\mathrm{Dir}_\alpha}\!\left[e^{n^\gamma\widehat{H}(\theta)}\right]
&=
e^{n^\gamma H(\theta^*)}
\mathbb{E}_{\mathrm{Dir}_\alpha}\! \left[ e^{-n^\gamma\KL(\theta^*\mid\theta)} \right]
\nonumber\\
&=
e^{n^\gamma H(\theta^*)}e^{-n^\gamma\theta^*\cdot\log\theta^*}
\mathbb{E}_{\mathrm{Dir}_\alpha}\! \left[\prod_{k=m+1}^K \theta_k^{\,n^\gamma\theta_k^*} \right]
\nonumber\\
&=
e^{n^\gamma H(\theta^*)-n^\gamma\theta^*\cdot\log\theta^*}
\frac{B(\alpha+n^\gamma\theta^*)}{B(\alpha)}.
\end{align}

Combined with \eqref{eq:main_widehat_H}, we have the following lemma.
\begin{lem}[Beta Function Approximation]
\label{lem:beta-asymptotic}
Let $\theta^*$ be a point on $\Delta_{K-1}$ with $m$ zero components following the ordering assumption in \eqref{eq:ordering_assumption}. For any $\gamma>0$, as $n\to\infty$,
\[
B(\alpha+n^\gamma\theta^*)\sim C_{B}n^{-\gamma \frac{K-1-m}{2}}n^{-\gamma \sum_{i=1}^{m} \alpha_{i}}e^{n^\gamma \theta^* \cdot\log\theta^*}
\]
where $C_{B}=(2\pi)^{\frac{K-1-m}{2}}\prod_{i=1}^{m}\Gamma(\alpha_{i})\prod_{k=m+1}^{K}(\theta_{k})^{\alpha_{k}-\frac{1}{2}}>0$.
\end{lem}

\begin{remark} The result holds for any $\theta^*$ on the simplex (not just for a maximizer of $H$) since the choice of $\theta^*$ in the definition of $\widehat{H}$ is arbitrary.
The same result can proved directly by applying Stirling's formula. 

\end{remark}

This lemma will be important in quantifying the variance reduction achieved by the IS estimator.

\subsection{Analysis of MSE Reduction using Importance Sampling}

\subsubsection{Proof of Theorem \ref{thm:boundary-variance-is} (Variance reduction by the IS estimator)}
\label{sec:appendix-proof-boundary-variance-is}

\begin{proof}

The variance ratio is
\begin{align}
\label{eq:variance_ratio}
\frac{\text{Var}(\hat{p}_{\text{IS}}^{\gamma})}{\text{Var}(\hat{p}_{\text{MC}})} & =\frac{(1/N)\text{Var}_{\alpha+n^\gamma \theta^{*}}\left(e^{nH(\theta)}\frac{\text{Dir}_{\alpha}(\theta)}{\text{Dir}_{\alpha+n^\gamma\theta^{*}}(\theta)}\mathbf{1}_{\Delta_{K-1}^{\epsilon}}\right)}{(1/N)\text{Var}_{\alpha}\left(e^{nH(\theta)}\right)} \nonumber \\
 & =\frac{\mathbb{E}_{\text{Dir}_{\alpha+n^\gamma \theta^{*}}}\left[\left(e^{nH(\theta)}\frac{\text{Dir}_{\alpha}(\theta)}{\text{Dir}_{\alpha+n^\gamma\theta^{*}}(\theta)}\mathbf{1}_{\Delta_{K-1}^{\epsilon}}\right)^{2}\right]-\mathbb{E}_{\text{Dir}_{\alpha}}
 [e^{nH(\theta)}\mathbf{1}_{\Delta_{K-1}^{\epsilon}}]^{2}}{\mathbb{E}_{\text{Dir}_{\alpha}}\left[e^{2nH(\theta)}\right]-\mathbb{E}_{\text{Dir}_{\alpha}}[e^{nH(\theta)}]^{2}}
\end{align}
where the second moment of $\hat{p}_{\text{IS}}^{\gamma}$ can further expressed as
\begin{align}
\label{eq:second_moment_rewritten}
\mathbb{E}_{\text{Dir}_{\alpha+n^\gamma\theta^{*}}}\left[\left(e^{nH(\theta)}\frac{\text{Dir}_{\alpha}(\theta)}{\text{Dir}_{\alpha+n^\gamma\theta^{*}}(\theta)}\mathbf{1}_{\Delta_{K-1}^{\epsilon}}\right)^{2}\right]
&= \mathbb{E}_{\text{Dir}_{\alpha}}\left[e^{2nH(\theta)}\frac{\text{Dir}_{\alpha}(\theta)}{\text{Dir}_{\alpha+n^\gamma\theta^{*}}(\theta)}\mathbf{1}_{\Delta_{K-1}^{\epsilon}}\right] \nonumber \\ 
 & =\frac{B(\alpha+n^\gamma\theta^{*})}{B(\alpha)} \mathbb{E}_{\text{Dir}_{\alpha}}\left[e^{2nH(\theta)-n^\gamma\theta^{*}\cdot\log\theta}\mathbf{1}_{\Delta_{K-1}^{\epsilon}}\right] \nonumber \\
 & =\frac{B(\alpha+n^\gamma\theta^{*})e^{-n^\gamma\theta^{*}\cdot \log\theta^{*}}}{B(\alpha)}\mathbb{E}_{\text{Dir}_{\alpha}}\left[e^{2nH(\theta)+n^\gamma\text{KL}(\theta^{*}|\theta)}\mathbf{1}_{\Delta_{K-1}^{\epsilon}}\right],
\end{align}
where, as before, $\text{KL}(\theta^{*}|\theta)$
is the KL divergence defined in \eqref{eq:KL_divergence}. By the Laplace method (Thm \ref{thm:boundary-laplace-with-kl}), the last expectation in \eqref{eq:second_moment_rewritten} is of the same asymptotic order as $\mathbb{E}_{\text{Dir}_{\alpha}}\left[e^{2nH(\theta)}\right]$ (Thm \ref{thm:boundary-laplace_standard}).
Moreover the Beta function approximation lemma (Lemma~\ref{lem:beta-asymptotic}) gives that the first factor of \eqref{eq:second_moment_rewritten} is
\begin{equation}
    \frac{B(\alpha+n^\gamma\theta^{*})e^{-n^\gamma\theta^{*}\cdot \log\theta^{*}}}{B(\alpha)} \sim \Theta(n^{-\gamma\frac{K-1-m}{2}}n^{-\gamma\sum_{i =1}^{m}\alpha_{i}}) 
\end{equation}
which decays slower than $\mathbb{E}_{\text{Dir}_{\alpha}}
 [e^{nH(\theta)}]$ since $\gamma<1$. Moreover, since $\mathbb{E}_{\text{Dir}_{\alpha}}
 [e^{nH(\theta)}]$ and $\mathbb{E}_{\text{Dir}_{\alpha}}
 [e^{2nH(\theta)}]$ have the same polynomial factor, and since $\mathbb{E}_{\text{Dir}_{\alpha}}
 [e^{nH(\theta)}\mathbf{1}_{\Delta_{K-1}^\epsilon}]$ is of the same order as $\mathbb{E}_{\text{Dir}_{\alpha}}
 [e^{nH(\theta)}]$ by Localization lemma (Supplemental Material \ref{sec:localization_lemma}), we have that the second moments are the dominating terms in both the numerator and the denominator in \eqref{eq:variance_ratio}. Hence it suffices to look at the ratio of the second moments for the variance ratio. 

The ratio of second moments satisfies
    \begin{align*}
\frac{\mathbb{E}_{\text{Dir}_{\alpha}}\left[e^{2nH(\theta)}\frac{\text{Dir}_{\alpha}(\theta)}{\text{Dir}_{\alpha+n^\gamma   \theta^{*}}(\theta)}\mathbf{1}_{\Delta_{K-1}^{\epsilon}}\right]}{\mathbb{E}_{\text{Dir}_{\alpha}}\left[e^{2nH(\theta)}\right]} 
& =\frac{\frac{e^{-n^\gamma\theta^{*} \cdot \log\theta^{*}}B(\alpha+n^\gamma \theta^{*})}{B(\alpha)}\mathbb{E}_{\text{Dir}_{\alpha}}\left[e^{2nH(\theta)+n^\gamma\text{KL}(\theta^{*}|\theta)}\mathbf{1}_{\Delta_{K-1}^{\epsilon}}\right]}{\mathbb{E}_{\text{Dir}_{\alpha}}\left[e^{2nH(\theta)}\right]}\\
 & \stackrel{Thm \ref{thm:boundary-laplace_standard},\ref{thm:boundary-laplace-with-kl}}{\sim} \frac{C^{\prime}_H}{C_H} \cdot \frac{e^{-n^\gamma \theta^{*} \cdot \log\theta^{*}}B(\alpha+n^\gamma \theta^{*})}{B(\alpha)}\\
 & \stackrel{Lem \ref{lem:beta-asymptotic}}{\sim} \frac{C^{\prime}_H}{C_H}  \cdot \frac{e^{-n^\gamma \theta^{*} \cdot \log\theta^{*}}C_{B}n^{-\gamma\frac{K-1-m}{2}}n^{-\gamma\sum_{i =1}^{m}\alpha_{i}}e^{n^\gamma \theta^{*}\cdot\log\theta^{*}}}{B(\alpha)}\\
 & =Cn^{-\gamma\frac{K-1-m}{2}}n^{-\gamma\sum_{i =1}^{m}\alpha_{i}}
\end{align*}
where 
$C=\frac{C^{\prime}_H}{C_H}  \cdot \frac{C_{B}}{B(\alpha)}=  \frac{C_{B}}{B(\alpha)}>0$. 
\end{proof}

\subsubsection{Proof of Lemma \ref{lem:boundary-bias-is} (Negligible bias of the IS estimator) }
\label{sec:appendix-proof-boundary-bias-is}

\begin{proof}

Define the neighborhood radius $$r_\epsilon := \min_{i: \theta^*_i>0} \left( \theta_i^*-\epsilon\right)>0$$
which is well defined by the definition of $\epsilon$.
Then $||\theta-\theta^*||_1 > r_\epsilon$ on $\Delta_{K-1} \backslash \Delta_{K-1}^\epsilon$. Since $|| \cdot||_1 \le \sqrt{K} \,||\cdot||_2$ on $\mathbb{R}^K$, we have that
\begin{equation}
\label{eq:complement_truncated_simplex_ball}
    {\Delta}_{K-1}\backslash {\Delta}_{K-1}^\epsilon \subset \{ \theta \in {\Delta}_{K-1}:\, ||\theta -\theta^*||_2 > \frac{r_\epsilon}{\sqrt{K}}\}.
\end{equation}
By the strict gap property \eqref{eq:continuity_relaxation}, which follows from continuity of $H$ on the compact simplex $\Delta_{K-1}$ and uniqueness of the maximizer $\theta^*$, there exists $\delta = \delta(r_\epsilon/\sqrt{K})>0$ such that
 \begin{equation}
 \label{eq:ball_complement_H}
 \sup_{\theta\in \Delta_{K-1}\backslash B_{r_\epsilon/\sqrt{K}}(\theta^*)} H(\theta) 
\leq H(\theta^{*})-\delta.
 \end{equation}
From \eqref{eq:complement_truncated_simplex_ball} and \eqref{eq:ball_complement_H}, we have that $\sup_{\theta\in \Delta_{K-1}\backslash \Delta_{K-1}^{\epsilon}} H(\theta) 
\leq H(\theta^{*})-\delta.$
This gives
\[
\text{\text{Bias}}(\hat{p}_{\text{IS}}^{\gamma})=\mathbb{E}_{\text{Dir}_{\alpha}}[e^{nH(\theta)}\mathbf{1}_{\left(\Delta_{K-1}\backslash\Delta_{K-1}^{\epsilon}\right)}]\leq e^{n(H(\theta^{*})-\delta)}.
\]
Next we bound the variance of the standard Monte Carlo estimator from below.
From Theorem \ref{thm:boundary-laplace_standard}, we have that $$\mathbb{E}_{\text{Dir}_{\alpha}}\left[e^{2nH(\theta)}\right]=\Theta(n^{-\frac{(K-1-m)}{2}}n^{-\sum_{i=1}^{m}\alpha_{i}}e^{2nH(\theta^*)})$$ and 
$$\mathbb{E}_{\text{Dir}_{\alpha}}\left[e^{nH(\theta)}\right]^{2}=\Theta (n^{-(K-1-m)}n^{-2\sum_{i=1}^{m}\alpha_{i}}e^{2nH(\theta^{*})}).$$ 
Hence for large $n$, there exists a constant $C>0$ such that
\[
\text{Var}_{\alpha}(e^{nH(\theta)})\geq C\cdot n^{-\frac{(K-1-m)}{2}}n^{-\sum_{i=1}^{m}\alpha_{i}}e^{2nH(\theta^{*})}.
\]
Combining these bounds, we get
\begin{align*}
\frac{\text{Bias}^2(\hat{p}_{\text{IS}}^{\gamma})}{\text{Var}(\hat{p}_{\text{MC}})} & \leq\frac{e^{2n(H(\theta^{*})-\delta)}}{C\cdot n^{-\frac{(K-1-m)}{2}}n^{-\sum_{i=1}^{m}\alpha_{i}}e^{2nH(\theta^{*})}}\\
 & =C^{-1}n^{\frac{K-1-m}{2}}n^{\sum_{i=1}^{m}\alpha_{i}}e^{-2n\delta}\\
 & =O(e^{-2n\delta'})\quad\forall \enspace 0<\delta'<\delta.
\end{align*}
\end{proof}

\subsubsection{Proof of Theorem~\ref{thm:boundary-mse-is} (MSE reduction of the IS estimator)} 
\label{sec:appendix-proof-boundary-mse-is}

\begin{proof}[Proof of Theorem \ref{thm:boundary-mse-is}]

Since $\hat{p}_{\mathrm{MC}}$ is unbiased $\text{MSE}(\hat{p}_{\mathrm{MC}})=\text{Var}(\hat{p}_{\mathrm{MC}})$. From Lemma \ref{lem:boundary-bias-is} we have that
\begin{equation}
\frac{\text{Var}(\hat{p}_{\mathrm{IS}})}{\text{MSE}(\hat{p}_{\mathrm{MC}})}
\le 
\frac{\text{MSE}(\hat{p}_{\mathrm{IS}})}{\text{MSE}(\hat{p}_{\mathrm{MC}})}
\le 
\frac{\text{Var}(\hat{p}_{\mathrm{IS}})}{\text{MSE}(\hat{p}_{\mathrm{MC}})}+ O(e^{-2n\delta'})
\end{equation}
for some $\delta'>0$. Since any function $f$ in $O(e^{-2n\delta'})$ is also in $o(n^{-k})$ for every $k > 0$, the result follows from Theorem \ref{thm:boundary-variance-is}. 
\end{proof}

\subsection{Analysis of KL-Based Control Variate}
\label{sec: control_variate_analysis}

To analyze the correlation, an important quantity to analyze is $\mathbb{E}[e^{n(H(\theta)+\widehat{H}(\theta))}]$ which is the leading order term of the covariance. We note that the sum $H+\widehat{H}$ satisfies (A1)--(A4), which allows the application of Theorem \ref{thm:boundary-laplace_standard}.

\paragraph{(A1)--(A4) for $H+\widehat{H}.$}

From the discussion of $\widehat{H}$ in Section \ref{sec:reduced_KL}, it is clear that $H+\widehat{H}$ is uniquely maximized at the same maximizer $\theta^*$ of $H$. Moreover $H+\widehat{H}$ satisfies (A2) in the sense that it satisfies the strict gap condition \eqref{eq:continuity_relaxation} instead of the full continuity on $\Delta_{K-1}$.

The problem of maximizing $H+\hat{H}$ admits KKT multipliers $\tilde{\lambda}$ associated with the inequality constraints $\theta_i\ge0$, given by
$$\tilde{\lambda}
=
\lambda
+
\sum_{i=1}^{m} e_i,\quad e_i\in\mathbb{R}^K, $$
where $\lambda$ is the vector of KKT multipliers for $H$ discussed in \eqref{eq:KKT_conditions}.
Equivalently,
$$\tilde{\lambda}_i=
\begin{cases}
\lambda_i+1, & \theta_i^*=0,\\
\lambda_i, & \theta_i^*>0.
\end{cases}$$
This follows from the KKT multiplier properties of $\widehat{H}$ in \eqref{eq:gradient_of_wide_hat_H},
which modifies the stationary condition in \eqref{eq:KKT_conditions} by adding one unit in each strictly positive coordinate. Since $\tilde{\lambda}_i>0$, strict complementarity (A3) always holds for $H+\widehat H$ at $\theta^*$ even if (A3) fails for $H$.

By the negative definiteness of the reduced Hessians of $H$ and $\widehat{H}$, 
$$U^\top \nabla_\theta^2 (H+\widehat H)(\theta^*) U \prec 0$$
which gives (A4).

\subsubsection{Proof of Theorem \ref{thm:CV_correlation_boundary}}
\label{sec:CV_limiting_correlation_boundary}

\begin{proof}[Proof of Theorem \ref{thm:CV_correlation_boundary}]

Let $\rho_n$ denote the correlation in \eqref{eq:correlation}. We study the squared correlation $\rho_n^{2}$:
\begin{align}
 \rho_n^{2}=\frac{\text{Cov}(e^{nH(\theta)},e^{n\widehat{H}(\theta)})^{2}}{\text{Var}(e^{nH(\theta)})\text{Var}(e^{n\widehat{H}(\theta)})}.
\end{align}
The leading order term in the numerator is the square of $\mathbb{E} \left[e^{n\left(H(\theta)+\widehat{H}(\theta) \right)}\right]$. On the other hand, leading order term in the denominator should be the product of $\mathbb{E}\left[e^{2nH(\theta)} \right]$ and $\mathbb{E}\left[e^{2n\widehat{H}(\theta)} \right]$. The outline of the proof is the following:
\begin{enumerate}[itemsep=0pt, topsep=0pt, parsep=0pt, partopsep=0pt]
\item We apply the Laplace method to obtain the asymptotic rate for the leading order term $\mathbb{E} \left[e^{n\left(H(\theta)+\widehat{H}(\theta) \right)}\right]$ in the numerator.
\item Repeat for the leading terms $\mathbb{E}\left[e^{2nH(\theta)} \right]$ and $\mathbb{E}\left[e^{2n\widehat{H}(\theta)} \right]$ in the denominator of $\rho_n^2$
\item Evaluate the limiting correlation. The numerator and denominator have identical scaling in $n$, so the squared correlation is determined by the constants in the asymptotics.
\end{enumerate}

\paragraph{1. Asymptotics for the covariance.} 

We proceed to analyze the leading-order term in the covariance. It is useful to introduce the notation $$\widetilde{H}(\theta) := H(\theta) + \widehat{H}(\theta)$$ for the sum. Since $\widetilde{H}$ satisfies (A1)--(A4), Theorem \ref{thm:boundary-laplace_standard} gives that
\begin{align}
\label{eq:tilde_H_laplace_1}
\mathbb{E}[e^{n\widetilde{H}(\theta)}] 
 &\sim C_{\widetilde{H}}n^{-\frac{(K-1-m)}{2}}n^{-\sum_{i=1}^{m}\alpha_{i}}e^{n\widetilde{H}(\theta^{*})} 
\end{align}
where $C_{\widetilde{H}}$ is the constant given by
\begin{equation}
C_{\widetilde{H}}=\text{Dir}_{\alpha,m}(\theta^{*})\prod_{k=1}^{m}(\lambda_k+1)^{-\alpha_{k}}\Gamma(\alpha_{k})\frac{(2\pi)^{\frac{K-1-m}{2}}}{|\det(U^{\top}\left(\nabla^{2}H(\theta^{*})- \nabla^{2}\text{KL}(\theta^{*}|\theta)\big|_{\theta=\theta^*}\right)U)|^{1/2}}.
\end{equation}
Using $$\max_{\theta \in \Delta_{K-1}}\widetilde{H}(\theta) = \widetilde{H}(\theta^*)= 2H(\theta^*),$$ \eqref{eq:tilde_H_laplace_1} reduces to
\begin{align}
\label{eq:tilde_H_laplace}
\mathbb{E}[e^{n\widetilde{H}(\theta)}] 
 &\sim C_{\widetilde{H}}n^{-\frac{(K-1-m)}{2}}n^{-\sum_{i=1}^{m}\alpha_{i}}e^{2nH(\theta^{*})} 
\end{align}

\paragraph{2. Asymptotics for the variance terms.}
We now analyze the variance terms in the denominator. By the Laplace method in Theorem \ref{thm:boundary-laplace_standard} we have that
\begin{align}
\label{eq:product_term1_correlation}
\mathbb{E}_{\text{Dir}_\alpha} \left[ e^{2nH({\theta})}\right]
 &\sim C_H \, (2n)^{-\frac{(K-1-m)}{2}}(2n)^{-\sum_{i=1}^{m}\alpha_{i}}e^{2nH(\theta^{*})}
\end{align}
where
\[
C_H =\text{Dir}_{\alpha,m}(\theta^{*}) \prod_{k=1}^{m}\lambda_{k}^{-\alpha_{k}}\Gamma(\alpha_{k})\frac{(2\pi)^{\frac{K-1-m}{2}}}{|\det(U^{\top}\nabla^{2}H(\theta^{*})U)|^{1/2}}.\]
    From earlier calculation in \eqref{eq:main_widehat_H}, replacing $n^\gamma$ with $2n$, we have that 
\begin{align}
\label{eq:product_term2_correlation}
\mathbb{E}_{\text{Dir}_\alpha} \left[ e^{2n\widehat{H}({\theta})}\right] 
 & \sim C_{B} \cdot e^{2nH(\theta^{*})}(2n)^{-\frac{K-1-m}{2}}(2n)^{-\sum_{i=1}^m \alpha_{i}}
\end{align}
where
\begin{align}
    C_{B} &=\frac{(2\pi)^{\frac{K-1-m}{2}}}{B(\alpha)}\prod_{i=1}^{m}\Gamma(\alpha_{i})\prod_{k = m+1}^K(\theta_{k}^*)^{\alpha_{k}-\frac{1}{2}} \nonumber \\
    &= \frac{(2\pi)^{\frac{K-1-m}{2}}\text{Dir}_{\alpha,m}(\theta^{*})\prod_{i=1}^{m}\Gamma(\alpha_{i})}{\left|\det\!\left(U^{\top}(-\nabla^{2}\text{KL}(\theta^{*}|\theta)\big|_{\theta=\theta^*}U\right)\right|^\frac{1}{2}}.
\end{align}
Note that by taking $n$ instead of $2n$ in \eqref{eq:product_term1_correlation} and \eqref{eq:product_term2_correlation} to get asymptotics for  $\mathbb{E}_{\text{Dir}_\alpha} \left[ e^{nH({\theta})}\right]$ and $\mathbb{E}_{\text{Dir}_\alpha} \left[ e^{n\widehat{H}({\theta})}\right]$, and comparing the product with the asymptotics for $\mathbb{E}[e^{n\widetilde{H}(\theta)}]$ in \eqref{eq:tilde_H_laplace}, we can confirm that $\mathbb{E}\left[e^{n\left(H(\theta)+\widehat{H}(\theta) \right)}\right]$ is the leading order of the covariance. In other words, 
\begin{align*}
\text{Cov}(e^{nH(\theta)},e^{n\widehat{H}(\theta)}) &= \mathbb{E}\left[e^{n\left(H(\theta)+\widehat{H}(\theta) \right)}\right]
- \mathbb{E}\left[e^{n\widehat{H}(\theta)}\right] \mathbb{E}\left[e^{nH(\theta)}\right] \\
&\sim C_{\widetilde{H}}n^{-\frac{(K-1-m)}{2}}n^{-\sum_{i=1}^{m}\alpha_{i}}e^{2nH(\theta^{*})}.
\end{align*}

\paragraph{3. Evaluate the limiting correlation.}
We substitute the asymptotics for the leading terms in the numerator and denominator of $\rho^2_n$ and obtain the following:
\begin{align*}
\label{eq:limiting_correlation_boundary}
\rho^{2}_n &\sim\frac{\left(C_{\widetilde{H}}n^{-\frac{(K-1-m)}{2}}n^{-\sum_{i=1}^{m}\alpha_{i}}e^{2nH(\theta^{*})}\right)^{2}}{\left( C_{H}\cdot(2n)^{-\frac{(K-1-m)}{2}}(2n)^{-\sum_{i=1}^{m}\alpha_{i}}e^{2nH(\theta^{*})} \right) \left(C_{B} \, e^{2nH(\theta^{*})}\,(2n)^{-\frac{K-1-m}{2}}(2n)^{-\sum_{i=1}^{m}\alpha_{i}} \right)} \nonumber \\
 &= \frac{C_{\widetilde{H}}^2}{C_{H}\cdot C_{B}} 2^{(K-1-m) + 2\sum_{i=1}^{m} \alpha_{i}}  \nonumber \\
 &= \frac{\left(\text{Dir}_{\alpha,m}(\theta^{*}) \prod_{k=1}^{m}(\lambda_k+1)^{-\alpha_{k}}\Gamma(\alpha_{k})\frac{(2\pi)^{\frac{K-1-m}{2}}}{|\det(U^{\top}\nabla^{2}\widetilde{H}(\theta^*)U)|^{1/2}}\right)^2}{\left(  \text{Dir}_{\alpha,m}(\theta^{*})\prod_{k=1}^{m}\lambda_{k}^{-\alpha_{k}}\Gamma(\alpha_{k})\frac{(2\pi)^{\frac{K-1-m}{2}}}{|\det(U^{\top}\nabla^{2}H(\theta^{*})U)|^{1/2}} \right) } \times \frac{2^{\,(K-1-m) + 2\sum_{i=1}^{m}\alpha_i}}{ \frac{(2\pi)^{\frac{K-1-m}{2}}\text{Dir}_{\alpha,m}(\theta^{*})\prod_{i=1}^{m}\Gamma(\alpha_{i})}{\left|\det\!\left(U^{\top}(-\nabla^{2}\text{KL}(\theta^{*}|\theta)\big|_{\theta=\theta^*}U\right)\right|^\frac{1}{2}}} \nonumber \\
&\stackrel{eq.\eqref{eq:det_of_projected_KL}}{=}   \prod_{k=1}^{m}\left(\frac{4 \lambda_k}{(\lambda_k+1)^2}\right)^{\alpha_k}
\left( 2^{(K-1-m)} \, \frac{\left|\det\!\left(U^{\top}\nabla^{2}H(\theta^{*})U\right)\right|^\frac{1}{2} \, \left|\det\!\left(U^{\top}(-\nabla^{2}\text{KL}(\theta^{*}|\theta)\big|_{\theta=\theta^*}U\right)\right|^\frac{1}{2}}
{\left|\det\!\left(U^{\top}\left(\nabla^{2}H(\theta^{*})-\nabla^{2}\text{KL}(\theta^{*}|\theta)\big|_{\theta=\theta^*}\right)U\right)\right|}\right) 
\end{align*}
where we utilize the identity in \eqref{eq:det_of_projected_KL} in the last equality.

\end{proof}

\subsubsection{Proof of Corollary~\ref{cor:CV_correlation_interior} (Limiting Correlation in the Interior Case)}
\label{sec:appendix-proof-CV_correlation_interior}

\begin{proof}[Proof of Corollary~\ref{cor:CV_correlation_interior}]

From the discussion of $U$ in \eqref{eq:U_basis}, when $\theta^*$ is an interior point $U=A$. Since $A$ has full rank and $\nabla^2 H(\theta^*)\prec 0$, $|\det(A^{\top}\nabla^2H (\theta^{*})A)|=\det(-A^{\top}\nabla^2 H(\theta^{*})A)$. 
This can be further evaluated by defining $$C=\left[\begin{array}{cc}
A & \mathbf{1}_{K}\end{array}\right]\in\mathbb{R}^{K\times K}$$ and using the Schur complement of $-A^{\top}\nabla^2 H(\theta^{*})A$ in $-C^{\top}\nabla^{2} H(\theta^{*})C$ (cf. \cite{boyd2004convex} A.5.5) to get
\begin{equation}
\det(-A^{\top}\nabla^2H(\theta^{*})A) 
=\det\left(-\nabla^{2}H(\theta^{*})\right)\cdot\left(\mathbf{1}_{K}^{\top}(-\nabla^{2}H(\theta^{*}))^{-1}\mathbf{1}_{K}\right).
\label{eq:schur}
\end{equation}

For $H$, we can observe that due to the expression for its Hessian $\nabla^{2}H(\theta)$ in \eqref{eq:hess}, we have the identity 
\begin{equation}
    \label{eq:hessian_gradient_identity}
    \nabla H(\theta)=-\nabla^2H(\theta) \theta.
\end{equation}
Since $\nabla H(\theta^*) = \mathbf{1}_{K}$, this implies that $(-\nabla^{2}H(\theta^{*}))^{-1} \mathbf{1}_{K} = \theta^{*}$ and therefore:
\[
\left(\mathbf{1}_{K}^{\top}(-\nabla^{2}H(\theta^{*}))^{-1}\mathbf{1}_{K}\right) = 1
\]
and by~\eqref{eq:schur}, we have that $$\det(-A^{\top}\nabla^2H(\theta^{*})A) 
=\det\left(-\nabla^{2}H(\theta^{*})\right).$$

Since $\KL$ and $\widetilde{H}(\theta)$ are strictly concave at $\theta^*$, we also have that expression~\eqref{eq:schur} also holds for their Hessians.
To further simplify the quadratic forms, we can observe from the expression of $\nabla^2 \KL$ in \eqref{eq:KL_Hessian} that
$(\nabla^{2}\KL(\theta^{*}|\theta)\big|_{\theta=\theta^*})^{-1}\mathbf{1}_{K} = \theta^{*}$ and so $\mathbf{1}_{K}^{\top}(\nabla^{2}\KL(\theta^{*}|\theta)\big|_{\theta=\theta^*})^{-1}\mathbf{1}_{K} = 1$. Finally, since $\widetilde{H}(\theta)$ involves the sum of $-\KL(\theta^*|\theta)$ and $H(\theta)$, it follows that $(-\nabla^{2}\widetilde{H}(\theta^*))\theta^{*} = 2\cdot\mathbf{1}_{K}$ and $\mathbf{1}_{K}^{\top}(-\nabla^{2}\widetilde{H}(\theta^{*}))^{-1}\mathbf{1}_{K} = 1/2$. In total this gives
\begin{align}
 \det(A^{\top}\nabla^2\KL(\theta^{*}|\theta)\big|_{\theta=\theta^*}A) 
&=\det\left(\nabla^{2}\KL(\theta^{*}|\theta)\big|_{\theta=\theta^*}\right) \\
\det(-A^{\top}\nabla^2\widetilde{H}(\theta^{*})A) 
&=\frac{1}{2}\det\left(-\nabla^{2}\widetilde{H}(\theta^{*})\right).
\end{align}

\end{proof}

\subsubsection{Proof of Theorem \ref{thm:epsilon_main_result} (Almost Mutually Orthogonal Case)}
\label{sec:appendix-proof-epsilon_main_result}

We collect several auxiliary lemmas used in the proof of 
Theorem~\ref{thm:epsilon_main_result}. For the proofs refer to Supplemental Material \ref{sec:aux_lemmas_lDA_orthogonal}.

\begin{lem}
\label{lem:gtbound_lemma}
Let
\[
f(t)=\frac{(t-1)^2}{t},
\qquad
g(t)=\log\!\left(\frac{\tfrac12(1+t)}{t^{1/2}}\right),
\quad t>0.
\]
Then
\[
g(t)\le \frac{1}{8} f(t),
\qquad t>0.
\]
\end{lem}

\begin{lem}
\label{lem:Frobenius_bound}
Let $Z \in \mathbb{S}_{++}^K$ have eigenvalues $\lambda_1, \dots ,\lambda_K$, and let $\lambda_{\min}(Z) $ denote the smallest eigenvalue of $Z$. Then
\begin{equation}
\label{eq:Frobenius_norm_bound}
\sum_{i=1}^K f(\lambda_i)
\le
\frac{\|Z-I\|_F^2}{\lambda_{\min}(Z)} .
\end{equation}
\end{lem}

\begin{lem}
\label{lem:Z-Ibound_calculation}
Assume $K \ge 3$. Suppose the topic vectors $\phi$ satisfy \textup{(B1)} and are $\varepsilon$-sparse where $\varepsilon<\varepsilon_0$. 
Let $H$ denote the LDA log-likelihood in \eqref{eq:entropy} and let $\theta^*$ be an interior maximizer. Define
$$X=\nabla^{2}\!\KL(\theta^* |\theta)\big|_{\theta=\theta^*},
\qquad Y= -\nabla^{2}H(\theta^*),
\qquad
Z = X^{-1/2} Y X^{-1/2}.$$
Then there exists $C>0$, independent of $\varepsilon$, such that
$$\frac{\|Z-I\|_F}{\sqrt{\lambda_{\min}(Z)}}
\le
C \varepsilon .$$
\end{lem}

\begin{proof}[Proof of Theorem \ref{thm:epsilon_main_result}] 

The outline of the proof is as follows:
\begin{enumerate}[itemsep=0pt, topsep=0pt, parsep=0pt, partopsep=0pt]
\item Rewrite $\log\rho^2$ in terms of $\log\det Z$ and $\log\det(I+Z)$ where $Z$ is a positive definite matrix.
\item Using properties of $\log\det$, we write $\log\rho^2$ as $ -\sum_{i=1}^K g(\lambda_i)$ where $\lambda_i$ are eigenvalues of $Z$.
\item By combining Lemmas~\ref{lem:gtbound_lemma}, \ref{lem:Frobenius_bound}, and \ref{lem:Z-Ibound_calculation}, we obtain an upper bound on $\sum_{i=1}^K g(\lambda_i)$ of the form $C^2\varepsilon^2$ for some constant $C>0$. Consequently, this yields the lower bound
$\log \rho^2 \geq -C^2\varepsilon^2.$

\end{enumerate} 

\paragraph{1. Rewrite $\log \rho^2$.} The $\log$ form of the correlation quantity $\rho^2$ in \eqref{eq:limiting_rho_interior} can be written as
\begin{align} 
\label{eq:rho_re_expression}
\log \rho^2&=\frac{1}{2}\left(\log\det X+\log\det Y\right)-\log\det(\frac{1}{2}(X+Y))  
\end{align}
This follows from the fact that $$| \det(A)|=\det(-A)$$ for any negative definite matrix $A$. Equivalently we can rewrite \eqref{eq:rho_re_expression} as 
\begin{equation}
    \label{eq:rho_re_expression_Z}
    \log \rho^2 = \frac{1}{2}\log\det Z-\log\det(\frac{1}{2}(I+Z)).
\end{equation}

\paragraph{2. Rewrite $\log\rho^2$ in terms of eigenvalues.}
 By positive definiteness of $X$ and $Y$, $Z$ is positive definite and its eigenvalues are positive. Let $\lambda_1, \dots, \lambda_K$ denote the eigenvalues. Then \eqref{eq:rho_re_expression_Z}
can be written as 
\begin{align}
\label{eq:log_rewrite_g}
\frac{1}{2}\log\det Z-\log\det(\frac{1}{2}(I+Z)) &=\frac{1}{2} \sum_{i=1}^K \log\lambda_i-\sum_{i=1}^K \log (\frac{1}{2}(1+\lambda_i)) \nonumber \\
&= -\sum_{i=1}^K g(\lambda_i).
\end{align}
Therefore from \eqref{eq:log_rewrite_g}, to find a lower bound of $\log \rho^2$, it suffices to find an upper bound of $\sum_{i=1}^K g(\lambda_i)$. 

\paragraph{3. Upper bound of $\sum_{i=1}^K g(\lambda_i)$.} 
By lemma \ref{lem:gtbound_lemma}, we can bound $\sum_{i=1}^K g(\lambda_i)$ by $ \sum_{i=1}^Kf(\lambda_i)$. 
From lemma \ref{lem:Frobenius_bound}, $\sum_{i=1}^Kf(\lambda_i)$ can be upper bounded by $\frac{\|Z-I\|_F^2}{\lambda_{\min}(Z)}$. 
To bound  $\frac{\|Z-I\|_F^2}{\lambda_{\min}(Z)}$, we use the properties of the $Z$ matrix. Since the topic-vector $\phi$ is $\varepsilon$-sparse, lemma \ref{lem:Z-Ibound_calculation} states that there exists a constant $C>0$ such that
$$||Z-I||_F \frac{1}{\sqrt{\lambda_{\min} (Z)}}\le C\varepsilon.$$ Collecting all bounds, we have that
\begin{align*}
    \sum_{i=1}^K g(\lambda_i) 
    &\le \frac{1}{8} \sum_{i=1}^K f(\lambda_i) \\
    &\le  \frac{1}{8} \cdot \frac{\|Z-I\|_F^2}{\lambda_{\min}(Z)}\\
    &\le \frac{1}{8} (C\varepsilon)^2.
\end{align*}
We achieve the desired result by multiplying $-1$ and taking exponent on both sides. 

\end{proof}

\section*{Acknowledgments}
The author thanks Paul Glasserman for numerous discussions and helpful feedback.

\bibliographystyle{plainnat}  
\bibliography{references}

\clearpage
\section{Supplemental Material: Proofs of Technical Lemmas}
\label{app:supplement}

As before, for $\theta^*$ we assume the ordering assumption in \eqref{eq:ordering_assumption}. 

\subsection{Localization Lemma}
\label{sec:localization_lemma}

We state the localization lemma, which allows us to replace the original domain of an integral with a closed neighborhood of the maximizer. We adapt Lemma 38 of \cite{breitung1994asymptotic} to our setting. 

\begin{lem}[Localization]
\label{lem:localization}
Let $G:\Delta_{K-1}\to\mathbb{R}\cup\{-\infty\}$ and let $T$ be the map defined in
\eqref{eq:T_map}. Assume that $y^*\in\tilde{\Delta}_{K-1}$ satisfies $G(T(y^*))\in\mathbb{R}$.
Suppose that for every
$\delta>0$ there exists $\eta(\delta)>0$ such that
\begin{equation}
\label{eq:strict_gap_assumption}
    \sup_{y\in \tilde{\Delta}_{K-1}\setminus B(y^*,\delta)} G(T(y))
\le G(T(y^*))-\eta(\delta).
\end{equation}
Then for any closed set $W\subset\tilde{\Delta}_{K-1}$ such that there exists
$r>0$ with
\begin{equation}
\label{eq:closed_neighborhood_condition}
    \tilde{\Delta}_{K-1}\cap B(y^*,r)\subseteq W,
\qquad
B(y^*,r):=\{y\in\mathbb{R}^{K-1}:\|y-y^*\|_2<r\},
\end{equation}
as $n\to\infty$,
$$
\int_{\tilde{\Delta}_{K-1}} e^{nG(T(y))}
\, \mathrm{Dir}_{\alpha}^{(K-1)}(y)\,dy
\sim
\int_{W} e^{nG(T(y))}
\, \mathrm{Dir}_{\alpha}^{(K-1)}(y)\,dy .
$$
\end{lem}

\begin{proof}
    For the proof, see p.~53 of \cite{breitung1994asymptotic}. 
    Breitung's Lemma 38 applies with $h$ as $\text{Dir}_\alpha^{(K-1)}$, $f$ as $G(T(\cdot))$, and $F$ as $\tilde{\Delta}_{K-1}$. Note that Breitung assumes continuity of $h$ on a closed domain, which may fail here since $\text{Dir}_\alpha^{(K-1)}$ can be unbounded near the boundary when some $\alpha_i<1$. However, his proof requires only integrability of $e^{G(T(\cdot))}$ against $|h|$ (his assumption 2), positivity near $y^*$ (his assumption 4), and positive mass near $y^*$ (his assumption 5). All hold: $\text{Dir}_\alpha^{(K-1)}(y)>0$ for any $y$ in $\tilde{\Delta}_{K-1}$, $\int_{\tilde{\Delta}_{K-1}}\text{Dir}_\alpha^{(K-1)}(y)dy<\infty$ (which implies $\int_{\tilde{\Delta}_{K-1}} e^{nG(T(y))}
\, \mathrm{Dir}_{\alpha}^{(K-1)}(y)\,dy<\infty$), and $\int_{W\cap\tilde{\Delta}_{K-1}}\text{Dir}_\alpha^{(K-1)}(y)dy>0$ for any neighborhood $W$ of $y^*$.
\end{proof}

This lemma has two immediate consequences.
\begin{enumerate}[itemsep=0pt, topsep=0pt, parsep=0pt, partopsep=0pt]
\item If $G$ is continuous on $\Delta_{K-1}$ and $y^*$ is the unique global maximizer of $G(T(\cdot))$, then the strict gap condition in \eqref{eq:strict_gap_assumption} holds automatically. Continuity on the compact set $\tilde{\Delta}_{K-1}$ implies that $G(T(\cdot))$ attains its maximum on any closed subset, and uniqueness forces a strict inequality away from $y^*$.

\item Taking $G=H$ and $W$ to be the projected truncated simplex $\tilde{\Delta}_{K-1}^{\epsilon}$ defined in \eqref{eq:epsilon_projected_simplex}, and noting that $T$ is a bijection between $\tilde{\Delta}_{K-1}^\epsilon$ and $\Delta_{K-1}^\epsilon$, we obtain
$$
\mathbb{E}_{\mathrm{Dir}_{\alpha}}
\!\left[e^{nH(\theta)}\mathbf{1}_{\Delta_{K-1}^{\epsilon}}\right]
\sim
\mathbb{E}_{\mathrm{Dir}_{\alpha}}
\!\left[e^{nH(\theta)}\right],
\qquad n\to\infty.
$$
\end{enumerate}

\subsection{Bound around Maximum Lemma}
\label{sec:bound_around_maximum_lemma}

\begin{proof}    

We adapt the proof of Lemma 45 (Chapter 5, p. 65) in \cite{breitung1994asymptotic}, which was originally derived for an optimal point satisfying the single boundary constraint ${y_n \ge 0}$ with $y_n^* = 0$. We extend this argument to the case where the optimal point lies on the boundary of multiple constraints ${y_i \ge 0}$ and on the projected simplex $\tilde{\Delta}_{K-1}$.

    Let $\tilde{\theta}^*=T^{-1}(\theta^*)$. WLOG assume that $H(T(\tilde{\theta}^*))=H(\theta^*)=0$. We prove by contradiction. Assume that no constants $c_1,c_2$ exists that satisfies \eqref{eq:H_hess_bound}. Then for every constant $c_1,c_2>0$, we can find $y\in \tilde{\Delta}_{K-1}$ that violates \eqref{eq:H_hess_bound}. Setting $c_1=c_2=i^{-1}$, we can construct a sequence $\{y^{(i)} \}\subset \tilde{\Delta}_{K-1} \backslash \{ \ttheta \}$ such that
   
\begin{equation}
\label{eq:tocontradict_1}
  0 \ge H(T(y^{(i)})) \left( \sum_{i=1}^{m} |y^{(i)}-\ttheta| + \sum_{i=m+1}^{K-1}(y^{(i)}  - \ttheta)^2 \right)^{-1} > -i^{-1}.
\end{equation} 
Using the expression of $P_m$ in \eqref{eq:projection_matrix}, \eqref{eq:tocontradict_1} can be rewritten as
\begin{equation}
\label{eq:tocontradict}
  0 \ge H(T(y^{(i)})) \left( ||P_m (y^{(i)}-\ttheta) ||_{1} + ||(I_{K-1}-P_m)(y^{(i)}  - \ttheta)||_{2}^{2} \right)^{-1} > -i^{-1}.
\end{equation}
    Since $\tilde{\Delta}_{K-1}$ is compact, the sequence $\{y^{(i)}\}$ admits a convergent subsequence, which we denote again by $\{y^{(i)}\}$. By \eqref{eq:tocontradict},
\[
H(T(y^{(i)})) \to H(T(\tilde{\theta}^*))=0 .
\]
Using the strict gap property \eqref{eq:continuity_relaxation} of $H$ around the maximizer, which follows from the continuity of $T$ with a unique maximizer,
it follows that
$$y^{(i)} \to \tilde{\theta}^* .$$
    Next define the sequence $\{\bar{y}^{(i)}\}$ by
$$\bar{y}^{(i)} = P_m \tilde{\theta}^* + (I_{K-1} - P_m)y^{(i)} .$$
By construction, $\bar{y}^{(i)}$ agrees with $\tilde{\theta}^*$ on the first $m$ coordinates and with $y^{(i)}$ on the remaining coordinates. 

Let $i$ be sufficiently large such that $T(y^{(i)}) \in W$, where $W$ is a neighborhood of $\theta^*$ on which $H$ is $C^2$ in assumption (A2). By construction, $ ||\bar{y}^{(i)}-\ttheta||_2 \le || y^{(i)}-\ttheta||_2$ and therefore $T(\bar{y}^{(i)}) \in W$ as well. By adding and subtracting $H(T(\bar{y}^{(i)}))$, the difference can be written as \begin{equation}
\label{eq:split sum}
H(T(y^{(i)}))-H(T(\tilde{\theta}^*)) = [H(T(y^{(i)}))-H(T(\bar{y}^{(i)}))] + [H(T(\bar{y}^{(i)}))-H(T(\tilde{\theta}^*))]
\end{equation}

For the first difference in \eqref{eq:split sum}, we note that every point $z$ on the line segment $\{(1-t)\bar{y}^{(i)}+ty^{(i)}: \,t\in [0,1]\}$ satisfies $||z-\ttheta||\le ||y^{(i)}-\ttheta||$. Therefore the line segment is contained in $T^{-1}(W)$ which is open, and on which $H(T(\cdot))$ is $C^1$. 
The mean value theorem (cf. \cite{rudin1964principles} Thm 5.10) implies that there exists $\gamma_i^{(1)}\in(0,1)$ such that
\begin{align}
\label{eq:first_order_taylor}
    &H(T(y^{(i)}))-H(T(\bar{y}^{(i)})) \nonumber \\
        & = \left(\ A^\top \nabla H(T(\gamma_i^{(1)}y^{(i)}+(1-\gamma_i^{(1)})\bar{y}^{(i))}) \right)^\top  (y^{(i)}-\bar{y}^{(i)}). 
\end{align} 
 Moreover, we now derive an upper bound for \eqref{eq:first_order_taylor}. Since $y^{(i)}\to \ttheta$, we have $\gamma_i^{(1)}y^{(i)}+(1-\gamma_i^{(1)})\bar{y}^{(i)}\to \tilde{\theta}^*$. Moreover, by strict complementarity condition (A3), $\left(A^\top \nabla H(\theta^*)\right)^\top  (y^{(i)}-\bar{y}^{(i)})=-\sum_{k=1}^m \lambda_k y^{(i)}_k\le 0$ where $\lambda_i>0$ for $i\le m$. 
By continuity of $\nabla H(T(\cdot))$, by taking $i$ larger if necessary, we obtain that \begin{align}
\label{eq:intermediate_bound_1}
    \left(A^\top \nabla H(\gamma_i^{(2)}y^{(i)}+(1-\gamma_i^{(2)})\bar{y}^{(i)}) \right)^\top  (y^{(i)}-\bar{y}^{(i)}) 
&\le  \frac{1}{2}\left(A^\top \nabla H(\theta^*)\right)^\top  (y^{(i)}-\bar{y}^{(i)}). 
\end{align}

Next, for the second difference in \eqref{eq:split sum}, similar to above, the line segment $\{(1-t)\ttheta+t\bar{y}^{(i)}: \,t\in [0,1]\}$ is contained in $T^{-1}(W)$ on which $H(T(\cdot))$ is $C^2$. Therefore by Taylor's theorem (cf. \cite{rudin1964principles} Thm 5.15), there exists $\gamma_i^{(2)}\in (0,1)$ such that
\begin{align}
\label{eq:second_order_taylor}
    &H(T(\bar{y}^{(i)}))-H(T(\tilde{\theta}^*))  \nonumber  \\
        & = \left(A^\top \nabla H(\theta^*) \right)^\top  (\bar{y}^{(i)}-\tilde{\theta}^*)  \nonumber \\
        &+ \frac{1}{2} (\bar{y}^{(i)}-\tilde{\theta}^*)^\top A^\top \nabla^2 H(T(\gamma_i^{(2)} \bar{y}^{(i)} +(1-\gamma_i^{(2)})\ttheta)) A(\bar{y}^{(i)}-\tilde{\theta}^*)  \nonumber  \\
        &=\frac{1}{2} (\bar{y}^{(i)}-\tilde{\theta}^*)^\top A^\top \nabla^2 H(T(\gamma_i^{(2)} \bar{y}^{(i)}+(1-\gamma_i^{(2)})\ttheta)) A(\bar{y}^{(i)}-\tilde{\theta}^*)
\end{align}
where the second equality follows from the fact that  $\left(A^\top \nabla H(\theta^*) \right)^\top (\bar{y}^{(i)}-\tilde{\theta}^*) =0$. Like the first difference, by taking $i$ larger if necessary, we can place an upper bound on \eqref{eq:second_order_taylor} as follows
\begin{align}
\label{eq:PBC_conclusion}
    &\frac{1}{2} (\bar{y}^{(i)}-\tilde{\theta}^*)^\top A^\top \nabla^2 H(T(\gamma_i^{(2)} \bar{y}^{(i)}+(1-\gamma_i^{(2)})\ttheta)) A(\bar{y}^{(i)}-\tilde{\theta}^*) \nonumber \\ &\le \frac{1}{4}  (\bar{y}^{(i)}-\tilde{\theta}^*)^\top A^\top \nabla^2 H(\theta^*) A(\bar{y}^{(i)}-\tilde{\theta}^*)  
\end{align}

Combining \eqref{eq:first_order_taylor}, \eqref{eq:intermediate_bound_1}, \eqref{eq:second_order_taylor}, and
\eqref{eq:PBC_conclusion}, and plugging them into \eqref{eq:split sum} gives that
\begin{equation}
\label{eq:PBC_conclusion_2}
    H(T(y^{(i)})) \le \frac{1}{2} \left( \left(A^\top \nabla H(\theta^*)\right)^\top  (y^{(i)}-\bar{y}^{(i)}) + \frac{1}{2} (\bar{y}^{(i)}-\tilde{\theta}^*)^\top A^\top \nabla^2 H(\theta^*) A(\bar{y}^{(i)}-\tilde{\theta}^*)  \right) 
\end{equation} 
The first term in \eqref{eq:PBC_conclusion_2} can be bounded by noting that $\left(A^\top \nabla H(\theta^*)\right)^\top  (y^{(i)}-\bar{y}^{(i)})=-\sum_{k=1}^m \lambda_k y^{(i)}_k$, and
\begin{equation}
\label{eq:bd_1}
    \left( A^\top \nabla H(\theta^*)\right)^\top  (y^{(i)}-\bar{y}^{(i)}) \le \left( \max_{k= 1, \dots, m } \lambda_k \right) \cdot |y^{(i)}-\bar{y}^{(i)}|_1
\end{equation}
For the second term, note that $A\left( \bar{y}^{(i)}-\tilde{\theta}^*\right)\in \mathcal{C}(\theta^*)$ and let $z\in \mathbb{R}^{K-1-m}$ be a vector such that $A\left( \bar{y}^{(i)}-\tilde{\theta}^*\right)=Uz$ where $U$ is a basis of $\mathcal C(\theta^*)$ defined in \eqref{eq:U_basis}. Using that $U$ has orthonormal columns, $||z||=||Uz||=||A(\bar{y}^{(i)}-\tilde{\theta}^*)||$, 
\begin{align*}
    \frac{1}{2} (\bar{y}^{(i)}-\tilde{\theta}^*)^\top A^\top \nabla^2 H(\theta^*) A(\bar{y}^{(i)}-\tilde{\theta}^*) 
    &= \frac{1}{2} z^\top U^\top \nabla^2 H(\theta^*) Uz\\
    &\le \frac{1}{2} \eta_{\max} 
    \left(  U^\top \nabla^2 H(\theta^*) U \right) ||z||^2 \\
    &= \frac{1}{2} \eta_{\max} \left( U^\top \nabla^2 H(\theta^*) U \right) ||A(\bar{y}^{(i)}-\tilde{\theta}^*)||^2
\end{align*}
where $\eta_{\max}$ is the largest eigenvalue of $U^\top \nabla^2H(\theta^*)U$. Since $\eta_{\max}(\left( U^\top \nabla^2 H(\theta^*) U \right))<0$ by (A4) and  $||A(\bar{y}^{(i)}-\tilde{\theta}^*)||^2\ge ||\bar{y}^{(i)}-\tilde{\theta}^*||^2$,
\begin{equation}
\label{eq:bd_2}
        \frac{1}{2} (\bar{y}^{(i)}-\tilde{\theta}^*)^\top A^\top \nabla^2 H(\theta^*) A(\bar{y}^{(i)}-\tilde{\theta}^*) 
  \le  \frac{1}{2} \eta_{\max} \left( U^\top \nabla^2 H(\theta^*) U \right) \cdot ||\bar{y}^{(i)}-\tilde{\theta}^*||^2.
\end{equation}
Plugging \eqref{eq:bd_1} and \eqref{eq:bd_2} into \eqref{eq:PBC_conclusion_2}, and letting $c=\min\left\{
-\max_{j}\lambda_j,\;
-\frac{1}{2}\,\eta_{\max}\!\left(U^\top \nabla^2 H(\theta^*) U\right)
\right\}>0$, we have the bound
\begin{equation}   
\label{eq:contradiction}
        H(T(y^{(i)})) \le - \frac{c}{2} \left(||P_m ({y}^{(i)}-\tilde{\theta}^*)||_1 + ||(I_{K-1}-P_m)({y}^{(i)}-\tilde{\theta}^*)||_2^2  \right) 
\end{equation} 
which contradicts \eqref{eq:tocontradict} if $i$ is large enough such that $\frac{1}{i}<\frac{c}{2}$.
\end{proof}

\subsection{Derivative of $\KL\!\left(T(\ttheta)\,\middle|\,T(h_\beta(u)+\ttheta)\right)$ }
\label{sec:KL_derivative}

\begin{lem}

\begin{equation}
\label{eq:derivative_of_KL_with_Beta_R_beta}
    \frac{d}{d\beta}\,\KL\!\left(T(\ttheta) | T(h_\beta(u)+\ttheta)\right) = \beta^{-2} B_{\beta}(u)+R_\beta(u)
\end{equation}
where $B_\beta(u)$ is $O(\beta^{-1})$ and $R_\beta (u)$ is $O(\beta^{-3})$. 
\end{lem}

\begin{proof}
Suppose $m<K-1.$ 
The KL-divergence composed with $T$ takes the form 
\begin{align*}
\text{KL}(T(\ttheta)|T(y)) & =\sum_{k=m+1}^{K-1}\theta_{k}^{*}\log\frac{\tilde{\theta}_{k}^{*}}{y_{k}}+(1-\mathbf{1}^{\top}\tilde{\theta}^{*})\log\frac{(1-\mathbf{1}^{\top}\tilde{\theta}^{*})}{(1-\mathbf{1}^{\top}y)}. 
 \end{align*}
Letting $y=h_\beta(u)+\ttheta$, we have that 
\begin{align}
\frac{d}{d\beta}\,\KL\!\left(T(\ttheta) | T(h_\beta(u)+\ttheta)\right)
&=\beta^{-2} \left( \sum_{k=m+1}^{K-1}
\frac{\ttheta_k\,u_k}{\beta^{-1}u_k+\theta_k^*} 
-\sum_{k=m+1}^{K-1}
\frac{(1-\mathbf 1^\top\ttheta)\,u_k}{D_\beta(u)} \right) \nonumber \\
&-\; 2\beta^{-3}\sum_{j=1}^{m}
\frac{(1-\mathbf 1^\top\ttheta)\,u_j}{D_\beta(u)}
\label{eq:derivative_KL_split_1}
\end{align}
where \[
D_\beta(u)=
1-\sum_{i=1}^{m}\beta^{-2}u_i-\sum_{r=m+1}^{K-1}\bigl(\beta^{-1}u_r+\tilde{\theta}_r^*\bigr)=\theta_K^*-\beta^{-2}\sum_{i=1}^{m}u_i-\beta^{-1}\sum_{r=m+1}^{K-1}u_r.
\]
Define the terms in \eqref{eq:derivative_KL_split_1} as
\[
B_\beta(u)
=
\sum_{k=m+1}^{K-1}
\frac{\tilde{\theta}_k^*u_k}{\beta^{-1}u_k+\tilde{\theta}_k^*}
-
\sum_{k=m+1}^{K-1}
\frac{(1-\mathbf {1}^\top\tilde{\theta}^*)u_k}{D_\beta(u)},\qquad R_\beta (u) =2\beta^{-3}\sum_{j=1}^{m}
\frac{(1-\mathbf 1^\top\ttheta)\,u_j}{D_\beta(u)}
\]
which gives the expression \eqref{eq:derivative_of_KL_with_Beta_R_beta}. 
When $m=K-1$, the first term in \eqref{eq:derivative_KL_split_1} disappears and gives $B_{\beta}(u)=0$ in \eqref{eq:derivative_of_KL_with_Beta_R_beta}.

\paragraph{1. $B_\beta(u)=O(\beta^{-1})$.}

Fix $u$. Let
\[
S_1:=\sum_{r=m+1}^{K-1}u_r,\qquad S_2:=\sum_{i=1}^{m}u_i,\qquad a:=\theta_K^*>0.
\]
Then $D_\beta(u)=a-\beta^{-1}S_1-\beta^{-2}S_2$, so
\[
|D_\beta(u)-a|
\le \beta^{-1}|S_1|+\beta^{-2}|S_2|.
\]
Hence there exists $\beta_0=\beta_0(u,\theta^*)$ such that for all $\beta\ge\beta_0$,
\begin{equation}\label{eq:D_lower_bound}
D_\beta(u)\ge \frac{a}{2}.
\end{equation}
Next, we analyze $B_\beta(u)$ by adding and subtracting $\sum_{k=m+1}^{K-1}u_k$:
\begin{equation}
\label{eq:B_beta_partialsum}
    B_\beta(u)
=
\sum_{k=m+1}^{K-1}\left(\frac{\ttheta_ku_k}{\beta^{-1}u_k+\ttheta_k}-u_k\right)
-
\sum_{k=m+1}^{K-1}\left(\frac{a u_k}{D_\beta(u)}-u_k\right).
\end{equation}
For the first summand  in $B_\beta (u)$: for each $k\in\{m+1,\dots,K-1\}$, write
\[
\frac{\ttheta_k u_k}{\beta^{-1}u_k+\ttheta_k}-u_k
=
u_k\left(\frac{\ttheta_k}{\ttheta_k+\beta^{-1}u_k}-1\right)
=
-\frac{\beta^{-1}u_k^2}{\ttheta_k+\beta^{-1}u_k}.
\]
Since $\ttheta_k>0$ is fixed, there exists $\beta_{k,0}=\beta_{k,0}(u,\ttheta)$ such that
$\ttheta_k+\beta^{-1}u_k \ge \ttheta_k/2$ for all $\beta\ge\beta_{k,0}$. Therefore for all
$\beta\ge \max_k \beta_{k,0}$,
\begin{equation}\label{eq:term1_bound}
\left|\frac{\ttheta_ku_k}{\beta^{-1}u_k+\ttheta_k}-u_k\right|
\le
\frac{2}{\theta_k^*}\,|u_k|^2\,\beta^{-1}.
\end{equation}
Similarly, for the second sum
\[
\frac{a u_k}{D_\beta(u)}-u_k
=
u_k\left(\frac{a}{D_\beta(u)}-1\right)
=
u_k\,\frac{a-D_\beta(u)}{D_\beta(u)}
=
u_k\,\frac{\beta^{-1}S_1+\beta^{-2}S_2}{D_\beta(u)}.
\]
Using \eqref{eq:D_lower_bound}, for all $\beta\ge\beta_0$,
\begin{equation}\label{eq:term2_bound}
\left|\frac{a u_k}{D_\beta(u)}-u_k\right|
\le
|u_k|\cdot \frac{\beta^{-1}|S_1|+\beta^{-2}|S_2|}{a/2}
\le
\frac{2}{a}\,|u_k|\,\bigl(|S_1|+|S_2|\bigr)\,\beta^{-1}.
\end{equation}
Applying the triangle inequality together in \eqref{eq:B_beta_partialsum} with \eqref{eq:term1_bound}--\eqref{eq:term2_bound} gives, for all
$\beta$ large enough,
\[
|B_\beta(u)|
\le
\beta^{-1}\sum_{k=m+1}^{K-1}\frac{2}{\ttheta_k}|u_k|^2
+
\beta^{-1}\sum_{k=m+1}^{K-1}\frac{2}{a}|u_k|\bigl(|S_1|+|S_2|\bigr)
=
C(u,\ttheta)\,\beta^{-1},
\]
where $C(u,\ttheta)<\infty$ depends only on $u$ and $\ttheta$. Hence $B_\beta(u)=O(\beta^{-1})$.

\paragraph{2. $R_\beta (u)=O(\beta^{-3})$.}

Since \(u\) is fixed and \(D_\beta(u)\to\theta_K^*>0\) as \(\beta\to\infty\), there exists \(\beta_0>0\) such that for all \(\beta\ge\beta_0\),
\[
|D_\beta(u)|\ge \frac{\theta_K^*}{2}.
\]
Therefore, for \(\beta\ge\beta_0\),
\[
|R_\beta(u)|
\le
2\beta^{-3}\sum_{j=1}^{m}
\frac{|1-\mathbf 1^\top\ttheta|\,|u_j|}{|D_\beta(u)|}
\le
\frac{4|1-\mathbf 1^\top\ttheta|}{\ttheta_K}
\left(\sum_{j=1}^{m}|u_j|\right)\beta^{-3}=C(u,\ttheta)\,\beta^{-3}.
\]
where $C(u,\ttheta)<\infty$ depends only on $u$ and $\ttheta$.

\end{proof}

\subsection{Discussion on \cite{cover1984algorithm} algorithm}
\label{sec:cover_algorithm}

\cite{cover1984algorithm} studies the log-optimal (Kelly) portfolio problem. Suppose we observe a random non-negative return vector $X=(X_1,\dots,X_m)$ with a known distribution $F$. A portfolio $b$ is a probability vector on the simplex $\{ b\ge 0, \sum_i b_i=1\}$.
\begin{equation}
\label{eq:optimum_portfolio}
W(b)=\mathbb{E}[\log (b^\top X)], \qquad W^*=\max_b W(b)
\end{equation}
 The objective is to find an optimizer $b^*$. The paper proposes the following algorithm. First define the component-wise gradient
 \begin{equation}
     a_i(b)=\mathbb{E}\left[ \frac{X_i}{b^\top X} \right]
 \end{equation}
so that $a(b)=\nabla W(b)$. The (multiplicative) update rule is 
$$ b_i^{n+1} = b_i^{(n)} a_i (b^{(n)}),\quad \text{ for } i=1,\dots,m,$$
provided that the initial $b^{(0)}$ has only strictly positive components. The algorithm converges to the optimum in value
\begin{equation}
    W(b^n)\uparrow W^*.
\end{equation}
If the support of $X$ has full dimension, then $b^*$ is unique and $b^{(n)}\to b^*$. 

The problem of maximizing the LDA log-likelihood $H$ in \eqref{eq:entropy} can be viewed as a special case of the log-optimal portfolio problem
\eqref{eq:optimum_portfolio}. Let $X\in\mathbb{R}^K$ be a random vector taking values
$(\phi_1(v),\dots,\phi_K(v))$ with probabilities $p_v$, and identify
the portfolio vector $b\in\Delta_{K-1}$ with $\theta$.
Then
$$
\mathbb{E}[\log(\theta^\top X)]
=
\sum_{v=1}^{V} p_v 
\log\!\left(\sum_{k=1}^{K}\theta_k \phi_k(v)\right)
= H(\theta).
$$ 
Under this identification, the distribution $F$ of $X$ is induced by the empirical word distribution 
$\{p_v\}_{v=1}^V$.

Moreover, the support of $X$ has full dimension if and only if the topic
matrix $\phi\in\mathbb{R}^{K\times V}$ has full row rank. In this case
the maximizer $\theta^*$ is unique.

\subsection{Auxiliary Lemmas in the proof of Theorem \ref{thm:epsilon_main_result} }
\label{sec:aux_lemmas_lDA_orthogonal}

\subsubsection{Proof of Lemma \ref{lem:gtbound_lemma}}
\begin{proof}[Proof of Lemma \ref{lem:gtbound_lemma}] 

Let \(x=\frac{1+t}{2\sqrt t}\). Using the well-known inequality \(\log x\le x-1\) for \(x>0\),
$$g(t)=\log\frac{1+t}{2\sqrt t}
\le
\frac{1+t}{2\sqrt t}-1
=
\frac{(\sqrt t-1)^2}{2\sqrt t}.$$
Moreover,
$$f(t)=t+\frac1t-2
=
\frac{(\sqrt t-1)^2(\sqrt t+1)^2}{t}.$$
Hence
$$\frac{g(t)}{f(t)}
\le
\frac{\frac{(\sqrt t-1)^2}{2\sqrt t}}
{\frac{(\sqrt t-1)^2(\sqrt t+1)^2}{t}}
=
\frac{\sqrt t}{2(\sqrt t+1)^2}
\le \frac18,$$
since \(0\le(\sqrt t-1)^2\). Thus \(g(t)\le \tfrac18 f(t)\).
\end{proof}

\subsubsection{Proof of Lemma \ref{lem:Frobenius_bound}}

\begin{proof}[Proof of Lemma \ref{lem:Frobenius_bound}] 

Let $$ Z = U \operatorname{diag}(\lambda_1,\dots,\lambda_K) U^\top,$$
be a spectral decomposition of $Z$, where $U$ is orthogonal and $\lambda_1,\dots,\lambda_K>0$ are positive eigenvalues of $Z$. Then
$$Z-I=UDU^\top,\qquad D:=\operatorname{diag}(\lambda_1-1,\dots,\lambda_K-1)$$
By unitary invariance of the Frobenius norm (cf.~\cite{horn2012matrix} section 5.6)
\begin{equation}
    \|Z-I\|_F^2=\|UDU^\top\|_F^2 = \|D\|^2_F. 
\end{equation}
Since $D$ is diagonal,  
\begin{equation}
    ||D||^2_F = \sum_{i=1}^K (\lambda_i - 1)^2. 
\end{equation}
Therefore,
\begin{align*}
    \sum_{i=1}^K f(\lambda_i) &=\sum_{i=1}^K \frac{(\lambda_i -1)^2}{\lambda _i}\\
&\le \frac{1} {\min \lambda_i} \sum_{i=1}^K (\lambda_i - 1)^2 \\
&=\frac{\|Z-I\|_F^2}{\lambda_{\min}(Z)}.
\end{align*}
\end{proof}

\subsubsection{Proof of Lemma \ref{lem:Z-Ibound_calculation}}

\begin{proof}[Proof of Lemma \ref{lem:Z-Ibound_calculation}] 

The outline of the proof is as follows:
\begin{enumerate}[itemsep=0pt, topsep=0pt, parsep=0pt, partopsep=0pt]
\item We bound the off-diagonal terms of $Z-I$.
\item We bound the diagonal terms of $Z-I$.
\item We combine previous two steps to find an upper bound of $\| Z-I \|_F^2$.
\item Using Weyl's perturbation lemma, we find a lower bound of $\lambda_{\min}$ in terms of $\| Z-I \|_F^2$.
\item We find an upper bound on $||Z-I||_{F}\,\frac{1}{\sqrt{\lambda_{\min}(Z)}}$ in the form of $C\varepsilon$ where $C>0$.
\end{enumerate}
We first note the explicit expression for $Z$
\begin{equation}
Z_{ij}
=
\sqrt{\theta_i^* \theta_j^*}
\sum_{v}
p_v
\frac{\phi_i(v)\phi_j(v)}{s_v^{2}},\quad s_v=\theta^{*\top}\phi(v).
\end{equation}
For all $v\in V$, note that 
 \begin{equation}
\label{eq:s_v_bound}
s_v \ge \theta^*_{k(v)} \phi_{k(v)}(v).
\end{equation}
We will use this to bound the magnitude of the entries in $Z-I.$

\paragraph{1. Bounding the off-diagonal terms of $Z-I$.}

With assumption (B1), we can define a partition of the vocabulary set, defined as
$$ S_i :=\{v:k(v)=i\},\quad i =1,\dots,K.$$
Each set $S_i$ collects words that are most concentrated on a topic $i$. From the definition of sparsity in \eqref{eq:sparsity}, note that $$\phi_j(v) \le \varepsilon \phi_{k(v)}(v),\quad \text{if } j \neq k(v).$$
With this, we can bound the off-diagonal terms $Z_{ij}$ (where $i\neq j$)
\begin{align}
\label{eq:Z_ij_abs_bound}
|Z_{ij}|
&\le
\sqrt{\theta_i^{*}\theta_j^{*}}
\Biggl[
  \underbrace{\sum_{v\in S_i}
      p_v\,\frac{\phi_i(v)\,(\varepsilon\phi_i(v))}{s_v^2}}_{(1)}
  \;+\;
  \underbrace{\sum_{v\in S_j}
      p_v\,\frac{(\varepsilon\phi_j(v))\,\phi_j(v)}{s_v^2}}_{(2)}
  \;+\;
  \underbrace{\sum_{v\in (S_i \cup S_j)^c}
      p_v\,\frac{(\varepsilon\phi_{k(v)}(v))^2}{s_v^2}}_{(3)}
\Biggr].
\end{align}
Using \eqref{eq:s_v_bound} and the definition of $C_{\max}^{(1)}$ and $C_{\max}^{(2)}$ in \eqref{eq:constants_max}, we can bound each term:
\[
\begin{aligned}
(1) \quad v\in S_i:&\qquad
\sqrt{\theta_i^{*}\theta_j^{*}}\,
\frac{\varepsilon\,\phi_i(v)^2}{s_v^2}
\;\le\; \frac{\varepsilon \sqrt{\theta_j^{*}/\theta_i^{*}}}{\theta_i^{*}} \le \varepsilon  C_{\max}^{(1)}\\
(2) \quad v\in S_j:&\qquad
\sqrt{\theta_i^{*}\theta_j^{*}}\,
\frac{\varepsilon\,\phi_j(v)^2}{s_v^2}
\;\le\;\frac{\varepsilon \sqrt{\theta_i^{*}/\theta_j^{*}}}{\theta_j^{*}} \le \varepsilon C_{\max}^{(1)} \\
(3) \quad v\in (S_i \cup S_j)^c :&\qquad
\sqrt{\theta_i^{*}\theta_j^{*}}\,
\frac{\varepsilon^{2}\,\phi_{k(v)}(v)^2}{s_v^2}
\;\le\;
\varepsilon^2
\frac{\sqrt{\theta_i^{*}\theta_j^{*}}}{\theta_{k(v)}^{*\,2}} \le \varepsilon^2  C_{\max}^{(2)}.
\end{aligned}
\]
Plugging these bounds in \eqref{eq:Z_ij_abs_bound}, and using that 
$\sum_{v} p_v =1$   
\begin{align}
\label{eq:Z_ij_bound}
|Z_{ij}|
&\le\;2C_{\max}^{(1)}\varepsilon 
+C_{\max}^{(2)}\varepsilon^2, \quad i \neq j.
\end{align}

\paragraph{2. Bounding the diagonal terms of $Z-I$.}

We recall from \eqref{eq:LDA_gradient}, $\nabla H(\theta^*)=\mathbf{1}_K$. This implies that
\begin{equation}
        \nabla H(\theta^*)_i =\sum_v p_v \frac{\phi_i (v)}{s_v}=1.
\end{equation}
Therefore,
\begin{align}
     \label{eq:Z_ii-1} 
Z_{ii}-1	&\nonumber =\sum_{v}\theta_{i}^{*}p_{v}\frac{\phi_{i}(v)^{2}}{s_{v}^{2}}-\sum_{v}p_{v}\frac{\phi_{i}(v)}{s_{v}}\\
	 &=\sum_{v}p_{v}\frac{\phi_{i}(v)}{s_{v}}\left(\theta_{i}^{*}\frac{\phi_{i}(v)}{s_{v}}-1\right).
\end{align}
For $v\in S_i$, we have
$$
s_v
= \theta_i^* \phi_i(v) + \sum_{k\ne i}\theta_k^* \phi_k(v)
\le (1+\varepsilon \rho_i)\theta_i^* \phi_i(v),
\qquad
\rho_i := \max_{k}\frac{\theta_k^*}{\theta_i^*}.
$$
Together with \eqref{eq:s_v_bound}, for \(v\in S_i\), we have
$$
\theta_i^*\phi_i(v)
\le s_v
\le (1+\varepsilon\rho_i)\theta_i^*\phi_i(v).
$$
Dividing by \(\theta_i^*\phi_i(v)\) and taking reciprocals gives
$$
\frac{1}{1+\varepsilon\rho_i}
\le
\theta_i^*\frac{\phi_i(v)}{s_v}
\le 1.
$$
Hence
\begin{equation}
\label{eq:Z_ii_bound_1}
0 \ge \theta_i^*\frac{\phi_i(v)}{s_v}-1
\ge -\frac{\varepsilon\rho_i}{1+\varepsilon\rho_i}
\ge -\varepsilon\rho_i.    
\end{equation}
Moreover, using again \eqref{eq:s_v_bound}, we obtain
\begin{equation}
\label{eq:Z_ii_bound_2}
    \frac{\phi_i(v)}{s_v}\le \frac{1}{\theta_i^*}.
\end{equation}
Combining \eqref{eq:Z_ii_bound_1} and \eqref{eq:Z_ii_bound_2}, for $ v\in S_i$ we have
\begin{equation}
\label{eq:Z_ii_bound_3}
    \left|
\frac{\phi_i(v)}{s_v}
\left(\theta_i^* \frac{\phi_i(v)}{s_v}-1\right)
\right|
\le
\frac{\varepsilon \rho_i}{\theta_i^*}.
\end{equation}
For $v\notin S_i$, since \(\phi_i(v)\le \varepsilon \phi_{k(v)}(v)\) and \(k(v)\neq i\),
$$\frac{\phi_i(v)}{s_v}
\le \frac{\varepsilon \phi_{k(v)}(v)}{\theta_{k(v)}^*\phi_{k(v)}(v)}\le
\frac{\varepsilon}{\theta_{\min}^*},$$
while
$$
\left|\theta_i^*\frac{\phi_i(v)}{s_v}-1\right|\le 1.
$$
Therefore for $v\notin S_i,$
\begin{equation}
\label{eq:Z_ii_bound_4}
    \left|
\frac{\phi_i(v)}{s_v}
\left(\theta_i^* \frac{\phi_i(v)}{s_v}-1\right)
\right| \le \frac{\varepsilon}{\theta_{\min}^*}.
\end{equation}
Combining \eqref{eq:Z_ii_bound_3} and \eqref{eq:Z_ii_bound_4}  gives
\begin{align}
\label{eq:Z_ii_bound}
|Z_{ii}-1|
&\le
\sum_{v\in S_i} p_v \frac{\varepsilon \rho_i}{\theta_i^*}
+\sum_{v\notin S_i} p_v \frac{\varepsilon}{\theta_{\min}^*} \nonumber\\
&\le\varepsilon\left(\frac{\rho_i}{\theta_i^*}
+\frac{1}{\theta_{\min}^*}\right)
\le 2\,\frac{\theta_{\max}^*}{(\theta_{\min}^*)^2}\,\varepsilon
=2 C_{\max}^{(2)}\varepsilon.
\end{align}
\end{proof}

\paragraph{3. Upper bound of $\| Z-I \|_F^2$.}
Combining the bounds \eqref{eq:Z_ij_bound} and \eqref{eq:Z_ii_bound}, we have that
\begin{align}
\label{eq:Z_I_bound}
||Z-I||_F^2 &= \sum_i (Z_{ii}-1)^2+\sum_{i\neq j}Z_{ij}^2 \nonumber \\
&\le 4\varepsilon^2 (C_{\max}^{(2)})^2 K +K(K-1)(2C_{\max}^{(1)}\varepsilon 
+C_{\max}^{(2)}\varepsilon^2)^2 \nonumber \\
&= \varepsilon^2 {F(\varepsilon)} 
\end{align}
where 
\begin{equation}
\label{eq:F_epsilon}
F(\varepsilon)=4 (C_{\max}^{(2)})^2 K+K(K-1)\left(2C_{\max}^{(1)}+C_{\max}^{(2)}\varepsilon\right)^{2}>0.  
\end{equation} 
\paragraph{4. Lower bound of $\lambda_{\min}$.}

From Weyl's perturbation theorem (cf. \cite{horn2012matrix} Corollary 4.3.15), we have that  
$$\lambda_{\min} (Z) \ge \lambda_{\min}(I)+\lambda_{\min} (Z-I).$$
Since $Z-I$ is symmetric, 
$$\lambda_{\min}(Z-I) \ge - \| Z-I \|_2.$$
Using that the spectral norm is bounded by the Frobenius norm ($\|\cdot \|_2 \le \| \cdot \|_F $),
we obtain
\begin{align}
\label{Z:eigenvalue_bound}
\lambda_{\min}(Z)	
	&\ge\lambda_{\min}(I)-||Z-I||_{F}  \nonumber \\
	&\ge1- \varepsilon\sqrt{F(\varepsilon)} .
\end{align}

\paragraph{5. Upper bound on $||Z-I||_{F}\,\frac{1}{\sqrt{\lambda_{\min}(Z)}}$.}

Putting \eqref{eq:Z_I_bound} and \eqref{Z:eigenvalue_bound} together, we have that
\begin{align}
||Z-I||_{F}\,\frac{1}{\sqrt{\lambda_{\min}(Z)}}
&\le \frac{ \varepsilon\sqrt{F(\varepsilon)}}{\sqrt{1 -  \varepsilon\sqrt{F(\varepsilon)} }}. \label{eq:finalbound}
\end{align}
Now we note that $H(\varepsilon):=\frac{\sqrt{F(\varepsilon)}}{\sqrt{1-\sqrt{F(\varepsilon)}\varepsilon}}$ is monotonically increasing on $[0,\bar{\varepsilon})$ where $\bar{\varepsilon}$ is the unique root of $\varepsilon \sqrt{F(\varepsilon)}=1$. To see that $H'(\varepsilon)>0$, write $s(\varepsilon):=\sqrt{F(\varepsilon)}$. Then
$$H(\varepsilon)=s(\varepsilon)\bigl(1-\varepsilon s(\varepsilon)\bigr)^{-1/2}.$$
Differentiating gives
$$H'(\varepsilon)
=\frac{
s'(\varepsilon)\Bigl(1-\tfrac12\varepsilon s(\varepsilon)\Bigr)
+\tfrac12 s(\varepsilon)^2
}{
\bigl(1-\varepsilon s(\varepsilon)\bigr)^{3/2}
}.$$
Since $F(\varepsilon)>0$ and $F'(\varepsilon)\ge0$, we have $s(\varepsilon)>0$ and$$
s'(\varepsilon)=\frac{F'(\varepsilon)}{2\sqrt{F(\varepsilon)}}\ge0.$$
Moreover, $\varepsilon<\varepsilon_0$ implies $1-\varepsilon s(\varepsilon)>0$, and hence
$1-\tfrac12\varepsilon s(\varepsilon)>0$. Therefore both the numerator and
denominator are positive, and thus
$$H'(\varepsilon)>0.$$
Therefore, $H(\varepsilon)<H(\varepsilon_0)$ for all $\varepsilon<\varepsilon_0$. 
Setting $C = H(\varepsilon_0)$ and applying~\eqref{eq:finalbound}, we obtain
\begin{align}
||Z-I||_{F}\,\frac{1}{\sqrt{\lambda_{\min}(Z)}}
&\le C\varepsilon.
\end{align}

\subsection{Simulation}

\subsubsection{Boundary Instance Generation}
\label{sec:simulation_boundary_cases}

Let $(K,m)$ be fixed.
For the synthetic experiments discussed in Section~\ref{sec:synthetic_experiments}, each problem instance $(\phi, \theta^*, p)$ is constructed so that
  $\theta^\ast$ lies on the boundary of the simplex with exactly $m$
  zero coordinates and is simultaneously the unique maximizer of $H$. We also require the strict complementarity condition (A3).

We generate problem instances in which the KKT multipliers $\lambda_i$ for active components are strictly bounded away from zero. The idea is that for every boundary point $\theta^*$ sampled, we look for a $p=(p_v)$ that $\theta^*$ satisfies KKT with $\nabla H(\theta^*)=1-\lambda_i$ where $\lambda_i \ge \lambda_{\min}$ for some $\lambda_{\min}>0$. The KKT condition is an equivalent condition for unique optimality of $\theta^*$ since $H$ is strictly concave.

We rewrite the gradient of $H$ as
\begin{equation}
\label{eq:boundary_algo_grad}
\nabla H(\theta^*)
=\sum_{v=1}^V w_v\, \phi(v),
\qquad w_v := \frac{p_v}{\theta^{*\top}\phi(v)}.
\end{equation}
For every candidate $b$ of the gradient $\nabla H(\theta^*)$ such that 
$$b\in
\mathcal{K}(\theta^*):=
\left\{
b \in
\mathbb{R}^K :
\begin{cases}
b_k = 1, & k \in \mathrm{supp}(\theta^*)\\
b_k < 1, & k \notin \mathrm{supp}(\theta^*)
\end{cases}
\right\},$$
we look for $w=(w_v)\ge 0$ such that $b=\phi w$. Once we recover such a $w$, we invert the relationship in \eqref{eq:boundary_algo_grad} to recover $p=(p_v)$.
The algorithm works as follows.

\textbf{Step 1. Generate topic--word distributions.}
 Draw independent samples of $\phi_k \sim \text{Dir}_{\beta}$ $K$ times where $\beta = 0.1 \cdot \mathbf{1}_V$ to generate $\phi$. 

  \textbf{Step 2. Generate topic proportions.}
Fix the active index set $\mathcal{A} = \{1,\ldots,m\}$ and its complement $\mathcal{I} = \{m+1,\ldots,K\}$. Set
  $\theta^\ast_k = 0$ for $k \in \mathcal{A}$.
  The inactive coordinates are drawn from
  $\theta^\ast_{\mathcal{I}} \sim \mathrm{Dir}(\mathbf{1}_{K-m})$.

    \textbf{Step 3. Generate $b$ with strict complementarity.}
  For each active topic $k \in \mathcal{A}$, draw a target dual variable
  $\lambda_k \sim \mathrm{Uniform}(\lambda_{\min}, \lambda_{\max})$
  with $\lambda_{\min} >0 $ and $\lambda_{\max} \le 1 $. $\lambda_{\max} \le 1$ is required since the gradient of $H$ is always non-negative.
  Define the target gradient vector
  $b_k = 1$ for $k \in \mathcal{I}$ and $b_k = 1 - \lambda_k$ for $k \in \mathcal{A}$.

\textbf{Step 4. Recover $w$ and $p$.}
  With $b$ constructed in Step 3, we seek $w = (w_v) \geq 0$ satisfying $\phi w = b$ and the normalization constraint $s^\top w = 1$,      where $s_v = \theta^{\ast\top}\phi(v)$.       Stacking these into the augmented system       $A = [\phi;\, s^\top] \in \mathbb{R}^{(K+1)\times V}$ and               $b^+ = [b;\, 1]$,  we solve $Aw = b^+$ via least-squares projection onto the feasible set, initialized at $w_0 = \mathbf{1}_V/V$. 
  (Since $V$ is much larger than $K$, the system $Aw = b^+$ is heavily under-determined,
  so many solutions exist. The minimum-norm solution relative to $w_0$ is unique since it corresponds to an orthogonal projection onto an affine subspace.) If the minimum-norm solution does not satisfy $w\ge0$, steps $1-4$ are repeated until a non-negative $w$ is obtained.
  The word-probability vector is then recovered by using $p_v = s_v w_v / \sum_{v'} s_{v'} w_{v'}$. By construction, $(\phi, p, \theta^\ast)$ satisfies the KKT conditions for $\theta^\ast$ to be the maximizer of $H$ with strict complementarity margin $\lambda_k \in [\lambda_{\min}, \lambda_{\max}]$.

\begin{remark}
    Variance reduction is still observed in settings where strict complementarity fails; however, such cases fall outside the scope of the analysis presented in the paper.
\end{remark}

\subsubsection{Estimation of Plain Monte Carlo Performance}
\label{sec:degeneracy_plainMC}

A natural baseline for estimating the variance ratio $1-\rho^2$ and the
moments $\mathbb{E}[e^{nH(\theta)}]$ is plain Monte Carlo under
the prior $\mathrm{Dir}_{\alpha}$. However, 
the quantities of interest are very difficult to estimate precisely for large $n$ with plain MC.
The reason is that for large document lengths $n$,
the integrand $e^{nH(\theta)}$ concentrates sharply near the maximizer $\theta^*$,
while the prior $\mathrm{Dir}_{\alpha}$ fails to sample enough around $\theta^*$.  

As a concrete example, consider the negative KL divergence 
$$H(\theta) = -\KL(\theta^*|\theta) = \sum_k \theta^*_k \log \theta_k - \sum_k \theta^*_k \log \theta_k^*$$
where $\mathbb{E}[e^{nH(\theta)}]$ has a closed-form expression:
\begin{equation}
\label{eq:control_variate_expectation}
    \mathbb{E}_{\text{Dir}_\alpha}[e^{nH(\theta)}]=\frac{B(\alpha+n\theta^*)}{B(\alpha)}\cdot e^{-n\theta^*\cdot \log\theta^*}.
\end{equation}
Figure~\ref{fig:mc_diagnostic} shows that across $100$ independent runs, the plain MC estimate deviates from the exact value beyond $n = 1000$,
  while the IS estimate remains accurate across all $n$, even though MC uses $10$ times more samples than IS. By $n = 15{,}000$ the plain MC estimate underestimates the
  closed-form value by a factor of $10^{7}$ on average, and this
  discrepancy grows with $n$, rendering
  the plain MC estimates unreliable.
  
  For experiments in Figure~\ref{fig:imp_sampling_epsilon_0p1_K_5}, at $n = 10{,}000$ with $\gamma = 0.9$, the plain MC estimate of
  $\mathbb{E}[e^{nH(\theta)}]$ (using $10^7$ samples) can be up to $10^{17}$ times
  smaller than the IS estimate (using $10^6$ samples) in the boundary case and $10^{12}$ in the interior case, indicating that the MC
  estimator severely underestimates the integral. We therefore apply importance sampling with $\gamma=0.9$ and $\epsilon=0.1$ to estimate the variance and related quantities under plain MC for synthetic experiments in Section~\ref{sec:synthetic_experiments} to verify the asymptotic rates.

\begin{figure}[t]
\centering
\includegraphics[width=0.8\linewidth]
{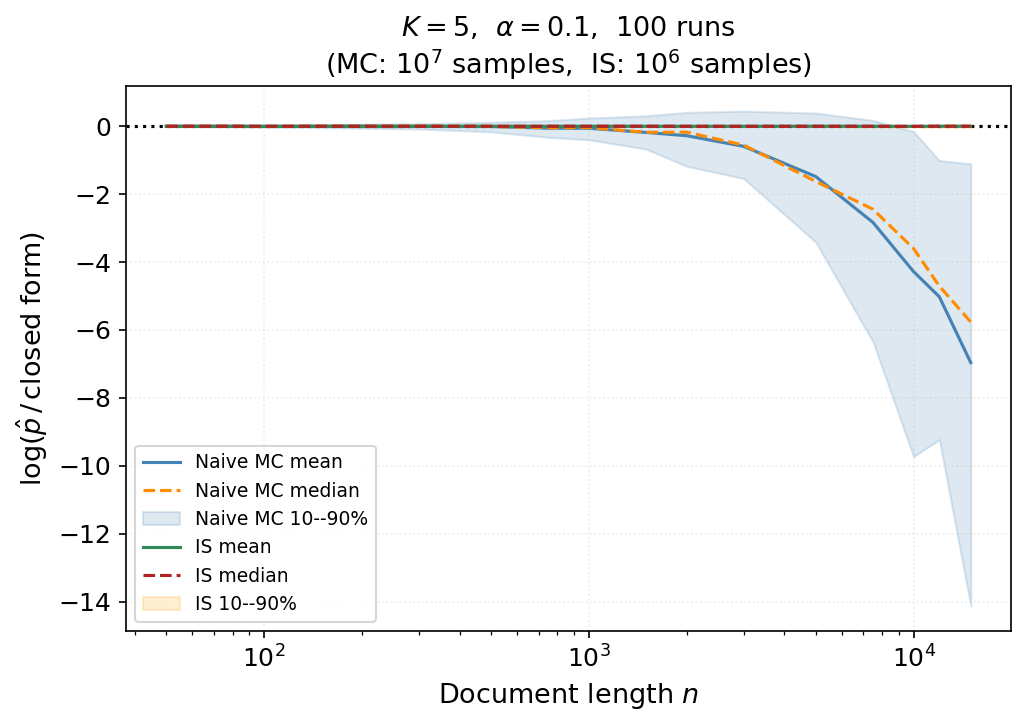}
\caption{\label{fig:mc_diagnostic}
Ratio of the estimated moment to the closed-form
expression~\eqref{eq:control_variate_expectation} on a log scale,
for $K = 5$, $\alpha = 0.1$, over 100 independent runs.
Naive MC ($10^7$ samples) increasingly underestimates the integral
beyond $n = 1000$, while IS ($10^6$ samples) remains accurate
across all $n$. Shaded regions show 10--90\% quantiles.}
\end{figure}

\end{document}